\documentclass[11pt,reqno]{amsart}

\usepackage{amsmath,amsfonts,amsthm,amssymb,amsxtra}
\usepackage{upref}
\usepackage{color}
\usepackage{tikz}
\usepackage[shortlabels]{enumitem}
\usepackage{bbm} 
\usepackage{stmaryrd}
\usepackage{nicematrix}
\usepackage{graphicx,subfigure}

\usepackage[margin=1.1in]{geometry}

\usepackage[colorlinks,citecolor=red,urlcolor=blue]{hyperref}

\newtheorem{theorem}{Theorem}[section]

\newtheorem{proposition}[theorem]{Proposition}
\newtheorem{lemma}[theorem]{Lemma}
\newtheorem{corollary}[theorem]{Corollary}

\theoremstyle{definition}

\newtheorem{definition}{Definition}
\newtheorem{defn}[definition]{Definition}
\newtheorem{example}{Example}
\newtheorem{ex}[example]{Example}
\newtheorem{assumption}{Assumption}
\newtheorem{question}{Question}
\theoremstyle{remark}

\newtheorem{remark}{Remark}[section]

\numberwithin{equation}{section}
\allowdisplaybreaks

\newcommand{\C}{\mathbb{C}}

\renewcommand{\epsilon}{\varepsilon}

\newcommand{\N}{\mathbb{N}}
\newcommand{\V}{\mathbb{V}}

\renewcommand{\phi}{\varphi}
\newcommand{\R}{\mathbb{R}}

\newcommand{\Z}{\mathbb{Z}}

\newcommand{\E}{\mathbb{E}}

\newcommand{\one}{\mathbbm{1}}

\newcommand{\wt}{\widetilde}

\DeclareFontFamily{U}{mathx}{}
\DeclareFontShape{U}{mathx}{m}{n}{<-> mathx10}{}
\DeclareSymbolFont{mathx}{U}{mathx}{m}{n}
\DeclareMathAccent{\widehat}{0}{mathx}{"70}
\DeclareMathAccent{\widecheck}{0}{mathx}{"71}
\renewcommand{\check}{\widecheck}

\newcommand{\intbr}[2]{\llbracket  #1 ,   #2 \rrbracket }

\newcommand{\ipc}[2]{ \langle #1 , #2  \rangle }

\newcommand{\braket}[3]{\left \langle #1 \, | #2 |\, #3 \right \rangle }

\renewcommand{\P}{\mathbb{P}}

\renewcommand{\emptyset}{\o}

\newcommand{\eps}{\varepsilon}

\newcommand{\vertiii}[1]{{\left\vert\kern-0.25ex\left\vert\kern-0.25ex\left\vert #1
    \right\vert\kern-0.25ex\right\vert\kern-0.25ex\right\vert}}

\newcommand\numberthis{\addtocounter{equation}{1}\tag{\theequation}}
\renewcommand{\emptyset}{\varnothing}

\makeatletter
\def\l@subsection{\@tocline{2}{0pt}{2.6pc}{5pc}{}}
\makeatother

\begin{document}

 \title[Landscape estimates of the IDS for Jacobi operators on graphs]{Landscape estimates of the integrated density of states for Jacobi operators on graphs}

\author[L. Shou, W. Wang, S. Zhang]{Laura Shou, Wei Wang, Shiwen Zhang}
  \date{}

\begin{abstract}
    We show the integrated density of states for a variety of Jacobi operators on graphs, such as the Anderson model and random hopping models on graphs with Gaussian heat kernel bounds, can be estimated from above and below in terms of the localization landscape counting function. Specific examples of these graphs include stacked and decorated lattices, graphs corresponding to band matrices, and aperiodic tiling graphs. 
The upper bound part of the landscape law  also applies to the fractal Sierpinski gasket graph.
As a consequence of the landscape law, we obtain 
landscape-based proofs of the Lifshitz tails in several models including random band matrix models, certain bond percolation Hamiltonians on $\Z^d$, and Jacobi operators on certain stacks of graphs.
We also present intriguing numerical simulations exploring the behavior of the landscape counting function across various models.
\end{abstract}

  \maketitle

\tableofcontents

\section{Introduction}

The structure and dynamics of quantum systems for electrons are generally described using a Hamiltonian
operator.  The simplified non-interacting Hamiltonian $H$ of a crystalline solid can often be modeled by a graph operator 
$H=-\Delta+V$, which is a weighted graph Laplacian perturbed by an on-site potential. The integrated density of states (IDS) of $H$ is one of the main quantities encoding the energy levels of the system. It describes the number of energy levels per unit volume below a certain energy level. 
In the present work,  we establish bounds on the IDS of $H$ via the
counting function of the so-called localization landscape \cite{filoche2012universal}. Such a \emph{Landscape Law} approach was first established in \cite{DFM} for the standard Schr\"odinger operator on $\R^d$, and later extended to 
$\Z^d$ in \cite{arnold2022landscape,desforges2021sharp}.

The goal of the present paper is to establish the landscape law for more general graph operators with both on-site potential and bond interactions. 
This extension is in two directions: First, by moving to graphs, we lose the regular structure of $\R^d$ and $\Z^d$ used in \cite{DFM,arnold2022landscape} to perform partition and translation arguments.
Second, by allowing for disordered bond interactions, we extend the class of operators to include random hopping and Jacobi operators, which may have degenerate off-diagonal hopping amplitudes.
The resulting landscape law will apply to a range of models, including Jacobi operators on graphs \emph{roughly isometric} to $\Z^d$ (or more generally, graphs with Gaussian heat kernel bounds), which includes stacked and decorated lattices, graphs corresponding to banded matrices, and aperiodic tiling graphs.
We also obtain a landscape law upper bound for the Sierpinski gasket graph, which is not roughly isometric to any $\Z^d$ and has non-Gaussian (specifically, sub-Gaussian) heat kernel bounds.

After establishing the landscape law, the general method described in \cite{DFM,arnold2022landscape} will allow us to recover Lifshitz tail estimates for several of the random hopping models on $\Z^d$, including random band matrices, bond percolation Hamiltonians with certain boundary conditions, and discrete acoustic operators $\operatorname{div}A\cdot\nabla$. 
This provides alternate proofs, using the landscape function, for Lifshitz tails in these models (cf. 
\cite{KirschMueller,MuellerStollman,MuellerStollman2,KloppNakamura,rband}). 
The landscape law method here also suggests an avenue to obtain Lifshitz tails for non-regular graphs, if one proves certain geometric properties of the graphs.

We now introduce some definitions in order to state our main
results.
We consider an (unweighted) graph $\Gamma=(\mathbb{V},{\mathcal E})$, where the vertex set $\mathbb{V}$ is countably infinite. We write $x\sim y$ to mean $\{x,y\}\in \mathcal E$, and in this case say that $y$ is a neighbor of $x$.   
 The natural graph metric, denoted $d_0(\cdot,\cdot)$, gives the length $n\in\N_0$ of the shortest path between two points. Define balls in $\Gamma$ with respect to $d_0$ as \[B^{d_0}(x,r)=\{y\in \V: d_0(x,y)\le r, \  x\in \V, r\ge 0 \}. \]
 We may omit the superscript dependence $d_0$ and write $B(x,r)=B^{d_0}(x,r)$, unless another metric is also being considered.  Since the metric $d_0$ is integer valued, we allow $r$ to be real and $B(x,r)=B(x,\lfloor r \rfloor)$, where $\lfloor \cdot  \rfloor$ is the floor function. 
 For any   subset $S\subset \V$, we denote by $|S|=\#\{x: x\in S\}$ the cardinality of $S$. 

As described below in Assumption~\ref{ass:0}, in this article we will always work with graphs with the following polynomial volume growth. This property is also called ``Ahlfors $\alpha$-regular'' in \cite{barlow-alphabeta}, where the exponent $\alpha$ is the graph analog of the Hausdorff dimension.
\begin{assumption}\label{ass:0}
We will always assume $\Gamma$ is connected and satisfies the following volume control property with parameter $\alpha$
(``Ahlfors $\alpha$-regular''): 
There are $\alpha\ge 1$, $c_1,c_2>0$ such that for any $x\in \V$ and $r\ge1$, 
        \begin{align}\label{eqn:vol-control}
          c_1r^{\alpha} \le  |B(x,r)| \le c_2 r^{\alpha}.
        \end{align}
Note this also implies bounded geometry: Letting $\deg (x)=\#\{y: x  \sim y\}$
be the degree of a vertex $x\in \V$, then
    \begin{align}\label{eqn:bded-geo}
        \sup_{x\in \V}{\deg(x)}:=M_\Gamma<\infty.
    \end{align}
\end{assumption}

We now consider the Jacobi operator 
$H$ on ${\mathcal H}:=\ell^2(\V)$ defined as
\begin{align}\label{eqn:H-gen}
    Hf(x)=\sum_{y:y\sim x}\big(f(x)-\mu_{xy}f(y)\big) +V_xf(x),\  \ x \in \V,
\end{align}
which has an on-site non-negative potential $\{V_x\ge 0\}_{x\in \V}$, and bond strengths  $\{0\le \mu_{xy}=\mu_{yx}\le 1\}_{x\sim y\in \V}$.

\begin{remark}
The placement of the bond strengths $\mu_{xy}$ in \eqref{eqn:H-gen} on only $f(y)$, and not as $\mu_{xy}(f(x)-f(y))$, may at first appear unusual.
However, the bond strength placement in \eqref{eqn:H-gen} corresponds to tight-binding hopping models commonly encountered in the physics literature, for which \eqref{eqn:H-gen} may be written in the form,
\[
H=-\sum_{\langle x,y\rangle}t_{xy}(|x\rangle\langle y|+|y\rangle\langle x|)+\sum_x V_x|x\rangle\langle x|+\sum_x\deg (x)|x\rangle\langle x|.
\]
The last diagonal term involving $\deg(x)$ is not usually present, but it is just a constant shift in energy for regular graphs like $\Z^d$. 

As another reason for considering the operator \eqref{eqn:H-gen}, we note that when there is no diagonal disorder ($V_x\equiv 0$), then \eqref{eqn:H-gen} becomes a random hopping model like in \cite{Dyson,ITA,TheodorouCohen}. For this random hopping model \eqref{eqn:H-gen} on $\Z^d$ with $V_x\equiv 0$, there are still the Lifshitz tails for the integrated density of states (cf. Section~\ref{subsec:bondperc}). However, for the different model $\wt Hf(x)=\sum_{y:y\sim x}\mu_{xy}(f(x)-f(y))$, the tail behavior of the  integrated density of states and the spectrum at low energies can be completely different, for example instead exhibiting van Hove asymptotics similar to that of the non-disordered Laplacian \cite{MuellerStollman2}. 

If one strongly desires the specific term $\mu_{xy}(f(x)-f(y))$ to appear, so as to match the combinatorial Laplacian of a weighted graph with edge weights $\mu_{xy}$, then the difference can be adjusted by changing the $\deg(x)$ part of the diagonal term to ${\deg(x)}-\mu_x$, where $\mu_x=\sum_{y:y\sim x}\mu_{xy}$.
\end{remark}

In order to define the integrated density of states (IDS) and landscape counting function, we first restrict to finite volume sets.
Let $A\subset \V$ be a finite subset of vertices and $\Gamma_A\subset \Gamma$ the subgraph induced by $A$. We denote by $H^A:\ell^2(A) \to \ell^2(A)$ the restriction of $H$ to $A$ with Dirichlet boundary conditions:
\begin{align}\label{eqn:HA}
H^Af(x)&={\deg(x)}f(x)-\sum_{y\in A:y\sim x}\mu_{xy}f(y)+V_xf(x).
\end{align}
The integrated density of states (IDS) of $H^A$ is then
\begin{align}\label{eqn:ids}
    N^A(E):=\frac{1}{|A|}\#\{{\rm eigenvalues}\ E'\ {\rm of}\ H^A \ {\rm such\ that\ }\ E'\le E\}.  
\end{align}

Next we define the {localization landscape function}.
Since $A$ is finite, $H^A$ is a strictly positive, bounded self-adjoint operator on $\ell^2(A)$, with a nonegative Green's function, $G_A(x,y)=(H^A)^{-1}(x,y)\ge0,\, \forall x,y\in A$. A direct consequence is that there is a unique vector $u_A\in C_+(A)=\R_+^{A}=\{f: A\to {\R_{>0}}\}$, called the \emph{landscape function} of $H^A$, solving $H^Au_A(x)=1,\, \forall x\in A$.

The concept of the landscape function was introduced by Filoche and Mayboroda in \cite{filoche2012universal} for studying localization in elliptic operators on $\R^d$.  
This approach has led to a number of results and applications in quantum and semiconductor physics; see for example \cite{arnold2019computing} for an overview of the landscape method and applications. Additionally, other mathematical results and applications of localization landscape theory have recently been developed in e.g. \cite{Bachmann2023,bachmann2024two,Hoskins2024,Lu2022,Poggi2024,Sd,Steinerberger2023}.
The use of the landscape function as an ``effective potential'' in place of the original potential $V$ allows one to work in a non-asymptotic regime, in some sense to bring the ideas behind the classical Weyl law, as well as the volume-counting of the Uncertainty Principle of Fefferman and Phong \cite{fefferman1983uncertainty}, to various models without
restrictions on the potential or the pertinent eigenvalues.

In this article, we will use the \emph{landscape counting function} to estimate the true IDS \eqref{eqn:ids} for operators on graphs. The landscape counting function was defined and used to estimate the IDS for operators on $\R^d$ in \cite{DFM}, and extended to operators on $\Z^d$ in \cite{arnold2022landscape}. In both of these previous cases, the landscape counting function involved counting cleanly partitioned cubes $Q$ where $\min_Q\frac{1}{u}\le E$. For graphs, we lack the periodicity and will have to make do without a partition, and instead use rougher coverings by (graph distance) balls.
Given $R\ge 0$, we call $\mathcal{P}=\mathcal P(R)=\{B(z_i,R): z_i \in \mathcal \V\}_{i\ge 1}$  a countable covering of $\Gamma$ with centers   $\{z_i\}_{i\ge 1}$ if $ \bigcup_{i\ge 1} B(z_i,R)=\V$ and $\{z_i\}_{i\ge 1}$ is countable. 
  We denote by $\mathcal P|_A=\{B\in \mathcal P: B\cap A\neq \varnothing\}$ (or $\mathcal P^A$) the restriction of $\mathcal P$ on $A\subset \Gamma$. 
We can then define the landscape counting function with respect to this cover $\mathcal{P}$ as follows. 
\begin{definition}
The \emph{landscape counting function for the covering $\mathcal{P}=\mathcal{P}(R)$ and region $A$} is defined for $E>0$ as
\begin{align}\label{eqn:Nu}
N_u^{\mathcal{P},A}(E) &= \frac{1}{|A|}\#\!\left\{B\in\mathcal{P}|_A:\min_{x\in B\cap A}\frac{1}{u(x)}\le E\right\}.
\end{align}
 
For the landscape law, we will take the radius of the balls in the partition $\mathcal{P}$ to be proportional to $E^{-1/2}$.
 As we show in  Lemma~\ref{lem:partition-comp}, landscape counting functions defined by different covers $\mathcal{P}(R)$ and $\mathcal{P}'(R')$ are comparable to each other as long as the radii $R,R'$ are comparable in the sense $c\le R'/R\le C$, and the covers satisfy a finite covering property \eqref{eqn:finite-cover-lambda}.
We will thus frequently omit the superscripts $\mathcal P, A$ and simply denote the \emph{landscape counting function} with balls of radius $R=E^{-1/2}$ for the region $A$ as 
\begin{align}\label{eqn:Nu-1/2}
N_u(E)=  N_u^{\mathcal P(E^{-1/2}),A}(E),
\end{align}  
with the understanding that the notation $N_u(E)$ is only meant to be defined up to a constant pre-factor, and is always to be taken with a covering of balls of radius $R=E^{-1/2}$ satisfying the finite covering property (for a fixed constant) in \eqref{eqn:finite-cover}. 
\end{definition}

\subsection{Main results}

Our first result is the landscape law for a general class of graphs, essentially those with the so-called Gaussian heat kernel bounds. 
The landscape law will demonstrate that the landscape counting function defined in \eqref{eqn:Nu} can be used to bound the actual integrated density of states for $H^A$ from above and below.
Afterwards, we give conditions with which one can use the landscape law to obtain Lifshitz tails for the integrated density of states.
The assumptions listed in the following theorems are defined precisely in Section~\ref{subsec:graph-prop}. For the landscape law, one specific example to keep in mind is graphs $\Gamma$ that are \emph{roughly isometric} to $\Z^d$, which will satisfy the required conditions. 
We provide the precise definition in Section~\ref{subsec:rough-isom}, but for now we note that rough isometries (or quasi-isometries \cite{Gromov}) are maps that capture the large-scale, global structure of the graphs.
For obtaining Lifshitz tails, we will additionally require technical conditions on the graph $\Gamma$ involving the harmonic weight of balls, which will be described in Assumption~\ref{ass:harmonic-weight} in Section~\ref{subsec:graph-prop}.

\begin{theorem}[Landscape Law for graphs] 
\label{thm:LLaw}
Let $H^A$ be as in \eqref{eqn:HA} with any $\mu_{xy}\in [0,1]$ and $V_x\ge 0$. Suppose $\Gamma$ 
satisfies  Assumption~\ref{ass:0}  
and a weak Poincar\'e inequality (WPI) as described in Section~\ref{subsec:graph-prop}. 
In particular, graphs roughly isometric to $\Z^d$, or more generally those with Gaussian heat kernel estimates \eqref{eqn:gaussian-heatkernel}, meet these requirements. Then we have the following landscape law bounds. 
\begin{enumerate}[(i)]

\item Upper bound: 
    \begin{align}\label{eqn:LLaw-upper}
       N(E)\le N_u(CE) ,\ \ {\rm for \ all }\ E>0,
    \end{align} where $C$ depends on $\Gamma$, 
    and in particular, is independent of $A$,  $\mu_{xy}$ and $V_x$. 

\item Lower bound:   
There are constants $c_i,c_i',c^*$, depending only on $\Gamma$, such that for any  $0<\kappa<1/4$,
  \begin{align}\label{eqn:LLaw-lower1}
N(E)&\ge c_1\kappa^\alpha N_u(c_3\kappa^{\alpha+2}E)-c_2N_u(c_4\kappa^{\alpha+4}E),\quad\text{for }{E\le c^*\kappa^{-4}},
\end{align}
 and 
 \begin{align}\label{eqn:LLaw-lower2}
N(E)&\ge c_1'  N_u(c_3'\kappa^{2}E)-c_2'N_u(c_4'\kappa^{4}E),\quad\text{for }{E>c^*\kappa^{-4}}.
\end{align}

\end{enumerate}
\end{theorem}

The above Landscape Law as given by \eqref{eqn:LLaw-upper}, \eqref{eqn:LLaw-lower1}, and \eqref{eqn:LLaw-lower2} holds for any $\mu_{xy}\in[0,1]$ and $V_x\ge0$, requiring no additional assumptions on $\mu_{xy}$ or $V_x$. Below, we will discuss the disordered model where $\{\mu_{xy}\}$ and  $\{V_x\}$ are each sets of independent, identically distributed (i.i.d.) random variables.  We see that in this case, under additional assumptions on the graph, the lower bound in \eqref{eqn:LLaw-lower1} and \eqref{eqn:LLaw-lower2} can be improved by removing the negative term on the right hand side, leading to Lifshitz tail estimates for both the landscape counting function and actual IDS in terms of the  cumulative  
 distribution functions (CDFs) of $\mu_{xy}$ and $V_x$.

In what follows, let $\sigma_{xy}=1-\mu_{xy}\in [0,1]$  be i.i.d. random variables with a common CDF
\begin{align}\label{eqn:cdfF-mu}
    F_{\mu}(E)=\P(1-\mu_{xy}\le E),
\end{align}
and let $V_x\ge 0$ be i.i.d. random variables with a common CDF 
\begin{align}\label{eqn:cdfF-V}
   F_V(E)= \P(V_x\le E).
\end{align}

\begin{remark}
    We will need only at least one of $F_V$ or $F_{\mu}$ to be non-trivial. As an allowable example, if $F_V(0)=1$, so that $V_x$ is identically zero for all $x$, then $H$ is free of potential and there are only off-diagonal disorder terms $\mu_{xy}$. 
\end{remark}

We start with the general lower bound of the landscape counting function $N_u$   in terms of $F$, which gives a Lifshitz tail lower bound.

\begin{theorem}[Landscape Lifshitz tails]\label{thm:Lif}
Suppose $\Gamma$ satisfies  Assumption~\ref{ass:0} 
and a weak Poincar\'e inequality (as described in Definition~\ref{ass:WPI} in Section~\ref{subsec:graph-prop}). Let $F_V$ and $F_\mu$ both have zero as the infimum of their essential support\footnote{$E$ belongs to the essential support of a distribution function $F$ iff $dF(E-\eps,E+\eps)>0$ for any $\eps>0$. Since $0$ is the essential support of $F_V,F_{\mu}$, $ F_V(0)F_{\mu}(0)<1$ implies that there is $E_0>0$ such that $F_V(E)F_{\mu}(E)\le  F_V(E_0)F_{\mu}(E_0)<1$ for all $E\le E_0$.},
 and satisfy $ F_V(0)F_{\mu}(0)<1$.
For any finite set $A\subset\Gamma$ with ``sufficient overlap with balls'', i.e.
there is $c>0$ so that for any $x\in A$ and  $0\le r\le 
 {C}\operatorname{diam}A$, where $\operatorname{diam}A:=\sup_{x,y\in A}d_0(x,y)$ and $C\ge1$ is $\Gamma$-dependent constant, 
\begin{align}\label{eqn:AintersectB}
     |A\cap B(x,r)|\ge cr^\alpha,
\end{align}
then we have the following bounds on the landscape counting function.
\begin{enumerate}[(i)]
\item Lower bound: There are constants $c_i,E_0$ only depending on $\Gamma$ such that 
\begin{align}\label{eqn:Lif-lower-Nu}
    \E N_u(E)\ge c_1E^{\alpha/2}\big(F_{\mu}(c_2E)\big)^{c_3E^{-\alpha/2}}\big(F_V(c_4E)\big)^{c_5E^{-\alpha/2}}
\end{align}
for all $ {c_0}{({\rm diam}A)^{-2}}\le E\le E_0$.

\item Upper bound: 
With the additional harmonic weight assumptions in Assumption~\ref{ass:harmonic-weight}, there are constants $c_i,E_1$ only depending on $\Gamma$ such that 
\begin{align}\label{eqn:Lif-upper-Nu}
    \E N_u(E)\le c_6E^{\alpha/2}\Big(F_\mu (c_7E)F_V (c_7E)\Big)^{c_8E^{-\alpha/2}}
\end{align}
for all $ E<E_1$. 
\end{enumerate}
\end{theorem}

As a consequence of Theorems~\ref{thm:LLaw} and \ref{thm:Lif}, the methods from \cite{DFM,arnold2022landscape} then provide Lifshitz tail estimates for the actual integrated density of states (IDS).
\begin{corollary}[Landscape Law for {random models}]\label{cor:LLaw-random}
   Retain the assumptions in Theorem~\ref{thm:Lif}(ii). Then there are constants $c_i$ only depending on $\Gamma$ such that for all $E >0$, 
    \begin{align}\label{eqn:LLaw-random}
        c_1\E N_u(c_2E)\le \E N(E) \le  c_3\E N_u(c_4E),
    \end{align}
    so that the Lifshitz tail estimates \eqref{eqn:Lif-lower-Nu} and \eqref{eqn:Lif-upper-Nu} hold for $\E N(E)$.
\end{corollary}

Throughout this article, we work with the finite volume IDS (eigenvalue counting  per unit volume) and do not directly consider  the thermodynamic limit as $A\nearrow \V$. In general, such limit may not exist unless the operator is in some sense uniform in  the underlying graph (e.g., a periodic or random Schr\"odinger operator defined on a vertex transitive graph). When the infinite volume IDS cannot be defined for general  operators, one can consider the $\liminf$ or $\limsup$ instead, which always exist, and provide 
 lower and upper bounds for the infinite volume one when it exists. Noting that the Lifshitz tail estimates  \eqref{eqn:Lif-lower-Nu}, \eqref{eqn:Lif-upper-Nu} (the constants therein) for $N_\bullet=N,N_u$ are independent of $A$, this allows us to take $\liminf/\limsup$ (for any fixed $E$ small) and obtain
\begin{multline*}
    c_1E^{\alpha/2}\big(F_{\mu}(c_2E)\big)^{c_3E^{-\alpha/2}}\big(F_V(c_4E)\big)^{c_5E^{-\alpha/2}} \le \liminf_{A \nearrow \V}\E N_\bullet(E) \\ \le \limsup_{A \nearrow \V}\E N_\bullet(E)\le c_6E^{\alpha/2}\Big(F_\mu (c_7E)F_V (c_7E)\Big)^{c_8E^{-\alpha/2}}, \ \ \ E<E_0. 
\end{multline*}
One can check that the double-log limit (in energy $E$) of the limit in $A$ does always exist as long as $F_VF_{\mu}$ is not `too thin' near the bottom, e.g., if $F_V(E)F_{\mu}(E)\gtrsim E^{c}$ for some $c>0$, then 
\begin{align*}
  \lim_{E\searrow 0}   \lim_{A \nearrow \V}\frac{\log \big|\log\E N_\bullet(E)\big|}{\log E}=-\frac{\alpha}{2}, 
\end{align*}
where $\frac{\alpha}{2}$ is usually referred as the  Lifshitz tail exponent. We obtain such  Lifshitz tail estimates as  a by-product of the landscape method. We refer readers to the extensive literature for more details and background about Lifshitz tails, see e.g. \cite{kirsch2008invitation,konig2016parabolic}. 

One may also notice that the Lifshitz tail exponent $\frac{\alpha}{2}$ coincides with the volume control parameter $\alpha$ in \eqref{eqn:vol-control} divided by the (weak) Poincar\'e inequality parameter $2$ in \eqref{eqn:WPI}. These two parameters together appear in the \emph{Heat Kernel Bound} $\mathrm{HK}(\alpha,2)$, a property for the free (probabilistic) Laplacian on the graph, which we describe further in Propositions~\ref{prop:piconseq}.  
We discuss the relation between more general heat kernel bound parameters on a general graph and the Lifshitz tails exponent in Section~\ref{sec:SG}, where we consider the fractal Sierpinski gasket graph.

In the following subsections, we describe applications of Theorem~\ref{thm:LLaw} or Corollary~\ref{cor:LLaw-random} to various models.

\subsection{Bond percolation and discrete $\operatorname{div}A\cdot\nabla$ on $\Z^d$}\label{subsec:bondperc}
Bernoulli bond percolation graphs are random graphs formed from a graph $G=(\V,\mathcal E)$ by assigning independent $\operatorname{Bernoulli}(p)$ random variables $(\omega_{xy})_{xy}$ to the edges $\{x,y\}\in \mathcal E$, and considering the graph $G'=(\V,\mathcal E')$ with new edge set $\mathcal E'=\{\{x,y\}\in \mathcal E:\omega_{xy}=1\}$.
There has been much interest concerning the spectral properties of such random graphs; see for example the overview \cite{MuellerStollman} on percolation Hamiltonians.

The ``adjacency'' or ``pseudo-Dirichlet'' Laplacian $\Delta_{\mathrm{PD}}$ on $\Z^d$ considered in \cite{KirschMueller,MuellerStollman2} corresponds to the Jacobi operator \eqref{eqn:H-gen} with no potential ($V_x\equiv0$), and with i.i.d.  $\mu_{xy}\sim\operatorname{Bernoulli}(p)$. There  
it was shown that $\Delta_{\mathrm{PD}}$ has Lifshitz tails at the bottom and top of the spectrum, for any $0<p<1$.  By applying Corollary~\ref{cor:LLaw-random}, we obtain an alternative proof of the adjacency Laplacian Lifshitz tails result via the landscape law method. (Since the underlying lattice is just $\Z^d$, Assumption~\ref{ass:harmonic-weight} required for Corollary~\ref{cor:LLaw-random} always holds.)

We note that the landscape law method also applies to i.i.d.  $\mu_{xy}$ with distributions other than Bernoulli, providing a landscape-law based proof for Lifshitz tails for these models as well. The Lifshitz tails upper bound for such models was proved earlier in \cite{KloppNakamura} by comparison to on-diagonal disorder models.

By a duality/symmetry argument \cite{KirschMueller}, the bottom of the spectrum of bond percolation Hamiltonians can be used to study the top of the spectrum for the \emph{discrete} version of $\operatorname{div}A\cdot\nabla$ on $\Z^d$. Such operators describe acoustic waves in a medium \cite{FigotinKlein}.
Let $H=\nabla^\dagger A\nabla$ on $\Z^d$, where $A(x)$ is a real symmetric semidefinite $d\times d$ matrix for each $x\in\Z^d$. When $A(x)$ is a non-negative diagonal matrix, each diagonal entry of the $d\times d$ matrix $A(x)$ can be associated with one of the $2d$ edges in the graph $\Z^d$ involving the node $x$. The operator $\nabla^\dagger A\nabla$ then just becomes the weighted nearest neighbor combinatorial Laplacian, 
\begin{align}\label{eqn:ham}
Hf(x)\equiv(\nabla^\dagger A\nabla f)(x) = \sum_{y\in\Z^d:y\sim x}a(x,y)(f(x)-f(y)),
\end{align}
where $a(x,y)=a(y,x)$ is the corresponding entry from one of the matrices $A(x)$. 
When $a(x,y)$ are i.i.d.  $\operatorname{Bernoulli}(p)$ random variables, then $H$ is the ``Neumann Laplacian'' from \cite{KirschMueller,MuellerStollman2}, 
and can have different behavior at the top vs bottom of the spectrum. The top of the spectrum always exhibits Lifshitz tails, but the bottom of the spectrum can also exhibit ``van Hove singularities'' \cite{MuellerStollman2}.
As we do not specialize to the precise operator dual to $H$, the duality/symmetry argument applied with Corollary~\ref{cor:LLaw-random} or \cite{KloppNakamura} yields the Lifshitz tails \emph{upper} bound for $1-N(E)$ for discrete $\operatorname{div}A\cdot\nabla$ at the top of the spectrum. This also applies for non-Bernoulli i.i.d.  disorder $a(x,y)$, and one could investigate if the matching lower bound can be obtained by applying the landscape law method to the precise dual operator.

\subsection{Graphs roughly isometric to $\Z^d$}\label{subsec:appl-isom}

While we have not yet defined rough isometries between graphs (we will do so in Section~\ref{subsec:rough-isom}), we provide a few brief examples here of graphs roughly isometric to $\Z^d$, to which Theorem~\ref{thm:LLaw} applies. 
Roughly speaking, roughly isometric will mean that there is a map between the graphs that preserves distances up to some error.

We first give the example of the vertex graph $\Gamma_P$ of the Penrose rhomb tiling of the plane (see for example the textbook \cite{BaakerGrimm}), shown in Figure~\ref{fig:penrose}. 
\begin{figure}[ht]
 \includegraphics[width=0.5\textwidth]{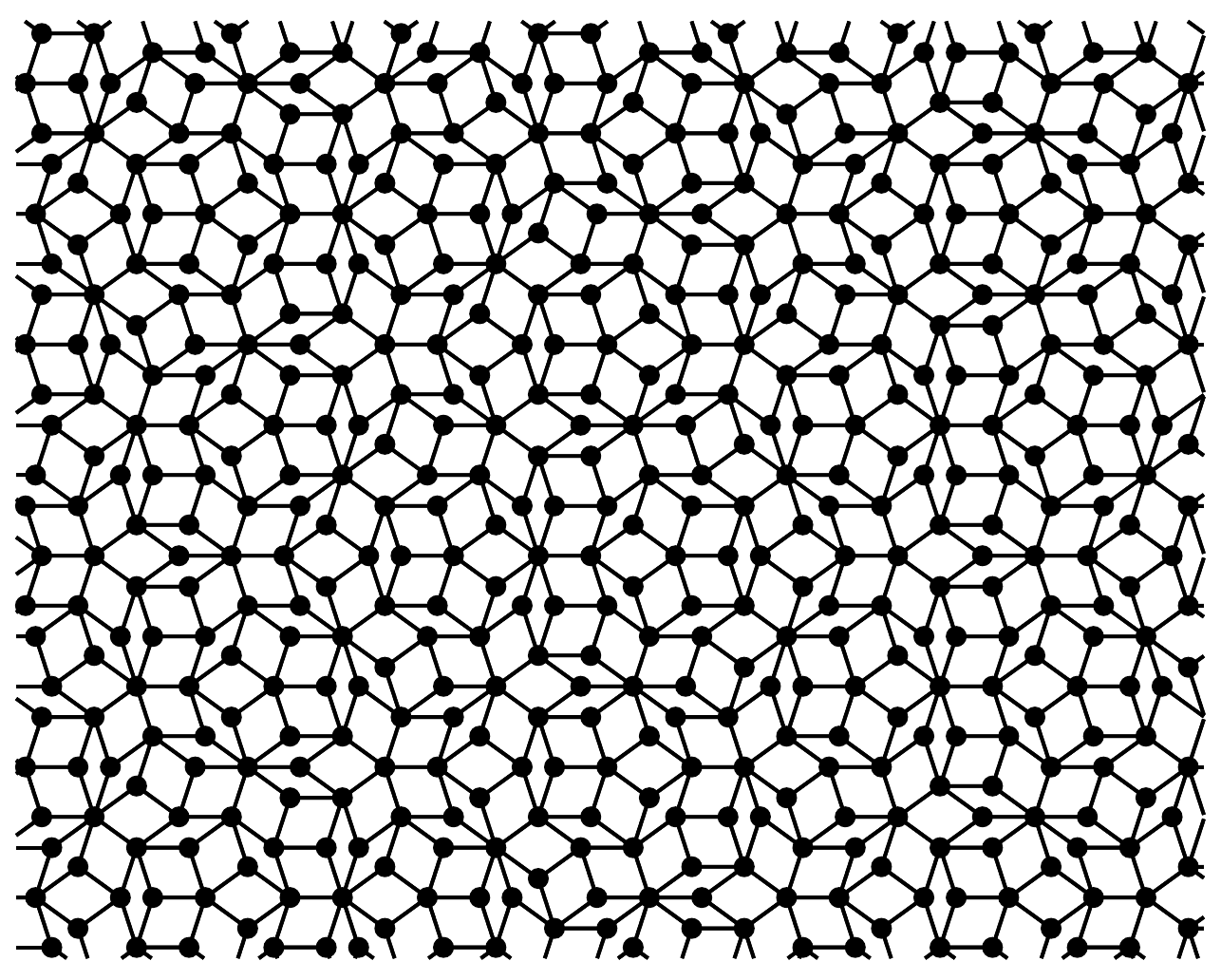}
 \caption{A section of the Penrose tiling with the vertices of $\Gamma_P$ drawn. }\label{fig:penrose}
 \end{figure}
The tiling involves two kinds of rhombuses, a wide rhombus with angles $72^\circ$ and $108^\circ$, and a narrow rhombus with angles $36^\circ$ and $144^\circ$.
This Penrose tiling possesses a five-fold rotational symmetry, but is aperiodic, with no translational symmetry.
While the vertex graph $\Gamma_P$ is also nonregular, it is roughly isometric to $\Z^2$ (see Section~\ref{subsec:rough-isom}), so that Theorem~\ref{thm:LLaw} implies the Landscape Law for Jacobi operators on the Penrose tiling graph $\Gamma_P$.

Next, we consider lattices and some of their variations. Lattices such as the triangular and hexagonal lattices are readily seen to be roughly isometric to $\Z^2$. Local perturbations of lattices, such as adding decorations to sites, also remain roughly isometric. Stacked lattices, obtained by taking $\Gamma\times\Z_M$ for a fixed $M$ and adding edges between identical sites in adjacent layers, also remain roughly isometric. The landscape law Theorem~\ref{thm:LLaw} then holds for all these graphs.
\begin{figure}[!ht]
\includegraphics[height=1.5in]{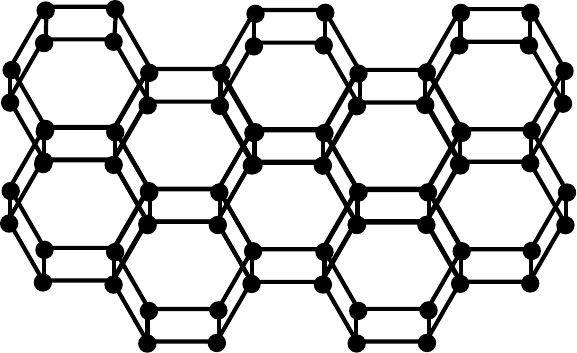}
\caption{Stacked lattices illustration.}\label{fig:stack}
\end{figure}

\subsection{Lifshitz tails for stacked graphs}\label{subsec:app-stack}
For a graph $\Gamma=(\V,\mathcal{E})$, construct the stacked graph $\Gamma\times\Z_M$ as follows, similarly as described for stacked lattices in the previous subsection. The vertices of $\Gamma\times\Z_M$ are $(x,j)$, for $x\in\V$ and $j\in\{1,\ldots,M\}$, and the edges are those in each copy of $\Gamma$ ($(x,j)\sim(y,j)$ if $x\sim y$ in $\Gamma$), along with new edges between identical sites in adjacent copies of $\Gamma$, $(x,j)\sim(x,k)$ if $|j-k|=1$. This is easiest to visualize when $\Gamma$ is e.g. a 2D lattice or tiling graph such as in Figure~\ref{fig:stack}.

If the properties in Assumption~\ref{ass:harmonic-weight} (defined below in Section~\ref{subsec:graph-prop}) hold with the natural metric $d_\Gamma$ and harmonic weight \eqref{eqn:natural-harmonic} for a graph $\Gamma$, then as we verify in Section~\ref{subsec:stack}, the required properties also hold for the described $\Gamma\times\Z_M$ with metric $\tilde{d}((x,j),(y,k)):=d_\Gamma(x,y)+\frac{1}{2}\mathbf{1}_{(x=y)\wedge (j\ne k)}$ 
and ``bad set'' $X_R\times\Z_M$. 
Thus if $\Gamma$ also satisfies the other hypotheses for Theorem~\ref{thm:Lif}, then this implies that Lifshitz tails for both $\Gamma$ and the stacked model $\Gamma\times\Z_M$ follow from Corollary~\ref{cor:LLaw-random} in this situation.

\subsection{Sierpinski gasket graph}\label{sec:SG-intro}
We briefly discuss the Sierpinski gasket graph (Figure~\ref{fig:sierpinski}), which is a fractal graph. It is not roughly isometric to any $\Z^d$, but it satisfies a similar type of heat kernel bounds, called \emph{sub-Gaussian} heat kernel bounds (for $\Z^d$, one has \emph{Gaussian} heat kernel bounds). The sub-Gaussian heat kernel bounds allow us to obtain the Landscape Law upper bound Theorem~\ref{thm:Lif}(ii), for Jacobi operators on the Sierpinski gasket graph. (Section~\ref{sec:SG}.)

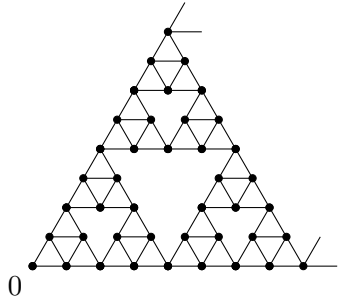
\begin{figure}[!ht]
\begin{tikzpicture}[scale=.45]
\def\points{{0,0},{1,0},{.5,.866025},{1.5,.866025},{1,0},{2,0},{1,1.73205},{0,0}}
\foreach \shift in {(0,0),(2,0),(4,0),(6,0),(1,1.73205),(5,1.73205),(2,3.4641),(4,3.4641),(3,5.19615)}{
	\foreach \c[remember=\c as \clast (initially {0,0})] in \points
	{
		\draw[shift=\shift cm] (\clast) -- (\c);
		\filldraw[shift=\shift cm, black] (\c) circle (3pt);
	}
}
\node[below left] at (0,0) {$0$};
\draw (4,6.9282) -- ++(60:1cm);
\draw (4,6.9282) -- ++(0:1cm);
\draw (8,0)--(9,0);
\draw (8,0) -- ++(60:1cm);
\end{tikzpicture}
\caption{Sierpinski gasket graph. This is a graph embedded in $\R^2$ whose edges are essentially the edges of triangles in a Sierpinski triangle. The bottom left corner is $0$, each edge is length $1$, and the graph grows unboundedly in the positive $x$ and $y$ directions.}\label{fig:sierpinski}
\end{figure}

\subsection{Random band models}\label{subsec:rbm}
In this part, we discuss the graph induced by a band matrix on $\Z^d$, for which we will see all assumptions required for the Lifshitz tails \eqref{eqn:LLaw-random} in Corollary~\ref{cor:LLaw-random} hold. More precisely, let $\Gamma_{d,W}=(\Z^d,\mathcal E_W)$ be a (naturally weighted) graph with the vertex set $\Z^d$. For a positive integer $W$, the edge set $\mathcal E_W$ is defined as follows:
\begin{align}\label{eqn:Gamma-dW}
    \mathcal E_W=\{(x,y)\in \Z^d\times \Z^d: 0<\|x-y\|\le W\},
\end{align}
where the norm $\|\cdot\|$ can be the 1-norm $\|\cdot\|_{\ell^1(\Z^d)}$ (induced by the shortest path metric $d_0$ on $\Z^d$), or the $\infty$-norm $\|\cdot\|_{\ell^{\infty}(\Z^d)}$, or  the Euclidean-norm $ \|\cdot\|_{\ell^2(\Z^d)} $, or any other norm which is equivalent to these norms.
For $x,y\in \Z^d$, we then write $x\sim y$ iff $0<\|x-y\|\le W$. The Jacobi operator in \eqref{eqn:H-gen} is denoted as $H_{d,W}$ and acts on functions on $\Gamma_{d,W}$ as
\begin{align}\label{eqn:HdW}
    H_{d,W}f(x)&= \sum_{y:y\sim x} \big(f(x)-\mu_{xy}f(y)\big) +V_x f(x). 
\end{align}
The (negative) Laplacian on $\Gamma_{d,W}$, denoted as $-\Delta_{d,W}$, is given by setting $\mu_{xy}=1$ and $V_x=0$ in $H_{d,W}$. An example of the graph for $W=2$, $d=1$ is shown in Figure~\ref{fig:Lap72}, where the matrix representation of the negative Dirichlet sub-graph Laplacian $-\Delta_{1,2}^I,I=\intbr{1}{7}$ is 
 \begin{align*} 
    -\Delta_{1,2}^{I}=&\begin{pmatrix}
	4 & -1 & -1  & 0 & 0 & 0 &  0 \\
	-1 &   4 & -1 & -1 &0&0  &  0 \\
-1	& -1  & 4 &  -1 & -1  & 0 & 0\\
0 & -1& -1& 4  & -1 & -1 & 0 \\
0  &0 & -1 & -1& 4 & -1   & -1  \\
0  &0 & 0 & -1& -1  & 4  & -1  \\
0  &0 & 0 & 0&  -1 & -1  & 4
\end{pmatrix}.
\end{align*}  
The random operator $H_{d,W}$ in matrix form corresponds to adding the potential $V_x$ to the diagonal, and replacing off-diagonal $-1$s with the appropriate $-\mu_{xy}$.
\begin{figure}[htb]
    \centering
     \begin{tikzpicture}
 \foreach \x in {0,1,2,3,4,5,6,7,8}
   \draw (\x cm,1pt) -- (\x cm,-1pt) node[anchor=north] {$\x$};
\draw (1,0) -- (2,0)--(3,0) -- (4,0)--(5,0) -- (6,0)--(7,0);
   \draw (1,0) .. controls (2,1)  .. (3,0);
 \draw (2,0) .. controls (3,-1)  .. (4,0);
   \draw (3,0) .. controls (4,1)  .. (5,0);
   \draw (4,0) .. controls (5,-1)  .. (6,0);
  
    \draw (5,0) .. controls (6,1)  .. (7,0);
    \draw[dashed] (0,0) -- (1,0);
    \draw[dashed] (7,0) -- (8,0);
       \draw[dashed] (-1,0) .. controls (0,1)  .. (1,0);
       \draw[dashed] (0,0) .. controls (1,-1)  .. (2,0);
        \draw[dashed] (7,0) .. controls (8,1)  .. (9,0);
         \draw[dashed] (6,0) .. controls (7,-1)  .. (8,0);
\end{tikzpicture}
    \caption{The graph $\Gamma_{1,2}=(\Z,\mathcal E_2)$ consists of vertices $\Z$ and edges $\mathcal E_2=\{(i,j)\in \Z^2: 0<|i-j|\le 2\}$. 
    }
    \label{fig:Lap72}
\end{figure}
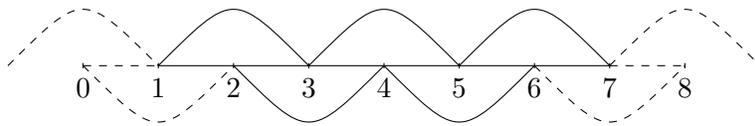

For general $d$ and $W$, the band graph $\Gamma_{d,W}$ is a notable example where all assumptions required in Theorems~\ref{thm:LLaw} and \ref{thm:Lif} hold. We will prove this later in Section \ref{sec:RBM}. As a consequence, we obtain both the landscape law and the Lifshitz tails,
\begin{corollary}[application to random band matrices]\label{cor:HdW}
    Let $A\subset \Z^d$ and $H^A_{d,W}$ be the restriction of $H_{d,W}$ on $A$. Then the {\it{Landscape Law} } \eqref{eqn:LLaw-upper}, \eqref{eqn:LLaw-lower1}, and \eqref{eqn:LLaw-lower2} holds with $\alpha=d$, where all constants $c_i$ depend only on $d,W$.  

If $\{\mu_{xy}\}$ and $\{V_x\}$ are sets of i.i.d. random variables with the same assumptions in Theorem~\ref{thm:Lif}, and $A$ satisfies \eqref{eqn:AintersectB},
then  
there are constants $c_i,E_0$ only depending on $d$ and $W$ such that 
\begin{align}\label{eqn:LLaw-random-dW}
       c_1\E N_u(c_2E)\le \E N(E) \le  c_3\E N_u(c_4E) 
\end{align}
for all $E>0$. Additionally, 
\begin{align}\label{eqn:Lif-dW}
c_5E^{d/2}\big(F(c_6E)\big)^{c_7E^{-d/2}}\le     \E N_u(E),\E N(E)\le c_8E^{d/2}\big(F(c_9E)\big)^{c_{10}E^{-d/2}},
\end{align}
for all   $ {c_0}{({\rm diam}A)^{-2}}\le E\le E_0$.  
\end{corollary}

Lifshitz tails for random band matrices were also proved in \cite{rband}, using different methods (a variational argument and comparison with a diagonal model).

For other graphs such as the Penrose tiling or Sierpinski gasket graph, we do not have a complete understanding of the geometric properties that would be required to obtain all of the desired results as in Corollary~\ref{cor:HdW}, though we do obtain some of the results as discussed in previous sections.
Instead, in Section~\ref{sec:numerics}, we numerically study the landscape law  on some of these models, providing strong evidence that the landscape counting functions can be used to establish effective upper and lower bounds for the IDS across various models. Additionally, the numerics demonstrate the advantage of the landscape counting function in productively reflecting the Lifshitz tails through appropriate scalings. 

\subsection{Outline}
The rest of this article is organized as follows. 
\begin{itemize}
\item In Section~\ref{sec:pre}, we introduce background and preliminaries concerning operators on graphs, the landscape function and landscape counting function, and rough isometries of graphs. 
We provide precise definitions for all terms and assumptions mentioned for the main results.
\item In Section~\ref{sec:llaw}, we prove Theorem~\ref{thm:LLaw}, first the upper bound part (i), and then the lower bound part (ii).
\item In Section~\ref{sec:lif}, we prove Theorem~\ref{thm:Lif} on the landscape law Lifshitz tails. 
\item In Section~\ref{sec:ex}, we finish the proofs for applying the main results to the specific examples in Sections~\ref{subsec:app-stack}--\ref{subsec:rbm}.
\item Finally, in Section~\ref{sec:numerics}, we provide numerical results for Anderson models on the Penrose tiling and Sierpinski gasket, and for random band models with bond disorders.
\end{itemize}

Throughout the paper, constants such as $C$, $c$, and $c_i$ may change from line to line. 
We will use the notation $X\lesssim Y$ to mean $X\le cY$, and $X\gtrsim Y$ to mean $X\ge cY$, for some
 constant $c$ depending only on $\Gamma$. If $X\lesssim Y\lesssim X$, we may also write $X\approx Y$.

\section{Preliminaries}\label{sec:pre}

In this section, we introduce graph operators and the landscape function, followed by the precise assumptions on the graphs and discussion concerning rough isometries.

\subsection{Operators on graphs and the landscape function}\label{sec:pre0}
 For convenience, we collect several useful facts from graph operator theory and landscape function theory in this section. 
More background and details can be found in  \cite[\S2]{arnold2022landscape}, \cite[\S2]{lratio}, and the references therein.

We consider functions on the vertices of the graph, which will be denoted by the function space 
$C(\V):=\R^{\V}=\{f: \V\to \R\}$.  
We denote by  $\one \in C(\V)$ the function which is identically $1$, and by $\one_A$ the indicator function of $A\subset \V$.
 The space $\ell^2(\V)$ is defined via the $\ell^2$ norm induced by the usual (non-weighted) inner product  $\langle f,g\rangle:=\sum_{x\in \V} {f(x)}g(x). $
The subspaces $\ell^2(A)$ and $\ell^{\infty}(A)$ are defined accordingly for any finite subset $A\subset \V$.

For $H^A$ as defined in \eqref{eqn:HA}, the corresponding quadratic form is
\begin{align*}
    \ipc{f}{H^Af}=   \sum_{x\in A}({\deg(x)}+V_x)f(x)^2 -  \sum_{\substack{ x,y\in A\\ y\sim x} }\mu_{xy}f(x) f(y),
\end{align*}
which can also be written as
\begin{align*}
  \ipc{f}{H^Af}=   \frac{1}{2} \sum_{\substack{ x,y\in A\\ y\sim x} }\mu_{xy}\big(f(x)-f(y)\big)^2+\sum_{x\in A}({\deg(x)}- \mu_x^A+V_x)f(x)^2,
\end{align*}
 where $ \mu_x^A=\sum_{y\in A:y\sim x}\mu_{xy} $.

  The following maximum principle will be useful in obtaining bounds on the landscape function.
\begin{definition}\label{def:maxP}
 For $A\subset \V$, we say $H^A:\ell^2(A)\to \ell^2(A)$ satisfies the (weak) maximum principle if $\inf _{x\in A}(H^Af)(x)\ge 0$ implies $\inf _{x\in A} f(x)\ge 0$ . 
\end{definition}
\begin{lemma}[Lemma~2.1, \cite{lratio}]\label{lem:u-positive}
Let $A\subset\V$ be a finite subset.
  If $H^A$ {is strictly positive and}  satisfies the (weak) maximum principle, then $ \braket{x}{(H^A)^{-1}}{y}\ge 0$ for all $x,y\in A$. As a consequence, the landscape function $u_A(x)=(H^A)^{-1}\one (x)>0$ for all $x\in A$. 
\end{lemma}

The (weak) maximum principle holds widely, including  for $H$ (and its restriction $H^A$) given in \eqref{eqn:H-gen}.  
\begin{lemma}[Lemma~2.2, \cite{lratio}]\label{lem:maxP}
Let $H$ be as in \eqref{eqn:H-gen}, and let $H^A$ be its restriction   on $\ell^2(A)$ for some finite subset $A\subset \V$. Assume that $0\le \mu_{xy}=\mu_{yx}\le 1$ and $V_x\ge 0$ for all $x,y\in A$. Then $H^A$ is self-adjoint 
and has positive spectrum, and satisfies the (weak) maximum principle as in Definition~\ref{def:maxP}. 
\end{lemma}

\begin{lemma}[Landscape uncertainty principle] There is a unique $u=u^A\in {\mathcal H^A}:=\ell^2(A) $ such that  $H^Au=\one_A$. In addition,  for $f\in {\mathcal H^A}$,
    \begin{align}\label{eqn:UCP-gen}
        \ipc{f}{H^Af}=\frac{1}{2}\sum_{\substack{x,y\in A:\\y\sim x}}\mu_{xy}u(x)u(y)\bigg(\frac{f(x)}{u(x)}-\frac{f(y)}{u(y)}\bigg)^2 + \sum_{{x\in A}}\frac{f^2(x)}{u(x)}.
    \end{align}
\end{lemma}

\subsection{Graph properties and assumptions}\label{subsec:graph-prop}
Roughly speaking, we will want to work on graphs that are amenable to tools from harmonic analysis. 
This will allow us to adapt tools and methods from \cite{DFM,arnold2022landscape}.
Our first obstruction in moving from $\R^d$ or $\Z^d$ to graphs comes from the non-regular geometry. For the landscape laws in \cite{DFM,arnold2022landscape}, the periodic structure of $\R^d$ and $\Z^d$ was used to construct perfect partitions of cubes and perform arguments utilizing the resulting translation invariance. 
For graphs, we have no such possibilities and instead must work with \emph{coverings} rather than partitions, and with certain collections of coverings rather than translation invariance.
To this end, we will need certain geometric covering properties. 
We have the following covering results which are consequences of the volume control assumption \eqref{eqn:vol-control}.

\begin{proposition}[covering properties]\label{prop:covering}
Assume volume control \eqref{eqn:vol-control}. Then the following properties hold.
\begin{enumerate}[(a)]
\item Finite overlap:
Then there is a constant $b_\Gamma>0$ 
such that for any   $R\ge 1$, there is a covering $ {\mathcal P}(R)=\{B(z_i,R), z_i\in \V\}_{i\ge 1}  $ such that 
\begin{align}
   \sup_{y\in \V}\#\{B \in {\mathcal P}(R):y\in   B\}\le b_{\Gamma}<\infty.  \label{eqn:finite-cover}
\end{align}

\item Translation replacement: There are constants $C$ and $b$ so that for any $0<\kappa<1$ and $R\ge1$, there is a collection of at most $C\kappa^{-\alpha}$ covers $\mathcal{P}^{(j)}(R)=\{B(z_i^{(j)},R)\}_i$, all with finite covering constant $b$, so that
$\bigcup_j\bigcup_i B(z_i^{(j)},\kappa R)$ is a covering (not necessarily with any particular covering constant) of $\Gamma$. 
\end{enumerate}
\end{proposition}

\begin{proof}[Proof of Proposition~\ref{prop:covering}]
(a) Follows from \cite[Lemma 6.2]{barlow2017random}.

(b) First, we can cover a ball $B(z,R)$ with $\le C_\Gamma \kappa^{-\alpha}$ smaller balls $\{B(x^{(j)},\kappa R)\}_j$, where $x^{(j)}\in B(z,R)$, using a similar volume comparison argument as in \cite[Lemma 6.2]{barlow2017random}. Take a maximal set of points $x^{(j)}\in B(z,R)$ so that $B(x^{(j)},\kappa R/2)$ are disjoint, and let $M$ be the number of such points. Then $\bigcup_{j=1}^MB(x^{(j)},\kappa R)$ covers $B(z,R)$ by maximality, and $\bigcup_{j=1}^MB(x^{(j)},\kappa R/2)$ is a disjoint union contained in the larger ball $B(z,3R/2)$. Thus
\begin{align*}
C_U(3/2)^\alpha R^\alpha\ge |B(z,3R/2)|\ge \left|\bigcup_{j=1}^MB(x^{(j)},\kappa R/2)\right| &=\sum_{j=1}^M|B(x^{(j)},\kappa R/2)|\ge M \tilde{C}_L2^{-\alpha}\kappa^\alpha R^\alpha,
\end{align*}
where $\tilde{C}_L:=\min(C_L,1)$, so that $M\le \frac{C_U}{\tilde{C}_L}3^\alpha \kappa^{-\alpha}$.

Next, let $\mathcal{P}(R/2)=\{B(z_i,R/2)\}_i$ be a covering with finite overlap constant $b_\Gamma$ as in \eqref{eqn:finite-cover}. For each ball $B(z_i,R/2)$, cover it by the $M_i$ smaller balls $\mathcal{C}_i=\{B(x_i^{(j)},\kappa R/2)\}_{j=1}^{M_i}$ as described above (with radius $R/2$). We form the covers $\mathcal{P}^{(j)}(R)$ by taking one center $x_i^{(j)}$ from each $\mathcal{C}_i$. By repeating the elements of $\mathcal{C}_i$, 
we can ensure there are $M=\lfloor \frac{C_U}{\tilde{C}_L}3^\alpha \kappa^{-\alpha}\rfloor$ elements in $\mathcal{C}_i$. Then for $j=1,\ldots,M$, define
\begin{align*}
\mathcal{P}^{(j)}(R)&= \{B(x_i^{(j)}, R)\}_i,
\end{align*}
which covers $\Gamma$ since $B(x_i^{(j)},R)\supset B(z_i,R/2)$. 
By construction, $\bigcup_i\bigcup_{j=1}^{M} B(x_i^{(j)},\kappa R)$ covers $\Gamma$ (even with radii $\kappa R/2$ instead). 
Finally, each $\mathcal{P}^{(j)}(R)$ has finite overlap covering constant
\begin{align*}
\#\{i:y\in B(x_i^{(j)},R)\}\le \#\{i:y\in B(z_i,2R)\}\le \frac{C_U}{{C}_L}8^\alpha b_\Gamma,
\end{align*}
using \eqref{eqn:finite-cover-lambda} in Lemma~\ref{lem:scaled-cover} for the last inequality.
\end{proof}

\begin{example}\label{ex:Zd} 
For the standard $\Z^d$ lattice, we can consider different balls, such as a cube ($\ell^\infty$ ball), Euclidean $\ell^2$ ball, or $\ell^1$ (graph/natural metric) ball, 
and construct explicit coverings that satisfy Proposition~\ref{prop:covering}.
For example, the finite overlap property in part (a) is clear for cubes, and follows for the other balls, for example by embedding a cube within each ball and considering a covering with those cubes. 

For part (b), we can use the fact that $\Z^d$ does have translation invariance. 
Given a  cover $\mathcal P (R)=\{B(z_i,R)\}_i$ with finite overlap constant as in (a), consider its translation $\mathcal P^{(j)} (R)=\{B(T^jz_i,R)\}_i$ where $T^j$ is a translation (affine map) mapping the center $z_i$ to $T^jz_i\in B(z_i,R)$. 
Then for $r<R$, considering cubes or cubes embedded in the other balls, we only need at most $C(R/r)^d$ many $T^j$ so that $\{B(T^jz_i,r)\}_j$  covers $B(z_i,R)$ for any fixed $i$, and then $\cup_i\cup_jB(T^jz_i,r)$ covers $\Z^d$.
\end{example}

Under the finite overlap covering property in \eqref{eqn:finite-cover-lambda}, we prove the following  
 two lemmas, first that one also obtains a finite overlap property for larger balls with the same centers, and second, 
that the landscape counting function $N_u$ defined in \eqref{eqn:Nu} is comparable (up to a constant factor) across different choices of coverings $\mathcal{P}$.

\begin{lemma}[Scaled finite overlap covering property]\label{lem:scaled-cover}
Under the volume control ${\rm Vol}(B(x,r))\asymp r^\alpha$ from \eqref{eqn:vol-control}, the finite overlap  condition \eqref{eqn:finite-cover} for a partition $\mathcal{P}(r)=\{B(z_i,r)\}_{i\in\mathfrak{I}}$ implies that for any $\lambda\ge1$, 
\begin{align}\label{eqn:finite-cover-lambda}
   \sup_{y\in \V}\#\{z_i:y\in B(z_i,\lambda r)\} \le \frac{2^\alpha C_U}{C_L}\lambda^\alpha b_{\Gamma}.
\end{align}
\end{lemma}
\begin{proof}
This is a volume comparison. In what follows, $z_i$ will always denote a center of a ball in the given covering $\mathcal{P}(r)$. First, the set $\{z_i:y\in B(z_i,\lambda r)\}=\{z_i: d(y,z_i)\le \lambda r\}$ is the set $\{z_i:z_i\in B(y,\lambda r)\}$.
Let
$S=\bigcup_{z_i\in B(y,\lambda r)}B(z_i,r)$, which is contained in $B(y,\lambda r+r)$. By equation \eqref{eqn:finite-cover}, each point $x\in S$ can appear in at most $b_\Gamma$ of the balls $B(z_i,r)$, so that
\begin{align*}
\sum_{z_i\in B(y,\lambda r)}|B(z_i,r)| & = \sum_{x\in S}\#\{B(z_i,r):x\in B(z_i,r)\} \\
&\le b_\Gamma |B(y,\lambda r+r)|.
\end{align*}
Since $|B(z_i,r)|\ge C_Lr^\alpha$ and $|B(y,\lambda r+r)|\le C_U(\lambda+1)^\alpha r^\alpha$, this yields,
\begin{align*}
\#\{z_i:z_i\in B(y,\lambda r)\} &=\sum_{z_i\in B(y,\lambda r)}1 
\le b_\Gamma\frac{C_U(\lambda+1)^\alpha r^\alpha}{C_Lr^\alpha}=\frac{C_U}{C_L}(\lambda+1)^\alpha b_\Gamma.
\end{align*}
\end{proof}

If we consider different coverings with similar radii $R$ and $R'$, the corresponding landscape counting functions differ only by a constant pre-factor depending on $\Gamma$. This allows us to work with any specific choice of partition $\mathcal{P}$ satisfying the finite overlap property \eqref{eqn:finite-cover},
since the resulting landscape counting functions are equivalent up to the constant factor.
\begin{lemma}[partition comparison]\label{lem:partition-comp}
Suppose $\Gamma$ has volume growth control \eqref{eqn:vol-control}.
Let $\mathcal{P}$ be a covering of balls of radius $R$, and $\mathcal{P}'$ a covering of balls of radius $R'$. Suppose that $R'\le C_0R$ and that $\mathcal{P}$ has cover overlap constant $b_\Gamma$ in \eqref{eqn:finite-cover}. Then for any region $A$ and energy $E$,
\begin{align*}
N_u^{\mathcal{P}(R),A}(E)\le (1+C_0)^\alpha C_\Gamma b_\Gamma N_u^{\mathcal{P}'(R'),A}(E).
\end{align*}
\end{lemma}
\begin{proof}
Let $\mathcal{P}|_A=\{B\in\mathcal{P}:B\cap A\ne\emptyset\}$.
Note that since any point ${x}\in B\cap A$ with $\frac{1}{u(x)}\le E$ is also contained in some $B'\in\mathcal{P}'|_A$, we have
\begin{align}
\nonumber
\{B\in\mathcal{P}|_A:\min_{x\in B\cap A}\frac{1}{u(x)}\le E\} &\subseteq \{B\in\mathcal{P}|_A:B\text{ overlaps some }B'\in\mathcal{P}'|_A\text{ with }\min_{x\in B'\cap A}\frac{1}{u(x)}\le E\}\\
&\,\,=\bigcup_{\substack{B'\in\mathcal{P}'|_A\\\min_{ B'\cap A}\frac{1}{u(x)}\le E}}\{B\in\mathcal{P}|_A:B\text{ overlaps }B'\}.\label{eqn:BB'}
\end{align}
If $B=B(z,R)$ and $B'=B(z',R')$ overlap, we must have $d(z,z')\le R+R'\le (1+C_0)R$. By \eqref{eqn:finite-cover-lambda} in Lemma~\ref{lem:scaled-cover} applied to $z'$, there can be at most $(1+C_0)^\alpha C_\Gamma b_\Gamma$ balls $B\in\mathcal{P}$ with centers satisfying $d(z,z')\le (1+C_0)R$. Thus the cardinality of \eqref{eqn:BB'} is bounded above by $N^{\mathcal{P}'(R'),A}_u(E)(1+C_0)^\alpha C_\Gamma b_\Gamma$.
\end{proof}

Next, we discuss several traditional concepts from harmonic analysis that are known to carry over to the discrete setting on graphs.

\begin{definition}\label{ass:WPI}
    The graph $\Gamma$ satisfies a (weak) Poincar\'e inequality (PI) if there exist $C_P$ and $\lambda\ge 1$ such that, for all $x\in \V, r\ge 1$,  if $B=B(x,r),B^\ast=B(x,\lambda r) $ and  $f:B^\ast\to \R$, then 
     \begin{align}\label{eqn:WPI}
    \sum_{x\in B }\big(f(x)-\bar f^B\big)^2 \le  C_P\, r^2\sum_{x,y\in B^\ast: x\sim y}\big(f(x)-f(y)\big)^2  
    \end{align}
    where 
   \begin{align}\label{eqn:bar-f}
   \bar f^{B  }=\frac{1}{\deg(B)}\sum_{x\in B}f(x)\deg(x), 
   \quad \deg (x)=\sum_{ y\in\V: y\sim x}1,\quad \deg(B)=\sum_{x}\deg (x).
\end{align}
\end{definition}
It is well-known that the standard $\Z^d$ graph satisfies the weak PI \eqref{eqn:WPI},  
see e.g. \cite[Cor. 3.30]{barlow2017random}.
Note also that by embedding different balls within each other, if the graph $\Gamma$ satisfies the weak PI with one particular metric for defining the balls, then it also satisfies the weak PI (with possibly different scaling constant $\lambda$) with any strongly equivalent metric.

\begin{remark}\label{remk:beta-PI}
  In general, one can consider a $\beta$-Poincar\'e inequality for any $\beta\ge 2$,  with the power $r^\beta$ rather than  $r^2$ in \eqref{eqn:WPI}.  Under a $\beta$-Poincar\'e inequality, if one revises the definition of $N_u(E)$ in \eqref{eqn:Nu} using covers $\mathcal P(E^{-1/\beta})$ of radius $E^{-1/\beta}$ rather than $E^{-1/2}$, then one can still obtain the landscape law upper bound \eqref{eqn:LLaw-upper}. We will discuss this further for the Sierpinski gasket graph in Section \ref{sec:SG}. 
\end{remark}

The weak PI is known to be connected to several other notions concerning random walks and harmonic analysis on graphs.
For any $B_r=B(x_0,r)\subset \V$, let $\Delta^{B_r}$ be the Dirichlet Laplacian on $B(x_0,r)$,  
\begin{align}\label{eqn:laplacian}
(\Delta^{B_r}f)(x)=\sum_{y\in B_r(x_0):y\sim x}\big(f(y)-f(x)\big), 
\end{align}
and let $G_r(x,y)=(-\Delta^{B_r})^{-1}(x,y)$ be the Green's function 
on such a ball $\bar{B}(x_0,r)$.

\begin{proposition}\label{prop:piconseq}
If $\Gamma$ satisfies the weak PI \eqref{eqn:WPI} and volume control \eqref{eqn:vol-control} with parameter $\alpha\ge1$, then it also satisfies the following properties.
\begin{enumerate}[(i)]
\item Moser--Harnack inequality  for subharmonic functions: There exists a constant $C_{H}$ such that  for any $x\in \V, r>0$  and non-negative subharmonic function $f(x)$ (that is, $f(x)\ge0$ and $-\Delta f(x)\le 0$), then
    \begin{align}\label{eqn:Moser-Harnack}
        \sup_{y\in B(x,r)} f(y)^2 \le \frac{C_{H}}{|B(x,2r)|}\sum_{y\in B(x,2r)} f(y)^2.
    \end{align}

\item Free landscape/exit time upper bound: There is a positive constant $c$ such that for $x\in \V$, $r\ge 1$, and $B=B(x,r)$,
\begin{align}\label{eqn:u-upper}
     (-\Delta^{B })^{-1}  \one_B(x)\le cr^2. 
 \end{align}
(Additionally, by the maximum principle Lemma~\ref{lem:maxP}, the inequality also holds for the landscape function $u_B\le(-\Delta^{B })^{-1}\one_B$.)

\item Green's function estimates: For any $x\in B_R(x_0)$, let $r=d(x_0,x)$. There are constants $C_1,C_2$ depending only on $\alpha$ and the ratio $\frac{R}{r}>1$, such that 
  \begin{align}\label{eqn:Green-ass}
  C_1 r^{2-\alpha}  \le    G_R(x_0,x)\le C_2 r^{2-\alpha},    \ \ x \in  \partial B_r \cup \partial^i B_r.
\end{align}

\item Gaussian heat kernel bounds $\mathrm{HKC}(\alpha,2)$. (In fact, one has equivalence.)
These are defined in terms of the natural graph metric $d_0$ and continuous time heat kernel $q_t(x,y)= \P_x(Y_t=y)/\deg(y)$, where $Y_t$ is the continuous time simple random walk on $\Gamma$, as
\begin{equation}\label{eqn:gaussian-heatkernel}
\frac{c_3}{t^{\alpha/2}}\exp\left(-c_4d_0(x,y)^2/t\right)\le q_t(x,y)\le \frac{c_1}{t^{\alpha/2}}\exp\left(-c_2d_0(x,y)^2/t\right), 
\end{equation}
for $t\ge \max(1,d_0(x,y))$.
\end{enumerate}
\end{proposition}

Proofs and references for the above results can be found in the textbook \cite{barlow2017random}, particularly Theorems~6.19 and 7.18, and Lemma 4.21. We make two brief remarks concerning the proofs:
\begin{enumerate}
\item 
First, for property (i), Theorem~7.18 in \cite{barlow2017random} states the elliptic Harnack inequality (EHI) for  harmonic functions. 
The Moser--Harnack inequality \eqref{eqn:Moser-Harnack} for subharmonic functions can be derived from the harmonic version via standard elliptic PDE techniques. We include 
details for the discrete case in Appendix~\ref{sec:Moser-Harnack} for the reader's convenience.  

\item Second, while we do not directly work with the probabilistic Laplacian as in \cite{barlow2017random}, the results still follow since the Laplacian \eqref{eqn:laplacian} differs only by multiplication by the diagonal matrix of vertex degrees, which has bounded, nonzero entries by the bounded geometry and connectedness properties.
\end{enumerate}

The properties (i)--(iii) in Proposition~\ref{prop:piconseq} will be utilized in the landscape law proofs.
 We will see in Section~\ref{subsec:rough-isom} that Assumption~\ref{ass:0} and the weak Poincar\'e inequality, which are required for the  Landscape Law, will be preserved under rough isometry between graphs.
Since the required properties hold for $\Z^d$, then the landscape law Theorem~\ref{thm:LLaw} will hold for all graphs roughly isometric to $\Z^d$ as well.
While these examples have integer values of $\alpha=d$, Ref.~\cite{barlow-alphabeta} constructed graphs satisfying heat kernel bounds $\mathrm{HK}(\alpha,\beta)$ for any real $\alpha\ge1$ and $2\le \beta\le1+\alpha$.

In order to apply the Moser--Harnack inequality in the proof of Theorem~\ref{thm:LLaw}(ii), we will use the following corollary, which uses that the landscape function $u^A$ always satisfies $-\Delta u^A(y)\le 1$ for $y\in\Gamma$.
\begin{corollary}[Moser--Harnack for landscape]\label{cor:MH}
Let $r\ge1$. Suppose $\Gamma$ satisfies the exit time upper bound \eqref{eqn:u-upper}, and that $f\ge0$ satisfies $-\Delta f(y)\le 1$ for all $y\in B(x,2r)$. Then the Moser--Harnack inequality \eqref{eqn:Moser-Harnack} implies
 \begin{align}\label{eqn:MH-landscape0}
    \sum_{y\in   B(x,2r)}\,f(y)^2\ge  (2C_H)^{-1}| B(x,2r)|\left(  \sup_{ B(x,r)}f(y)^2- c_2r^4\right).
\end{align}
\end{corollary}
\begin{proof}
Let $B^*:=B(x,2r)$, and let $u_0^{B^*}$ be the free landscape function $u_0^{B^*}=-\Delta_{B^*}^{-1}\one_{B^*}$. Set $g:=f-u_0^{B^*}+\|u_0^{B^*}\|_\infty$, which then satisfies $g\ge0$ and for $x\in B^*$, $-\Delta g(x)\le0$. Applying the Moser--Harnack inequality and the exit time bound \eqref{eqn:u-upper} for $u_0^{B^*}$ then yields \eqref{eqn:MH-landscape0}.
\end{proof}

 In order to obtain the Lifshitz tail upper bound in  Theorem~\ref{thm:Lif}(ii)
for $N_u$, we will also require control on the lower bound of a ``harmonic weight'' on balls. 
Most of the time, we work on the natural balls given by the graph or chemical metric $d_0(x,y)$ (shortest-path distance) on $\Gamma$. 
However, for specific applications, it may be more convenient to consider other distance functions $d(x,y)$ which are (strongly) equivalent to $d_0$, in the sense that  $ c_1d_0(x,y)\le d(x,y)\le c_2 d_0(x,y)$ for some universal constants $c_1,c_2>0$ depending only on $\Gamma$. We thus state Assumption~\ref{ass:harmonic-weight} in terms of a more general metric function, since it may be easier to verify the required properties for a different metric.
We will denote by $B_R^d=B^d(\xi,R)$ the ball of radius of $R$ with respect to the metric $d$, and by $B_R=B_R^{d_0}$ the natural metric ball. 
Immediately, there is $c>0$ so that for any $\xi\in \V$ and $r\ge0$,  
\begin{align}\label{eqn:ball-eqv}
    B(\xi,r/c )\subset B^d(\xi,r)\subset B(\xi,c r). 
\end{align}

The mean value property of a harmonic function $f$ on $\R^d$ (on Euclidean balls) states that  $f(\xi)=\frac{1}{{\rm Vol}(B(\xi,R))}\int_{B(\xi,R)}f(y)dy.$ The integral weight for such an $\R^d$-harmonic function is thus the constant function $\frac{1}{{\rm Vol}(B(\xi,R))}$. A similar mean value property holds for harmonic functions on graphs, but with a general ``harmonic weight'' function (cf. \eqref{eqn:h-submean-ass}), which also depends on the metric on the graph. 
In general, the harmonic weight is not unique, and not necessarily a constant (with respect to the volume of the ball), as it is for $\R^d$. 

For the natural graph metric $d_0$, we have $\partial B(\xi,\rho)=B(\xi,\rho+1)\backslash B(\xi,\rho)$, which gives a filtered structure $  B(\xi,R)=\{\xi\}\cup \bigcup_{\rho=0}^{R-1}\partial B(\xi,\rho)$.  There is then the following natural way to define a harmonic weight and the volume average on $B(\xi,R)$, explicitly through the Poisson kernel or random walk hitting measure on each layer.  
Letting $Y_n$ be a discrete-time simple random walk and $\tau_A:=\min\{n\ge0:Y_n\not\in A\}$ the exit time from a region $A$, one can take a harmonic weight $h_{B(\xi,R)}:B(\xi,R)\to[0,1]$ for the ball $B(\xi,R)$ to be
\begin{align}\label{eqn:natural-harmonic}
h_{B(\xi,R)}(y) &= \frac{1}{R+1}\P_x[Y_{\tau_{B(\xi,r)}}=y; \,\tau_{B(\xi,r)}<\infty], \quad \text{where }r:=d_0(\xi,y).
\end{align}
This corresponds to values of the Poisson kernel of the ball $B(\xi,r)$, normalized by $R+1$ since we will average over $R+1$ layers. 
More precisely, the Poisson kernel $H_A:\overline{A}\times\partial A\to[0,1]$ of a region $A$ is defined as $H_A(x,y)=\P_x[Y_{\tau_A}=y]$. For $f$ harmonic in $A$ with boundary values $f(y)$, $y\in\partial A$, then
\begin{align}\label{eqn:poisson-exp}
f(x)&=\E_x[f(Y_{\tau_A})]=\sum_{y\in\partial A}H_A(x,y)f(y).
\end{align}
Thus for $f$ harmonic at points in $B(\xi,R-1)$, we have
\begin{align*}
f(\xi)
&=\frac{1}{R+1}\sum_{r=0}^{R}\sum_{y\in\partial B(\xi,r-1)}H_{B(\xi,r)}(\xi,y)f(y)
=\sum_{y\in B(\xi,R)} h_{B(\xi,R)}(y)f(y),
\end{align*}
where the $r=0$ term corresponds to just $y=\xi$ and $h_{B(\xi,0)}(\xi,\xi)=1/(R+1)$.
For further background and details, see Appendix~\ref{sec:poisson} and \cite[\S6.2]{lawler2010random}, \cite[Thm. 2.5]{barlow2017random}.

For a general metric $d$, the boundary of a ball $\partial B^d(x,r)$ may not relate nicely to larger balls. 
One can still introduce a harmonic weight in a similar spirit as the above, but which is technically more involved. 

We are most interested in the case when there is a harmonic weight  that is  uniformly bounded from below  on most parts of the ball. 
\begin{assumption}\label{ass:harmonic-weight}
There is a metric $d$ on $\Gamma$, strongly equivalent to the natural metric $d_0$, such that the following hold.
Given $R>0$,  $\xi\in\V$,  and $B^d_R(\xi)=B^d(\xi,R)$, there exists a function (the harmonic weight) 
 $h_{B_R^d(\xi)}(y): B^d_R(\xi)\to [0,1]$
satisfying
 \begin{itemize}
     \item Submean property: if $f$ is $\Delta$-subharmonic in the sense $-\Delta f \le0$ (on a set containing $B^d_R(\xi)\cup \partial B^d_R(\xi)$), then 
     \begin{align}\label{eqn:h-submean-ass}
         f(\xi)\le \sum_{y\in B_R^d(\xi)}h_{B_R^d(\xi)}(y)f(y),
     \end{align}
     and the equality holds if $\Delta f=0$. 
     \item  Uniform lower bound: 
     there is a constant $c>0$, depending only on $\Gamma$, and a `bad' subset $X_R(\xi)\subseteq B^d_R(\xi)$ such that 
    \begin{align}\label{eqn:h-lower-ass}
   h_{B_R^d(\xi)}(y)\ge \frac{c}{|B_R^d(\xi)|},\ \ y\in B^d_R(\xi)\backslash X_R(\xi), \qquad {\rm and} \qquad \lim_{R\to \infty} \frac{|X_R(\xi)|}{|B^d_R(\xi)|}=0.
    \end{align}
 \end{itemize}

\end{assumption}
From the preceding discussion, we see the continuous analogue of both properties holds immediately on $\R^d$ with $X_R=\varnothing$.  In \cite{arnold2022landscape},  a statement similar to \eqref{eqn:h-lower-ass} was obtained for $\Z^d$ cubes ($\ell^\infty$-balls), based on explicit formulas of the Green's function and Poisson kernel on cubes (see Lemma 4.3 of \cite{arnold2022landscape}). More generally, for the band graph $\Gamma_{d,W}$ \eqref{eqn:Gamma-dW}, we will prove in Section~\ref{sec:RBM} that such a harmonic weight exists with respect to the Euclidean metric on $\Z^d$, based on surface area control and  Poisson kernel estimates in e.g. \cite{lawler2010random}. 
We are curious about such properties on general graphs, in particular:
\begin{question}
    What properties of a graph $\Gamma$ guarantee the existence of a harmonic weight $h_{B_R^d}$ such that Assumption~\ref{ass:harmonic-weight} holds?
\end{question}

\subsection{Rough isometries and properties preserved under them}\label{subsec:rough-isom}

Here we finally formally define rough isometries between (unweighted) graphs and summarize several key properties preserved under them. We refer readers to \cite{barlow2017random} for further details. 
\begin{defn}[rough isometry]\label{def:isometric}\mbox{}
\begin{itemize}
\item Let $(X_1,d_1)$ and $(X_2,d_2)$ be metric spaces. A map $\varphi:X_1\to X_2$ is a \emph{rough isometry} if there exists constants $C_1,C_2$ such that
\begin{align*}
C_1^{-1}\big(d_1(x,y)-C_2\big) &\le d_2(\varphi(x),\varphi(y)) \le C_1\big(d_1(x,y)+C_2\big),\\
&\bigcup_{x\in X_1}B_{d_2}(\varphi(x),C_2)=X_2.
\end{align*}
If there exists a rough isometry between two spaces then they are \emph{roughly isometric}, and this is an equivalence relation.

\item Let $\Gamma_1$ and $\Gamma_2$ be connected graphs whose vertices have uniformly bounded degrees. A map $\varphi:\mathbb{V}_1\to\mathbb{V}_2$ is a \emph{rough isometry} if:
\begin{enumerate}
\item $\varphi$ is a rough isometry between the metric spaces $(\mathbb{V}_1,d_{\Gamma_1})$ and $(\mathbb{V}_2,d_{\Gamma_2})$ with constants $C_1$ and $C_2$.
\item there exists $C_3<\infty$ such that for all $x\in\mathbb{V}_1$,
\begin{equation*}
C_3^{-1}\operatorname{deg}(x) \le \operatorname{deg}(\varphi(x)) \le C_3\operatorname{deg}(x).
\end{equation*}
\end{enumerate}
Two graphs are \emph{roughly isometric} if there is a rough isometry between them, and this is an equivalence relation.
\end{itemize}
\end{defn}
\begin{ex}\label{ex:band-graph}
The band graph  $\Gamma_{d,W}$   in \eqref{eqn:Gamma-dW} is roughly isometric to $\Z^d$, see e.g. \cite[\S3.2.1]{lratio}.   
\end{ex}

\begin{ex}\label{ex:Penrose}
We revisit the Penrose tiling vertex graph $\Gamma_P$ from Section~\ref{subsec:appl-isom}, which we now view as embedded in $\R^2$ with edges of length 1, to provide details demonstrating it is roughly isometric to $\Z^2$.
Like in \cite{Telcs}, which considered another graph derived from the Penrose tiling (the tile graph rather than vertex graph), we start by comparing $\Gamma_P$ to $\varepsilon\Z^2$ for a small $\varepsilon>0$.
Letting $\psi:\Gamma_P\to\varepsilon\Z^2$ be the map defined by taking $x\in\Gamma_P\subset\R^2$ to a closest point $\psi(x)\in\varepsilon\Z^2$, we obtain for $x\ne y$ and sufficiently small $\varepsilon$ (which ensures that $\psi$ is injective),
\begin{align*}
\sqrt{2}d_{\Gamma_P}(x,y)\ge \sqrt{2}\|x-y\|_2
\ge\|x-y\|_1
&\ge \|\psi(x)-\psi(y)\|_1-\|x-\psi(x)\|_1-\|y-\psi(y)\|_1\\
&\ge c\|\psi(x)-\psi(y)\|_1,
\end{align*}
since $\|x-\psi(x)\|_1$ and $\|y-\psi(y)\|_1$ are both of order $\varepsilon$, while $\|\psi(x)-\psi(y)\|_1\ge c_0-\mathcal{O}(\varepsilon)$ is lower-bounded using the minimum distance between two corners of the rhombi.
Then defining $\varphi:\Gamma_P\to\Z^2$ via $\varphi(x)=\varepsilon^{-1}\psi(x)$, we obtain for $x\ne y$ in $\Gamma_P$,
\begin{align*}
\sqrt{2}{d_{\Gamma_P}}(x,y)&\ge c\varepsilon\|\varphi(x)-\varphi(y)\|_1=c\varepsilon d_{\Z^2}(\varphi(x),\varphi(y)).
\end{align*}

Additionally, $d_{\Gamma_P}(x,y)\le C_1 d_{\Z^2}(\varphi(x),\varphi(y))+C_2$ for sufficiently small $\varepsilon>0$: One can consider the set of rhombi $\mathcal{R}$ in the tiling that intersect the straight line segment $L$ between $x$ and $y$ (including intersections on edges and corners). This set $\mathcal{R}$ allows for a path in $\Gamma_P$ between $x$ and $y$ of length at most twice the cardinality of $\mathcal{R}$.
The number $\#\mathcal{R}$ of such rhombi scales with the length $|L| = \|x-y\|_2$ by area considerations (take for example the rectangle around $L$ of five units in each direction, which covers $\mathcal{R}$), and so $d_{\Gamma_P}(x,y)\le C_{1,\varepsilon} d_{\Z^2}(\varphi(x),\varphi(y))+C_{2,\varepsilon}$. 

Finally, there is a numerical constant $C$ (based on the maximum distance between corners of rhombi) so that
$\bigcup_{x\in\Gamma_P}B_{\Z^2}(\varphi(x),C\varepsilon^{-1})=\Z^2.$
\end{ex}

As is readily seen, rough isometry preserves bounded geometry \eqref{eqn:bded-geo} and (as can be seen with more work) volume control \eqref{eqn:vol-control} (cf. \cite[Exercise 4.16]{barlow2017random}), so that Assumption~\ref{ass:0} is preserved under rough isometry.
The next proposition states that Proposition~\ref{prop:covering}
and the weak Poincar\'e inequality are also preserved under rough isometry.

\begin{proposition}\label{prop:isom}
Assume volume control \eqref{eqn:vol-control}. The following properties are preserved under rough isometries.
\begin{enumerate}[(i)]
\item Finite overlap property in Proposition~\ref{prop:covering}(a) 
(with a possibly different constant $b_\Gamma$, depending on the rough isometry constants).
\item Translation-type property in Proposition~\ref{prop:covering}(b).
\item The weak Poincar\'e inequality (Definition~\ref{ass:WPI}).
\item Gaussian heat kernel estimates $\mathrm{HKC}(\alpha,2)$ as in \eqref{eqn:gaussian-heatkernel}.
\end{enumerate}
\end{proposition}
Note that Proposition~\ref{prop:isom}(iii) combined with Proposition~\ref{prop:piconseq} implies that if a graph $\Gamma$ with volume control satisfies the weak Poincar\'e inequality, then it and any graph roughly isometric to it also satisfy the Moser--Harnack inequality and exit time upper bound \eqref{eqn:u-upper}.
For the proofs of Proposition~\ref{prop:isom}(iii,iv), see \cite[Thms. 3.33, 6.19]{barlow2017random}. 
Parts (i) and (ii) follow automatically from Proposition~\ref{prop:covering} since volume control \eqref{eqn:vol-control} is preserved by rough isometry.

\section{Proof of the Landscape Law for graphs and random hopping models}\label{sec:llaw} 
In this section, we prove the Landscape Law for graphs and Jacobi/random hopping models as stated in Theorem~\ref{thm:LLaw}. In the upper bound, the main differences from the $\R^d$ or $\Z^d$ case are that we repeatedly use the finite covering property in Proposition~\ref{prop:covering}(a) to make up for not having a clean partition into cubes, and that since the bond weights $\mu_{xy}$ can become arbitrarily close to (or equal to) zero, we must separately truncate and bound these small bond weights using leftover diagonal terms (Lemma~\ref{lem:ep-cut}). We must also consider boundary terms coming from inner and outer boundaries of graph balls, as well as the relation to the non-regular shape of $A$, carefully throughout.

In the lower bound, the main difference from $\R^d$ or $\Z^d$ is to utilize 
Proposition~\ref{prop:covering}(b) to make up for not having a clean partition or translation invariance. Combined with Proposition~\ref{prop:covering}(a) and scaling in Lemma~\ref{lem:partition-comp}, this will allow us to handle comparisons with overlapping covers.
Additionally, by allowing for  overlapping covers and graph-dependent constants, the proof we give actually provides a simpler proof of the $\Z^d$ case from \cite{arnold2022landscape}.

\subsection{Proof of the Landscape Law upper bound, Theorem~\ref{thm:LLaw}(i) }\label{sec:LL-upper-det}

Let $r(E)=E^{-1/2}$ be the covering radius for the partition, and define the set
\begin{align*}
    {\mathcal F}=\left\{\, B\in \mathcal P^A(r(E))\,  :\ \ \min_{x\in B}\frac{1}{u(x)}\,\le E   \right\},
\end{align*}
so that the landscape  landscape counting function \eqref{eqn:Nu} is 
$N_u(E)=\frac{\#\mathcal F}{|A|}$.

{\bf \noindent Case I:} We first consider $r(E)=E^{-1/2}\ge1$. The other case $r(E)=E^{-1/2}<1$ (large $E$) corresponds to balls consisting only of a single point, and follows immediately from the landscape uncertainty principle as described in Case II near the end of this subsection.

Let 
$$S=\left\{f\in \ell^2(A)\, :\, \bar f^{B }\equiv \frac{1}{\deg(B)}\sum_{x\in B}f(x)\deg(x)
=0,\ \ B\in {\mathcal F} \right\},$$
where $\bar f^{B }$ is the weighted average of $f$ (w.r.t.  the natural weight) on $B$ and $\deg(x)$ is the degree in the graph $\Gamma$, as in the weak PI \eqref{eqn:bar-f}. 
We will show that for $f\in S$, that $\langle f,H^Af\rangle\ge C_1E\langle f,f\rangle$, so that $N(C_1E)\le \operatorname{codim}S$. 
For any set $B$, let $v_B\in\C(\V)$ be the vector with elements $\deg(x)$ for $x\in B$ and 0 otherwise. Then the space $S$ is the orthogonal complement of $\operatorname{span}(\{v_B:B\in\mathcal{F}\})$.
The balls $B$ in $\mathcal{F}$ may not be disjoint, but this does not matter, since we just need an upper bound on $\operatorname{codim}(S)=\operatorname{dim}(\operatorname{span}(\{v_B:B\in\mathcal{F}\}))\le \#\mathcal{F}$.
 Hence, 
\[ {\rm codim}S\le \#{\mathcal F}, \]
and the main work is to bound $\ipc{f}{H^Af}$ from below for $f\in S$, which will lead to an upper bound on the number of eigenvalues below an energy $E$ in terms of ${\rm codim}S$.

In order to obtain a bound like $C_1E\langle f,f\rangle\le \langle f,H^Af\rangle$, we consider coordinates $x\in A$ in the sum $\langle f,f\rangle=\sum_{x\in A}f(x)^2$ according to whether they are in a ball $B\in\mathcal{P}^A\setminus\mathcal{F}$ or in a ball $B\in\mathcal{F}$. Coordinates $x\in A$ may be in both types of balls, but due to the finite covering property in Proposition~\ref{prop:covering},
such overcounting is allowable. 
We start with balls $B\in\mathcal{P}^A\setminus{\mathcal F}$. In this case, the property $\min_{x\in B} \frac{1}{u(x)} \ge E$ implies that 
\begin{align}
  E \sum_{B\not\in {\mathcal F}}\sum_{x\in B\cap A}\, f(x)^2  
 &\le \sum_{B\not\in {\mathcal F}}\sum_{x\in B\cap A}\, \frac{1}{u(x)}f(x)^2
 \le   b_\Gamma \sum_{x\in   A}\, \frac{1}{u(x)}f(x)^2 
 \le b_\Gamma\ipc{f}{H^Af}, \label{eqn:B-not-in-F}
\end{align}
where we used the finite covering property \eqref{eqn:finite-cover-lambda} followed by the landscape uncertainty principle \eqref{eqn:UCP-gen} for the last two inequalities. 

For $B\in {\mathcal F}$, we first apply the weak Poincar\'e  inequality \eqref{eqn:WPI} to $f\in S$ with $\bar f^{B}=0$ to obtain,
\begin{align}\label{eqn:PI-pf}
   \sum_{x\in B} f(x)^2= \sum_{x\in B}\big(f(x)-\bar f^{B }\big)^2 \le  
 C_P\, r^2\sum_{x,y\in B^\ast: x\sim y}\big(f(x)-f(y)\big)^2,
    \end{align}
    where  $B=B(z_i,r)$ and $B^\ast=B(z_i,\lambda r)$ are given as in \eqref{eqn:WPI}. 
It remains to
 bound the right hand side of \eqref{eqn:PI-pf} from above by 
 $\ipc{f}{H^Af}$. Note that the kinetic energy term (non-potential term) in  the Hamiltonian $H^A$ defined in \eqref{eqn:HA} has a weight $\mu_{xy}\in [0,1]$, which may take the degenerate value zero. 
If we had $\mu_{xy}\ge \eps>0$ for all $\mu_{xy}$, then \eqref{eqn:PI-pf} is readily bounded using the kinetic term in \eqref{eqn:HA} and the finite-overlap property. In the general case where $\mu_{xy}$ takes values arbitrarily close (or equal to) $0$, we first must truncate the weight $\mu_{xy}$ and compare the resulting kinetic energy to an additional diagonal term which can be eventually absorbed into $\ipc{f}{H^Af}$. 
The following truncation lemma allows us to compare the kinetic term  in $\langle f,H^Af\rangle$, which may have degenerate weights $\mu_{xy}$, to the right hand side of \eqref{eqn:PI-pf}. The proof of the lemma is given at the end of this section. 
\begin{lemma}[$\varepsilon$-cutoff]\label{lem:ep-cut}
For $\mu_{xy}\in[0,1]$ and any $\varepsilon>0$, let $\mu_{xy}^\varepsilon =\max\{\varepsilon,\mu_{xy}\}$. Then for any region $R\subseteq\Gamma$,
\begin{equation}\label{eqn:eps-cutoff}
\frac{1}{2}\sum_{\substack{x,y\in R\\x\sim y}}\mu_{xy}^\varepsilon(f(x)-f(y))^2\le \frac{1}{2}\sum_{\substack{x,y\in R\\x\sim y}}\mu_{xy}(f(x)-f(y))^2+4\varepsilon(1-\varepsilon)^{-1}\sum_{x\in R}(\deg^Rx-\mu_x^R)f(x)^2.
\end{equation}
As a consequence, by choosing $\varepsilon=\frac{1}{5}$,
\begin{align}
\sum_{\substack{x,y\in R\\x\sim y}}(f(x)-f(y))^2\le 10\bigg(\frac{1}{2}\sum_{\substack{x,y\in R\\x\sim y}}\mu_{xy}(f(x)-f(y))^2+\sum_{x\in R}(\deg^R x-\mu_x^R)f(x)^2\bigg).\label{eqn:5}
\end{align}
\end{lemma}
Applying the above equation \eqref{eqn:5} to $B^*$ in \eqref{eqn:PI-pf}, and recalling that $r^2=E^{-1}$, we obtain for $f\in S$,
\begin{align}\label{eqn:F-bound2}
\sum_{x\in B}f(x)^2&\le 10C_PE^{-1}\bigg(\frac{1}{2}\sum_{\substack{x,y\in B^*\\x\sim y}}\mu_{xy}(f(x)-f(y))^2+\sum_{x\in B^*}(\deg^{B^*} x-\mu_x^{B^*})f(x)^2\bigg).
\end{align}
For showing the right side of \eqref{eqn:F-bound2} is bounded by a factor of $\langle f,H^Af\rangle$, one complication is that $B^*$ may contain points outside $A$. This will be handled using the $\deg(x)-\mu_x^A$ term in $\langle f,H^Af\rangle$ as follows.
Using that $f$ is zero outside $A$ and that $\mu_{xy}=\mu_{yx}$, we can estimate
\begin{align*}
\frac{1}{2}&\sum_{\substack{x,y\in B^*\\x\sim y}}\mu_{xy}(f(x)-f(y))^2+\sum_{x\in B^*}(\deg^{B^*}x-\mu_x^{B^*})f(x)^2\\
&\le \sum_{x\in B^*\cap A}(\deg^{B^*}x-\mu_x^{B^*})f(x)^2+\frac{1}{2}\sum_{x\in B^*\cap A}\sum_{\substack{y\in A\\y\sim x}}\mu_{xy}(f(x)-f(y))^2
+\sum_{x\in B^*\cap A}\sum_{\substack{y\in B^*,y\not\in A\\y\sim x}}\mu_{xy}f(x)^2\\
&=\sum_{x\in B^*\cap A}(\deg^{B^*}x-
\mu_x^{B^*\cap A})f(x)^2+\frac{1}{2}\sum_{x\in B^*\cap A}\sum_{\substack{y\in A\\y\sim x}}\mu_{xy}(f(x)-f(y))^2.\numberthis\label{eqn:nextbound}
\end{align*}
The diagonal term can be bounded using $\deg^{B^*}x-\mu_x^{B^*\cap A}\le \deg(x)-\mu_x^A$, and the kinetic term using that
by \eqref{eqn:finite-cover-lambda}, for any function $g\ge0$ on $A$,
\begin{align}
\nonumber\sum_{B\in\mathcal{P}^A}\sum_{x\in B^*\cap A}g(x)
=\sum_{x\in A}g(x)\, \#\{B\in\mathcal{P}^A:x\in B^*\}
&\le C\lambda^\alpha b_\Gamma\sum_{x\in A}g(x).
\end{align}
Applying these bounds in \eqref{eqn:nextbound} and summing over $B\in\mathcal{F}$ then yields
\begin{multline*}
\sum_{B\in\mathcal{F}}\bigg[\frac{1}{2}\sum_{\substack{x,y\in B^*\\x\sim y}}\mu_{xy}(f(x)-f(y))^2+\sum_{x\in B^*}(\deg^{B^*}x-\mu_x^{B^*})f(x)^2\bigg] \\
\le C\lambda^\alpha b_\Gamma \sum_{x\in A}\bigg[\frac{1}{2}\sum_{\substack{y\in A\\y\sim x}}\mu_{xy}(f(x)-f(y))^2 +(\deg(x)-\mu_x^A)f(x)^2\bigg].
\end{multline*}
As the latter is exactly $C\lambda^\alpha b_\Gamma\langle f,H^Af\rangle$, combining with
 \eqref{eqn:B-not-in-F} and \eqref{eqn:F-bound2} then yields
\begin{align*}
    E\sum_{x\in A} f(x)^2\le E \sum_{B \in {\mathcal P^A}}\sum_{x\in B}f(x)^2  \le C_1\ipc{f}{H^Af},
\end{align*}
for $C_1=(10C_PC\lambda^\alpha+1) b_\Gamma$. 
Therefore, the number of eigenvalues of $H^A$ below $C_1^{-1}E$ is bounded from above by the codimension of the subspace $S$, that is, for $E\le 1$,
\begin{align} \label{eqn:upper-unscaled}
    N(C_1^{-1}E)\le \frac{\# {\mathcal F}}{|A|}=N_u(E).
\end{align}
By rescaling the energy as $\wt E=C_1E\le 1$, we obtain that for $E\le C_1^{-1}$,
 \begin{align} \label{eqn:upper-scaled}
    N(E)\le N_u(C_1E).  
\end{align}

{\bf \noindent Case II: $r(E)=E^{-1/2}<1$. }   
In this case, we just have
\begin{align*}
 {\mathcal F}=\left\{\, x\in A\,  :\ \  \frac{1}{u(x)}\,\le E   \right\},
\text{ and } S=\left\{f\in \ell^2(A)\, :\,  f(x) =0,\ \ {\rm if}\ \ x\in \mathcal F  \right\}.
\end{align*}
For $f\in S$, the landscape uncertainty principle \eqref{eqn:UCP-gen} then implies
\begin{align*}
  \ipc{f}{H^Af}\ge \sum_{x\in A}\frac{f(x)^2}{u(x)} &= \sum_{x\not\in\mathcal{F}}\frac{f^2(x)}{u(x)}
   \ge E \sum_{x\in A} f(x)^2,
\end{align*}
which completes the proof of Theorem~\ref{thm:LLaw}(i).
\qed

\begin{proof}[Proof of Lemma~\ref{lem:ep-cut} ($\varepsilon$-cutoff)]
Start by writing
\begin{align*}
\frac{1}{2}\sum_{\substack{x,y\in R\\x\sim y}}\mu_{xy}^\varepsilon(f(x)-f(y))^2 &=\frac{1}{2}\sum_{\substack{x,y\in R\\\mu_{xy}\ge\varepsilon}}\mu_{xy}(f(x)-f(y))^2 +\frac{1}{2}\sum_{\substack{x,y\in R\\x\sim y\\\mu_{xy}<\varepsilon}}\varepsilon(f(x)-f(y))^2.
\end{align*}
We upper bound the last term by comparing it to a sum involving $\deg^R x-\mu_x^R$. Since $\mu_{xy}=\mu_{yx}$, then
\begin{align*}
\sum_{\substack{x,y\in R\\x\sim y\\\mu_{xy}<\varepsilon}}\varepsilon(f(x)-f(y))^2 &\le 2\varepsilon\sum_{\substack{x,y\in R\\\mu_{xy}<\varepsilon}}(f(x)^2+f(y)^2)\\
&=4\varepsilon\sum_{x\in R}f(x)^2\cdot\#\{y\in R:\mu_{xy}<\varepsilon\}.
\end{align*}
If $\#\{y\in R:\mu_{xy}<\varepsilon\}$ is large, then so is $\deg^R x-\mu_x^R$: Using $\mu_x^R=\sum_{y\in R}\mu_{xy}$, then 
\begin{align*}
\mu_x^R=\sum_{y\in R}\mu_{xy} &\le \sum_{\substack{y\in R:\mu_{xy}<\varepsilon}}\varepsilon+\sum_{y\in R:\mu_{xy}\ge\varepsilon}1\\
&=\varepsilon\#\{y\in R:\mu_{xy}< \varepsilon\} + \deg^R x-\#\{y\in R:\mu_{xy}< \varepsilon\},
\end{align*}
and so
\begin{align*}
(1-\varepsilon)\#\{y\in R:\mu_{xy}< \varepsilon\} &\le \deg^R x-\mu_x^R.
\end{align*}
Thus
\begin{align*}
\sum_{\substack{x,y\in R\\\mu_{xy}<\varepsilon}}\varepsilon(f(x)-f(y))^2 &\le 4\varepsilon(1-\varepsilon)^{-1}\sum_{x\in R}(\deg^R x-\mu_x^R)f(x)^2,
\end{align*}
which gives \eqref{eqn:eps-cutoff}.
\end{proof}

\subsection{Proof of the Landscape Law lower bound, Theorem~\ref{thm:LLaw}(ii)}\label{sec:pf-lowerLLaw}

Similar to the landscape law upper bound, if we can bound $\ipc{f}{H^Af}$ from above by $E$ on some subspace $\mathcal{S}'\subset\ell^2(A)$, then the eigenvalue counting function at $E$ will be at least the dimension of $\mathcal{S}'$. In what follows, we will consider $0<\kappa<1/4$.

\vspace{2mm}
\noindent\textbf{Case I:} First we consider {$E\le c^2$} for $0<c<1$  to be determined later.
In view of the constants inside the argument of $N_u$ in \eqref{eqn:LLaw-lower1}, we pre-preemptively take the partition radius $R=(c\kappa^{-1})E^{-1/2}$. 
As we  will be using the Moser--Harnack inequality \eqref{eqn:MH-landscape0}, we will need to work with smaller balls, say of radius $r=\kappa R$ for $0<\kappa<1/4$. 
With the condition $E<c^2$, then $r\ge 1$ and $R\ge 4$.

For a ball $B=B(z,R)$, denote by $\check{B}:=B(z,r)$ the smaller ball with the same center $z$, and set
 \begin{align*}
  {\mathcal F}'=\left\{\, B\in \mathcal P(R)\,  :\ \ \min_{x\in  \check B}\frac{1}{u(x)}\,\le \,E 
 \quad {\rm and} \quad \min_{x\in   B}\frac{1}{u(x)}\,\ge \,  \kappa^2  E
 \right\}.
 \end{align*} 
First we define a preliminary set $\mathcal{S}_0$ as
\begin{align*}
   \mathcal{S}_0={\rm span}\left\{u {\chi}_{B}: B=B(z,R)\in{\mathcal F}'\right\},
\end{align*}
where $ {\chi}_{B(z,R)}$ denotes a cut-off function supported on $B(z,R)$ defined via
\begin{align}\label{eqn:cut-off-2}
 {\chi}_{B(z,R)}(x) &= \begin{cases}
1,&x\in B(z,R/2)\\
0,&   x\not\in B(z,R) \\
1-\frac{2}{R}(d(z,x)-R/2),&R/2\le d(z,x)\le R 
\end{cases}.
\end{align}

Note that for adjacent $x\sim y\in B(z,R/2)$, that
\begin{align*}
| {\chi}_{B(z,R)}(x)- {\chi}_{B(z,R)}(y)| \le \frac{4}{R}.
\end{align*}

The functions in $\mathcal{S}_0$ need not be orthogonal or linearly independent due to the overlap allowed in the covering. To remedy this, we construct $\mathcal{F}''$ and $\mathcal{S}'$ as follows:
To choose a set of disjoint balls $\mathcal{F}''=\{B(z,R)\}$ from $\mathcal{F}'$, go through the balls in $\mathcal{F}'$, and for each ball $B(z,R)$ still remaining, remove all other balls that overlap with $B(z,R)$. Let $\mathcal{F}''$ be the set of remaining balls.  If two balls overlap, then their centers satisfy $d(z,z_i)\le 2R$, and so by the finite cover property \eqref{eqn:finite-cover-lambda},
there can only be $C_\Gamma b_\Gamma$ balls that overlap $B(z,R)$. So for each ball we kept for $\mathcal{F}''$, we only removed at most $C_\Gamma b_\Gamma-1$ balls from $\mathcal{F}'$.
Since $\#\mathcal{F}'=\#\mathcal{F}''+\#\{\text{removed}\}\le \#\mathcal{F}''C_\Gamma b_\Gamma$, then we must have 
\[
\#\mathcal{F}''\ge \frac{1}{C_\Gamma b_\Gamma}\#\mathcal{F}'.
\] 
Now take
\begin{align*}
   \mathcal{S}'={\rm span}\left\{u {\chi}_{B}: B=B(z,R)\in{\mathcal F}''\right\}.
\end{align*}
Since such $u {\chi}_B$ have disjoint supports, they are linearly independent and $\operatorname{dim}\mathcal{S}'=\#\mathcal{F}''$. Additionally, their supports are separated from each other by distance at least $R\ge4$, so that $\langle u {\chi}_B,H^A(u {\chi}_{B'})\rangle=0$ for different balls $B,B'\in\mathcal{F}''$.

Now we want to bound $\dfrac{\langle u {\chi}_B,H^A(u {\chi}_B)\rangle}{\langle u {\chi}_B,u {\chi}_B\rangle}$ from above for each $u {\chi}_B\in\mathcal{S}'$.
First, by the landscape uncertainty principle \eqref{eqn:UCP-gen} and using that $\max_{x\in B}u(x)\le \kappa^{-2}E^{-1}$ by the definition of $\mathcal{F}'$, we can bound the numerator as
 \begin{align*} 
        \ipc{u {\chi}_B}{H^A(u {\chi}_B)}
   &=     \frac{1}{2}\sum_{\substack{x,y\in A\\x\sim y} } \mu_{xy}u(x)u(y)\left( {\chi}_B(x)- {\chi}_B(y)\right)^2 + \sum_{x\in A}u(x) {\chi}_B(x)^2\\
       & \le CM_\Gamma |B(z,R+1)|(\kappa^{-4}E^{-2})\left(\frac{2}{R}\right)^2+ |B(z,R)|(\kappa^{-2}E^{-1})\\
\numberthis &\le c_1R^{\alpha-2}\kappa^{-4}E^{-2}+c_1R^\alpha\kappa^{-2}E^{-1}. \label{eqn:uHu-upper}
    \end{align*}
For the denominator, using $ {\chi}_{B(z,R)}(x)=1$ for $x\in B(z,R/2)$ which includes $B(z,2r)$, followed by the Moser--Harnack inequality \eqref{eqn:MH-landscape0} yields for $r\ge1$,
\begin{align}\label{eqn:uchi-bound}
  \langle u {\chi}_B,u {\chi}_B\rangle &\ge \sum_{x\in B(z,2r)}u(x)^2
\ge c_H'| B(z,2r)|\left(\sup_{y\in B(z,r)}u(y)^2- c_2r^4\right).
\end{align}
Using that $ \max_{y\in B(z,r)}u (y)\ge E^{-1}$ by the definition of $\mathcal{F}'$ and that $r=\kappa R= cE^{-1/2}$, then \eqref{eqn:uchi-bound} becomes
\begin{align}
    \langle u {\chi}_B,u {\chi}_B\rangle \ge  C\kappa^{\alpha}R^{\alpha}E^{-2}(1-c_2c^4)
\ge c_3\kappa^\alpha R^\alpha E^{-2}, \label{eqn:uu-lower-gen}
\end{align}
provided that $c$ is chosen small enough that $c_2c^4<1$.

Combining \eqref{eqn:uu-lower-gen} and \eqref{eqn:uHu-upper} and using $R=\kappa^{-1}cE^{-1/2}$ then yields
\begin{align}
 \nonumber\frac{\ipc{u {\chi}_B}{H^A (u {\chi}_B)}}{\langle u {\chi}_B,u {\chi}_B\rangle}
 &\le c_1'\kappa^{-\alpha-4}R^{-2} +c_2'\kappa^{-\alpha-2} E \\
&= C_2\kappa^{-\alpha-2}E.\label{eqn:scale1}
\end{align}
Since the $u {\chi}_B$ for $B\in\mathcal{F}''$ have disjoint supports separated from each other by distance at least $R\ge 4$, the bound \eqref{eqn:scale1} holds for all linear combinations $f\in\mathcal{S}'$, and we obtain 
\begin{align*}
N(C_2\kappa^{-\alpha-2}E) &\ge \frac{1}{|A|} \operatorname{dim}\mathcal{S}' \ge\frac{1}{|A|C_\Gamma b_\Gamma}\#\mathcal{F}'.
\end{align*}

Now to compare $\#\mathcal{F}'$ favorably to $N(E)$, we will need to make use of 
Proposition~\ref{prop:covering}(b). From this property, there are covers $\mathcal{P}^{(j)}(R)=\{B(z_i^{(j)},R)\}_i$, for $j=1,\ldots, C\kappa^{-\alpha}$, all with a finite covering constant $b_\Gamma$, and such that $\bigcup_j\bigcup_i B(z_i^{(j)},\kappa R)$ is also a cover (not necessarily with a particular covering constant), say $\tilde{P}(r)$.  Applying the preceding argument to any $\mathcal{P}^{(j)}(R)$, we have
\begin{multline}\label{eqn:N1N2}
N(C_2\kappa^{-\alpha-2}E)\ge\frac{1}{|A| C_\Gamma b_\Gamma}\Big(\#\{B(z_i^{(j)},R)\in\mathcal{P}^{(j)}:\min_{B(z_i^{(j)},r)}\frac{1}{u}\le E\}-\\
-\#\{B(z_i^{(j)},R):\min_{B(z_i^{(j)},R)}\frac{1}{u}\le \kappa^2 E\}\Big).
\end{multline}
The negative term $|A|^{-1}\#\{B(z_i^{(j)},R):\min_{B(z_i^{(j)},R)}\frac{1}{u}\le \kappa^2E\}$ is already $N^{\mathcal{P}^{(j)}(R)}_u(\kappa^2 E)$. 
For the first term, we apply \eqref{eqn:N1N2} to each cover $\mathcal{P}^{(j)}(R)$, $j=1,\ldots,C\kappa^{-\alpha}$, and take the sum, which results in the landscape counting function $N_u^{\tilde{P}(r)}(E)$ for the covering $\tilde{P}(r)$ with smaller radius $r=\kappa R=cE^{-1/2}$. With the summation, we obtain
\begin{align}\label{eqn:llaw-unscaled}
N(C_2\kappa^{-\alpha-2}E)\ge (C_\Gamma b_\Gamma)^{-1}C^{-1}\kappa^\alpha N^{\tilde{P}(r)}_u(E)-c_\Gamma N_u^{\mathcal{P}(R)}(\kappa^2E),
\end{align}
where the $N_u^{\mathcal{P}(R)}(\kappa^2E)$ term is for any partition $\mathcal{P}(R)$ with finite overlap constant $b_\Gamma$, and is comparable to $N_u(\kappa^2 E)$. While $\tilde{P}(r)$ may not have a particular covering constant, we only need to note that $N_u^{\tilde{P}(r)}$ upper bounds $N_u$ (up to the geometric constant in Lemma~\ref{lem:partition-comp}).

\vspace{2mm}
\noindent\textbf{Case II:} $E>c^2$.
In this case $r=cE^{-1/2}<1$, so that $\check B=B(z,r)=\{z\}$. Retaining the definitions of $\mathcal{F}'$ and $\mathcal{F}''$ from before, then for any $B(z,R)$ in $\mathcal{F}'$ or $\mathcal{F}''$, we have $u(z)^2\ge E^{-2}$. 
If $R=c\kappa^{-1}E^{-1/2}\ge 4$, corresponding to $c^2<E\le c^2(4\kappa)^{-2}$, {then the numerator upper bound \eqref{eqn:uHu-upper} still holds.} A lower bound  
\begin{align*}
\ipc{u {\chi}_B}{u {\chi}_B}\ge u(z)^2\ge E^{-2}
\end{align*}
is immediate, and so using $E> c^2$ we obtain
\[  
 \frac{\ipc{u {\chi}_B}{H^A u {\chi}_B}}{\ipc{u {\chi}_B}{u {\chi}_B}}
 \le     \frac{  c_1R^{\alpha-2}  \kappa^{-4}E^{-2} + c_1R^{\alpha}\kappa^{-2}E^{-1}}{ E^{-2}}\\
 \le   C_3\kappa^{-\alpha-2} E ,
 \]
which is the same scaling of $\kappa$ as in \eqref{eqn:scale1}, leading to the same form as \eqref{eqn:llaw-unscaled}.

If instead $0<R<4$, so that $E>c^2(4\kappa)^{-2}$, then we make an adjustment to $\mathcal{S}'$ to ensure the supports are far enough apart. Define $\mathcal{F}'$ as before, and form $\mathcal{F}''$ by going through each ball $B(z,R)$ in $\mathcal{F}'$, and removing all other balls whose centers $z_i$ are within distance $2$ of $z$. By volume control \eqref{eqn:vol-control}, there are at most $C_U2^\alpha$ points within distance $2$ of $z$, so 
\begin{align*}
\#\mathcal{F}'' &\ge \frac{1}{C_U2^\alpha}\,\#\mathcal{F}.
\end{align*} Then take
$
\mathcal{S}'=\operatorname{span}\{\one_z:z\in\mathcal{F}''\}.
$
Since $E>c^2(4\kappa)^{-2}$, the landscape uncertainty principle implies for any $z\in\mathcal{F}''$,
\begin{align*}
{\langle\one_z,H^A\one_z\rangle}\le M_\Gamma \kappa^{-4}+\kappa^{-2}E \le C_4\kappa^{-2}E,
\end{align*}
which leads to
\begin{align*}
N(C_4\kappa^{-2}E) &\ge \frac{1}{|A|}\operatorname{dim}\mathcal{S}'
\ge \frac{1}{C_U2^\alpha|A|}\#\mathcal{F}',
\end{align*}
and a similar lower bound as \eqref{eqn:llaw-unscaled}, except the bound is for $N(C_4\kappa^{-2}E)$.

\vspace{2mm}
\noindent\textbf{Rescaling:}
With new constants $c_1,c_2$ and applying Lemma~\ref{lem:partition-comp} (noting that the constants that arise from having $R=(c\kappa^{-1})E^{-1/2}$ can be taken independent of $\kappa$ since $\kappa<1/4$), Cases I and II in summary imply,
\begin{align*}
N(C_2\kappa^{-\alpha-2}E)&\ge c_1\kappa^\alpha N_u(E)- c_2N_u(\kappa^2E),\quad \text{for }E\le c^2,\\
N(C_3\kappa^{-\alpha-2}E)&\ge c_1\kappa^\alpha N_u(E)- c_2N_u(\kappa^2E),\quad \text{for }c^2<E\le c^2(4\kappa)^{-2},\\
N(C_4\kappa^{-2}E)&\ge c_1  N_u(E)- c_2N_u(\kappa^2E),\quad\text{for } E>c^2(4\kappa)^{-2}.
\end{align*}

Rescaling and using that $N(E)$ is (non-strictly) increasing in $E$, 
we obtain for a constant $c^*$, that
\begin{align*}
N(E)&\ge c_1\kappa^\alpha N_u(c_3\kappa^{\alpha+2}E)-c_2N_u(c_4\kappa^{\alpha+4}E),\quad\text{for }{E\le c^*\kappa^{-4}}, \\
N(E)&\ge c_1'  N_u(c_3'\kappa^{2}E)-c_2'N_u(c_4'\kappa^{4}E),\quad\text{for }{E>c^*\kappa^{-4}},
\end{align*}
which completes the proof of Theorem~\ref{thm:LLaw}(ii). \qed

Note that when $E$ is small, the above argument requires the domain to contain at least one small ball of radius $r\approx E^{-1/2}$, leading to the condition ${\rm diam}A \gtrsim E^{-1/2}$. The restrictions on ${\rm diam}A$ can however be removed easily. 
If ${\rm diam}A \lesssim E^{-1/2}$, then $\max_A u \lesssim ({\rm diam}A)^2 \lesssim E^{-1}$. 
Hence,   $N_u(cE)$ vanishes for some constant $c$ depending only on $\Gamma$,
leading to $N(E)\ge 0 = N_u(cE)$. In other words, the landscape law lower bound \eqref{eqn:LLaw-lower1} holds (neglecting the negative term on the right hand side) trivially when ${\rm diam}A$ is small.

\section{Lifshitz tails for the landscape counting function}\label{sec:lif}

In this section, we prove Theorem~\ref{thm:Lif} on Lifshitz tails for the landscape counting function for Jacobi operators on graphs. 
The process of establishing Lifshitz tails behavior of $N_u$ was done in \cite{DFM} for $\R^d$ and then extended to $\Z^d$ in \cite{arnold2022landscape}. 
In those settings, only cubes with periodic boundary conditions and diagonal disorder were considered for  simplicity. 
In the current paper, we work on a general domain $A$ with Dirichlet boundary conditions and consider both diagonal and off-diagonal disorder.  
The argument follows the general approach of the $\R^d$ and $\Z^d$ cases, though requires additional consideration for the general graph structure and shape of balls, lack of exact formulas, and metric in Assumption~\ref{ass:harmonic-weight}.
One example where these assumptions are satisfied is the graph induced by a random band model, which we will discuss in Section~\ref{sec:RBM}.

Recall we are interested in a Jacobi operator of the form \eqref{eqn:H-gen} where $\{\mu_{xy}\}$ and $\{V_x\}$ are each sets of i.i.d. random variables, and that such an operator can be written in the form
\begin{align*}
   H f(x) 
 &=  \sum_{y:y\sim x}\mu_{xy}\big(f(x)- f(y)\big) +( \sigma_x+V_x)f(x),
\end{align*}
where 
\begin{align*}
   \sigma_{xy}=1-\mu_{xy}, \quad \sigma_x=\sum_{y:y\sim x}\sigma_{xy}=\deg(x)-\mu_x.
\end{align*}
The corresponding quadratic form is 
 \begin{align*}
    \ipc{f}{Hf}=\frac{1}{2} \sum_{x,y:y\sim x}\mu_{xy}\big(f(x)-f(y)\big)^2+\sum_x( \sigma_x+V_x)f(x)^2.
\end{align*}  
For a finite $A\subset \Gamma$, the Dirichlet restriction of $H$ is 
\begin{align*}
    H^A f(x)=&\sum_{y\in A:y\sim x}\mu_{xy}\big(f(x)- f(y)\big) +( \sigma_{x}^A+V_x)f(x)
\end{align*}
and
 \begin{align}\label{eqn:quadr-HA-Lif}
    \ipc{f}{H^Af}=\frac{1}{2} \sum_{\substack{ x,y\in A\\ y\sim x}}\mu_{xy}\big(f(x)-f(y)\big)^2+\sum_{x\in A}( \sigma_{x}^A+V_x)f(x)^2,
\end{align}  
where $\sigma_x^A=\sum_{\substack{y\in A\\y\sim x}}\sigma_{xy}=\deg (x)-\mu_x^A$.

The Lifshitz tail lower bound proof follows the same method as for $\R^d$ or $\Z^d$, but accounts for the random hopping terms and uses the ``sufficient overlap with balls'' property to make up for not having a partition into cubes.
The upper bound proof will require more adaptation, in particular dealing with the different metrics in Assumption~\ref{ass:harmonic-weight}, and the lack of exact formulas for e.g. the Green's function and Poisson kernel which were utilized in the $\R^d$ and $\Z^d$ cases. 

\subsection{Lifshitz tails lower bound: proof of Theorem~\ref{thm:Lif}(i)}

For $E>0$, let $R=mE^{-1/2}$  for an $m$ to be specified later depending only on $\Gamma$, and  let $\mathcal P =\{B(z_i,R)\}_i$ be a cover of 
$\Gamma$. We will require $R\ge1$, i.e. $E\le m^2:=E_0$, so that we may later apply volume control \eqref{eqn:vol-control}.

From the landscape counting function definition, we have 
\begin{align}
  \E N_u^{\mathcal{P}}(E)=\frac{1}{|A|}\sum_{B\in\mathcal{P}|_A}\P\left[\min_{x\in B\cap A}\frac{1}{u(x)}\le E\right], \label{eqn:Nu-sum-P}
\end{align}
so to prove the lower bound \eqref{eqn:Lif-lower-Nu}, we will want to bound $\P\big[ \min_{x\in B\cap A}\frac{1}{u(x)}\le E \big]$ from below in terms of the CDF $F$.

Denote by $kB$ the scaled ball $B(z,kR)$, and 
let $0\le  {\chi}_B \le 1$ be the discrete cut-off function supported on $2B$ defined as
\begin{align*}
 {\chi}_{B}(x) &= \begin{cases}1,&x\in B\\
0,&x\not\in 2B\\
1-\frac{E^{1/2}}{m}(d(z,x)-mE^{-1/2}), & mE^{-1/2}\le d(z,x)\le 2mE^{-1/2}
\end{cases}.
\end{align*}
Note for $x\sim y\in 2B$, that $| {\chi}_B(x)- {\chi}_B(y)|\le \frac{E^{1/2}}{m}$.

Applying the landscape uncertainty principle \eqref{eqn:UCP-gen} along with \eqref{eqn:quadr-HA-Lif}  implies
 \begin{align} 
\nonumber\sum_{x\in 2B}\frac{ {\chi}_B^2(x)}{u(x)}\le\ipc{ {\chi}_B}{H^A {\chi}_B}   &\le  \frac{1}{2} \sum_{\substack{ x,y\in A\\ y\sim x} }  \big( {\chi}_B(x)- {\chi}_B(y)\big)^2+\sum_{x\in A}(\sigma_x^A+V_x) {\chi}_B(x)^2  \\
 &\le M_\Gamma|B(z,2R+1)| \frac{E}{m^2}+|2B|\big(M_\Gamma \max_{\substack{ x,y\in 3B\\ y\sim x}}\sigma_{xy}+\max_{x\in 2B}V_x\big).\label{eqn:liflower1}
    \end{align}
Since
\begin{align*}
    \min_{x\in B\cap A}\frac{1}{u(x)}\le \frac{1}{|B\cap A|}\sum_{x\in  B\cap A}\frac{1}{u(x)}\le &   \frac{1}{|B\cap A|}\sum_{ x\in 2B}\frac{ {\chi}_B^2(x)}{u(x)},
\end{align*}
we then obtain using \eqref{eqn:liflower1}, the volume control \eqref{eqn:vol-control}, and the condition \eqref{eqn:AintersectB} that $|A\cap B(x,R)|\gtrsim R^{\alpha}$ for $R\le C \operatorname{diam}A$ and $x\in A$, that
\begin{align*}
 \min_{x\in B\cap A}\frac{1}{u(x)} 
   \le  c_1\frac{E}{m^2}+c_2 \max_{\substack{ x,y\in 3B\\ y\sim x}}\sigma_{xy} +c_2\max_{x\in 2B} V_x.
\end{align*}
Choosing $m^2=2c_1$ and using independence of the random variables then yields
\begin{align*}
    \P\bigg[\min_{x\in B\cap A}\frac{1}{u(x)}\,\le E \bigg]&\ge 
      \P\bigg[ c_2 \max_{\substack{ x,y\in 3B\\ y\sim x}}\sigma_{xy} +c_2\max_{x\in 2B} V_x\,\le \frac{1}{2}E \bigg]\\ 
   &\ge     \P\bigg [ c_2 \max_{\substack{ x,y\in 3B\\ y\sim x}}\sigma_{xy}  \le \frac{1}{4}E \bigg] \P\bigg [ c_2\max_{x\in 2B} V_x\,\le \frac{1}{4}E \bigg]\\ 
      &\ge  \big(F_{\mu}(c_3E)\big)^{c_4E^{-\alpha/2}}\big(F_V(c_5E)\big)^{c_6E^{-\alpha/2}}.
\end{align*}
Thus with \eqref{eqn:Nu-sum-P}, we obtain 
\begin{align*}
   \E N^{\mathcal P}_u(E) 
   &\ge  \frac{\#\{B\in \mathcal P', B\cap A \neq \varnothing\}}{|A|}\big(F_{\mu}(c_3E)\big)^{c_4E^{-\alpha/2}}\big(F_V(c_5E)\big)^{c_6E^{-\alpha/2}}.
\end{align*}
Since $\mathcal{P}$ is a covering of $A$, we must have
\[  \#\{B\in \mathcal P, B\cap A \neq \varnothing\}\ge  \frac{|A|}{|B|}\ge   cmE^{\alpha/2}|A|, \]
which completes the proof of \eqref{eqn:Lif-lower-Nu}. \qed

\subsection{Lifshitz tails upper bound: proof of Theorem~\ref{thm:Lif}(ii)}

For $E>0$, set  $R=c_0E^{-1/2}$ for $c_0$ to be determined later (cf. Lemma~\ref{lem:PmaxF}), and set $\mathcal{P}=\{B(z_i,R)\}_i$ be a covering of $\Gamma$  satisfying \eqref{eqn:finite-cover-lambda} 
with a finite overlap constant $b_{\Gamma}$.
Using \eqref{eqn:Nu-sum-P}, then
\begin{align}\label{eqn:Nu-upper-Pmax}
   \E N_u^{\mathcal{P}}(E)\le \frac{\# \{B\in \mathcal P|_A\} }{|A|}\max_{B\in \mathcal P|_A}\P\left [\min_{x\in B\cap A}\frac{1}{u(x)}\le E \right],
\end{align}
and so we must bound $\P\big [\min_{x\in B\cap A}\frac{1}{u(x)}\le E \big]$ from above for each $B\in\mathcal{P}$.
This is achieved by the following lemma.
\begin{lemma}\label{lem:PmaxF}
 Under the assumptions in Theorem~\ref{thm:Lif}(ii), there is $0<E_\ast<1$ and $c_i>0$ such that for $E<E_\ast$ and any $x_0$,
\begin{align}\label{eqn:Pr-key1}
\P\left[\min_{x\in B(x_0,c_0E^{-1/2})} \frac{1}{u(x)}\le E\right] \le   c_1 \left(F (c_2E) \right)^{c_3 E^{-\alpha/2}}.
\end{align}
\end{lemma}
Assuming this lemma, we then have
\begin{proof}[Proof of Theorem~\ref{thm:Lif}(ii)]
By the finite covering assumption with constant $b_\Gamma$, and sufficient overlap of $A$ with balls of radius $R\le\operatorname{diam}A$, observe that
    \[\frac{\# \{B\in \mathcal P|_A\} }{|A|}\le cE^{\alpha/2}, \]
    and so \eqref{eqn:Lif-upper-Nu} follows immediately from \eqref{eqn:Nu-upper-Pmax} and \eqref{eqn:Pr-key1}. 
\end{proof}

The key ingredient to prove Lemma~\ref{lem:PmaxF} is the following deterministic result for the growth of the landscape function $u$. 
\begin{lemma}[landscape growth]\label{lem:u-growth-gen} 
Let $u=u^{A}$ be the landscape function for $H^A$. 
Choose $r>0$ and $x_0\in \Gamma$ with $B(x_0,r)\cap A\neq \varnothing$, 
  and 
let 
\begin{align}\label{eqn:def-J-gen}
    J_r=J_{r,x_0}=\big\{x\in B(x_0,r)\cap A:\ V_x\ge mr^{-2} \ \ {\rm or} \ \sum_{\substack{y\in B(x_0,r)\cap A\\  y\sim x}}\, (1-\mu_{xy})\ge m   r^{-2} \big \},
\end{align}
for a constant $m>0$ to be chosen later.
Then under the assumptions in Theorem~\ref{thm:Lif}(ii), for any $0<\lambda<1$,  there are $\eps<1$, $M,m,r_\ast>1$, depending only on $\alpha$ and $\lambda$, such that if the following two conditions hold for $r_\ast \le r\le 2\operatorname{diam}A$:
\begin{enumerate}[(i)]
\item there is $\xi \in B(x_0,r) $ such that 
\begin{align}\label{eqn:Mr2-gen}
    u(\xi) \ge Mr^2,
\end{align}
\item there is the lower bound on the size of $J_r$,
\begin{align}\label{eqn:Jr-lower}
   \big|J_r \big |\ge \lambda \big |B(x_0,r)\cap A \big |, 
\end{align}

\end{enumerate}
then for $R= ( {1+\eps})r $, there is $\xi'\in  B(\xi,CR)\cap A$ 
such that 
\begin{align}\label{eqn:MR2-gen}
    u(\xi') \ge (1+3\eps)\,  u(\xi)\ge  M {R}^2,
\end{align}
where $C=2 c^2 $ with the constant  $c>0$ given in  \eqref{eqn:ball-eqv}. 
\end{lemma}

One of the key ingredients of the proof of Lemma~\ref{lem:u-growth-gen} is that the landscape function $u$ is superharmonic. We will use   submean properties of $u$, together with conditions on $\mu_{xy},V_x$ in \eqref{eqn:Jr-lower}, to obtain the growth of $u$ on a larger ball. Recall that due to Assumption~\ref{ass:harmonic-weight}, we will work with assumptions on a general metric $d$ rather than on the natural metric $d_0$. Balls with respect to the metric $d$ will always be denoted with a superscript, such as $B_\rho^d$ or $B^d(\xi,\rho)$, while balls with respect to the natural metric $d_0$ will either not have the superscript or will be identified with a superscript $d_0$.
\begin{proof}[Proof of Lemma~\ref{lem:u-growth-gen}]
Let $\xi \in B(x_0,r)$  satisfy $u(\xi)\ge Mr^2$ as in \eqref{eqn:Mr2-gen}. 
Let the $d$-metric ball $B^d_R$ and a harmonic weight $h_{B_R^d}$ be given as in Assumption~\ref{ass:harmonic-weight}, satisfying \eqref{eqn:h-submean-ass} and \eqref{eqn:h-lower-ass}.

Before starting with the main cases of the proof, we introduce several useful quantities. We will need to consider weighted averages of the landscape function $u$ on balls and spheres, with respect to the metric $d$ and the associated harmonic weight $h_{B_R^d}$. To that end, denote by $a^d_\rho$ the weighted surface average of $u$ on the exterior boundary $\partial B^d(\xi,\rho)$,
\begin{align*}
  a^d_{\rho}=\sum_{y\in \partial B^d(\xi,\rho)}P_{B^d_\rho}(\xi,y)u(y),\ \ \rho\ge 0,
\end{align*}
where $P_{B^d_\rho}$
is the Poisson kernel (centered at $\xi$, with respect to $\Delta$) defined as in \eqref{eqn:Poisson-solution} on $\partial B^d(\xi,\rho)$. 

Next, denote by $A^d_{\rho}$ the weighted volume average of $u$ on $B^d(\xi,\rho)$, with respect to $h_{B_R^d}$, defined as 
\begin{align}
  A^d_{\rho} =   \sum_{y\in B^d(\xi,\rho)}h_{B^d_{\rho}}(\xi,y)u(y),
\label{eqn:Ar-0}
\end{align}
One can verify the following by taking a constant function in the integration by parts formula \eqref{eqn:ibp-green2},
\begin{align*}
\sum_{y\in\partial B^d(\xi,\rho)}P_{B^d_\rho}(\xi,y)=1,\quad \text{and }\sum_{y\in B^d(\xi,\rho )}h_{B^d_{\rho}}(\xi,y)= 1 .
\end{align*}
The key properties of $a^d_\rho$ and $A^d_\rho$ we will need are the following, which follow from submean properties of $\Delta$ using that $-\Delta u\le1$ pointwise on $B^d(x,\rho)$:
\begin{align}
a^d_\rho&\ge u(\xi)-c_0\rho^2,\label{eqn:ar}\\
A^d_\rho&\ge u(\xi)-c_0\rho^2.\label{eqn:Ar}
\end{align}
\begin{itemize}
\item For \eqref{eqn:ar}, let $u'$ be the solution to $-\Delta u'=1$ on $B^d(\xi,\rho)$ with Dirichlet boundary conditions on $\partial B^d(\xi,\rho)$. By Proposition~\ref{prop:piconseq} and equivalence of the metric $d$ with the natural metric, $u'$ satisfies $0\le u'(\xi)\le c_0\rho^2$. Let $w$ be harmonic with boundary data $w(y)=u(y)$ for $y\in\partial B^d(\xi,\rho)$. By \eqref{eqn:ibp-green} then $w(\xi)=a^d_\rho$. Set $f:=w+u'-u$, so  $-\Delta f\ge0$  and $f=0$ on the boundary $\partial B^d(\xi,\rho)$. Then $f=w+u'-u\ge0$  by the maximum principle, which yields
\begin{align*}
u(\xi)\le u'(\xi)+w(\xi)\le c_0\rho^2+a_\rho.
\end{align*}

 \item 
Let $u'$ be as above. Then $-\Delta(u-u')\le 0$. For \eqref{eqn:Ar}, using the submean property \eqref{eqn:h-submean-ass} and the non-negativity of $h_{B_\rho^d}$ and $u'$, we obtain
\begin{align*}
    u(\xi)-u'(\xi)\le \sum_{y\in B_\rho^d}h_{B_\rho^d}\big[u(y)-u'(y)\big]\le \sum_{y\in B_\rho^d}h_{B_\rho^d} u(y) =A^d_{\rho} ,
\end{align*}
which  yields \eqref{eqn:Ar}. 
\end{itemize}

 Given $\xi\in B(x_0,r)$, we will work on enlarged balls centered at $\xi$ containing $B(x_0,r)$ (with slightly larger radius under the different metric).  In particular, let 
$c>0$ be the scaling constant between the metric $d_0$ and $d$ in  \eqref{eqn:ball-eqv}, so that 
 \begin{align}\label{eqn:ball-relation}
     B(x_0,r)\subset B(\xi,2r)\subset B^d(\xi,2c r)\subset B(\xi,2c^{2}r). 
 \end{align}
Let $A^d_{2cr}$ be the weighted average of $u$ on the $d$-metric ball $B^d(\xi,2cr)$ of radius $2cr$, as given in the definition \eqref{eqn:Ar-0}. 
Then \eqref{eqn:Ar} yields 
\begin{align*}
     A^d_{2cr}\ge u(\xi)-c_1r^2. 
\end{align*}
 
Now let $J_r$ be given in \eqref{eqn:def-J-gen} 
where $m$ will be picked later, and satisfying \eqref{eqn:Jr-lower}. By the conditions $|B(z,r)\cap A|\ge c_2r^\alpha$ for $r<C\operatorname{diam}A$, \eqref{eqn:Jr-lower}, and $r\le C\operatorname{diam}A$, we have 
\begin{align}\label{eqn:Jr-lower1}
    |J_r|\ge  \lambda   \big |B(x_0,r)\cap A \big |\ge
    c_2\lambda   r^\alpha .
\end{align}

 Let 
\begin{align}\label{eqn:def-J0}
    I_r=\{x\in J_r: u(x)<\, \frac{1}{2}\, A_{2cr}^d\, \},
\end{align}
which describes $x\in J_r$ where $u(x)$ takes small values,
and let
\begin{align*}
    S_1=\sum_{y\in I_r}h_{B_{2cr}^d}(y),\quad S_2=\sum_{y\in B^d(\xi,2cr)\backslash I_r}h_{B_{2cr}^d}(y)=1-S_1.
\end{align*}

Now we are ready to look for $\xi'$ satisfying \eqref{eqn:MR2-gen} by considering the following two cases:

\vspace{1mm}
\noindent {\bf{Case I:} $  |I_r|\ge \frac{1}{2}|J_r|$. } In this case, we have many points $x$ where $u(x)$ is much smaller than the weighted average $A_{2cr}^d$. As a consequence, the remaining values $u(x)$ must be large to compensate for those in $I_r$.
More precisely, summing over the complement of $I_r$, the definition of $I_r$ implies
\begin{align*}
\frac{\sum_{y\in B^d(\xi,2cr)\setminus I_r}h_{B_{2cr}^d}(y)u(y)}{\sum_{y\in B^d(\xi,2cr)\setminus I_r}h_{B_{2cr}^d}(y)} &\ge \frac{1-\frac{1}{2}S_1}{1-S_1} A_{2cr}^d 
\ge   \left(1+ \frac{1}{2}S_1  \right) A_{2cr}^d.
\end{align*}
If $S_1\ge  4c_3$ for some $c_3>0$ (depends only on $\Gamma$ and $\lambda$), then we will be done,  
since then there is some $\xi'\in B^d(\xi,2cr)\setminus I_r$ with $u(\xi')\ge(1+2c_3)A_{2cr}^d$. Using \eqref{eqn:Ar} and $u(\xi)\ge Mr^2$ with a sufficiently large $M$ then implies
\begin{align*}
u(\xi')\ge (1+2c_3)\Big(u(\xi)-c_1r^2\Big) \ge (1+2c_3)\big(1-c_1/M\big)u(\xi)\ge (1+ c_3  )u(\xi),
\end{align*}
provided  $M\ge  (1+2c_3)c_1/c_3 $.
 
To show we have $S_1\ge4c_3$, first apply Assumption~\ref{ass:harmonic-weight} to obtain a $c_4>0$ and  subset $X_r\subset B^d(\xi,2cr)$ such that  
\begin{align}\label{eqn:h-lower}
h_{B_{2cr}^d}(y)\ge \frac{c_4}{|B^d(\xi,2cr)|},\ \ y\in B^d(\xi,2cr)\backslash X_r, \ \ {\rm and }\ \  \lim_{r\to \infty}\frac{|X_r|}{|B^d(\xi,2cr)|} = 0.
\end{align}
The lower bound on $|J_r|$ from \eqref{eqn:Jr-lower1} and of $h_{B_{2cr}^d}$ from \eqref{eqn:h-lower}  imply for sufficiently large $r>r_\ast(\lambda)$,
\begin{align*}
     S_1=\sum_{y\in I_r\backslash X_r}h_{B_{2cr}^d}(y)+\sum_{y\in I_r\cap X_r}h_{B_{2cr}^d}(y)\ge &\, \frac{c_4}{|B^d(\xi,2cr)|}|I_r\backslash X_r|\\
    \ge   &\, \frac{c_5}{r^{\alpha}}\Big(\frac{1}{2}|J_r|-o(|B^d(\xi,2cr)|) \Big)\\
   \ge  &\,\frac{c_5}{r^{\alpha}}\Big(\frac{c_2}{2}\lambda r^{\alpha}-\frac{c_2}{4}\lambda r^{\alpha} \Big)\\   
   \ge  &\, \frac{1}{4}c_2c_5\lambda:=4c_3.
\end{align*}
 Then it is enough to pick   $3\eps \le c_3=\frac{1}{16}c_2c_5\lambda $  so that \eqref{eqn:MR2-gen} holds for $\xi' \in B^d(\xi,2cr)\subset   B(\xi,2c^2r) \subset B(\xi,2c^2R)$. Since $u(x)=0$ for $x \not\in A $, it is clear that $\xi'\in A \cap B(\xi,2c^2R)$. 

\vspace{1mm}
\noindent {\bf{Case II:} $  |I_r|< \frac{1}{2}|J_r|$. } Take $R= (1+\eps) r  $ for some   $\eps\le 1$, which gives $R^2-r^2\le 3\eps r^2$. 
Let $a_{2r},a_{2R}$ and $G_{2r},G_{2R}$ be the surface averages of $u$ and the Green's functions on $B_{2r}=B(\xi,2r),B_{2R}=B(\xi,2R)$, with respect to the natural metric $d_0$ (we omit the superscript $d_0$ on the balls for simplicity), respectively. 
Applying the discrete Green's identity (integration by parts formula) \eqref{eqn:ibp-green} to both $B_{2r}$ and $B_{2R}$ yields
\begin{align}\label{eqn:SR-r}
    a_{2R}-a_{2r}&=\sum_{x\in   B_{2R}\backslash B_{2r}}G_{2R}(\xi,y)\Delta  u(x)+\sum_{x\in   B_{2r}}\, \left(G_{2R}(\xi,x)-G_{2r}(\xi,x)\right)\, \Delta  u(x).
\end{align}
With the relation  $\Delta u(x)=   -1 +\sum_{y\in A:y\sim x}\sigma_{xy}u(y)+V(x)u(x)$, and defining
\begin{align}\label{eqn:bx}
    b_x=\begin{cases}
        \sum_{y\in A:y\sim x}\sigma_{xy}u(y)+V(x)u(x), \  \ & x\in A \\
        1, \  \ & x\not\in A 
    \end{cases},
\end{align}
equation \eqref{eqn:SR-r} then implies
\begin{multline}
    a_{2R}-a_{2r} 
    \ge   -\sum_{x\in   B_{2R}\backslash B_{2r}}G_{2R}(\xi,x)- \sum_{x\in   B_{2r}} \Big(G_{2R}(\xi,x)-G_{2r}(\xi,x)\Big)+\\ +\sum_{x\in   B_{2r}} \Big(G_{2R}(\xi,x)-G_{2r}(\xi,x)\Big)b_x.\label{eqn:sR-sr}
\end{multline}

We next make two uses of the maximum principle for  harmonic functions (see e.g. \cite[Lemma~2.2]{lratio}).
\begin{itemize}
\item Since $G_{2R}(\xi,x)$ is $\Delta $-harmonic in $B_{2R}\backslash B_{2r}$, the maximum principle implies  
\begin{align*}
    \max_{y\in  B_{2R}\backslash B_{2r}}G_{2R}(\xi,x) \le \max_{x\in  \partial(B_{2R}\backslash B_{2r})} G_{2R}(\xi,x) = \max_{x\in  \partial^{(i)} B_{2r}} G_{2R}(\xi,x),
\end{align*}
with the last inequality because $G_{2R}(\xi,x)=0$ for $x\in\partial B_{2R}$.

\item On the other hand, $G_{2R}(\xi,x)-G_{2r}(\xi,x)$ is $\Delta $-harmonic in $B_{2r}$ and equal to $G_{2R}(\xi,\cdot)$ on $\partial B_{2r}(\xi)$, so the maximum principle implies for $x\in B_{2r}$,
\begin{align*}
 \min_{y\in \partial B_{2r}} G_{2R}(\xi,y)   \le\, G_{2R}(\xi,x)-G_{2r}(\xi,x) \le \max_{y\in \partial B_{2r}} G_{2R}(\xi,y).
\end{align*}
\end{itemize}
Next, by Proposition~\ref{prop:piconseq}, if $x\in\partial B_{2r}\cup\partial^{(i)}B_{2r}$, there are the Green's function bounds
 \begin{align*}
 c_1'r^{2-\alpha}\le     G_{2R}(\xi,x)\le c_2'r^{2-\alpha},
 \end{align*}
 where $c_1',c_2'$ only depend on $\alpha$ and the ratio $\frac{2R}{2r}=1+\eps$. 
Therefore, continuing from \eqref{eqn:sR-sr},
\begin{align*}
a_{2R}-a_{2r}\ge & -c_2' r^{2-\alpha}| B_{2R}\backslash B_{2r}|-c_2' r^{2-\alpha}|   B_{2r}|+ c_1' r^{2-\alpha}\sum_{x\in B_{2r}}b_x  \\
   \ge &  -c_3'  r^2\ + c_1' r^{2-\alpha}\bigg( \sum_{x \in  B_{2r}\cap A} \sum_{\substack{y\in B_{2r}\cap A:\\y\sim x}}\sigma_{xy}u(y)+ \sum_{x \in  B_{2r}\cap A}V_xu(x)\bigg) \\
    \ge &  -c_3'  r^2\ + c_1' r^{2-\alpha} \sum_{x \in  J_r\backslash I_r} \bigg(\sum_{\substack{y\in B(x_0,r)\cap A:\\y\sim x}}\sigma_{xy}+ V_x\bigg)u(x), \numberthis\label{eqn:178}
\end{align*}
where in the last line we used 
that $J_r\subset {B(x_0,r)}\cap A\subset B (\xi,2r)\cap A$. Recall that by the definition \eqref{eqn:def-J-gen} of $J_r$, for all $x\in J_r$, 
\[V_x+\sum_{y\in {B(x_0,r)}\cap A:  y\sim x}\, (1-\mu_{xy})\ge m   r^{-2},  \]
and that by the definition of $I_r$ in \eqref{eqn:def-J0}, for all $x\in J_r\backslash I_r $,  $u(x)\ge \frac{1}{2}{A_{2cr}^d}$. Therefore, the last sum in \eqref{eqn:178} can be bounded from below as 
\begin{align*}
  \sum_{x \in  J_r\backslash I_r} \bigg(\sum_{y\in {B(x_0,r)}\cap A:y\sim x}\sigma_{xy}+ V_x\bigg)u(x)&\ge  mr^{-2}\Big(\frac{1}{2}{A_{2cr}^d}\Big) |J_r\backslash I_r|  \\
    &\ge   mr^{-2}\Big(\frac{1}{2}{A_{2cr}^d}\Big)\Big(\frac{1}{2}|J_r|\Big)  \ge \frac{1}{4} c_2mr^{\alpha-2}\lambda \, {A_{2cr}^d},
\end{align*}
where we have now made use of the Case II condition that $|I_r|<|J_r|/2$, and also applied the bound $|J_r|\ge c_2\lambda r^\alpha$ in \eqref{eqn:Jr-lower1}. Putting  everything together, we thus obtain
\begin{align*}
   \nonumber a_{2R}-a_{2r} &\ge-c_3'  r^2\ + c_1' r^{2-\alpha}\frac{1}{4} c_2mr^{\alpha-2}\lambda \, {A_{2cr}^d}   
   \ge  -c_3'   r^2\ + 4\eps   {A_{2cr}^d}, 
\end{align*}
provided $   m\ge  16\eps/(c_1'c_2\lambda) .$

Finally, applying the lower bounds of $a_{2r}$ and ${A_{2cr}^d}$ in \eqref{eqn:ar} and \eqref{eqn:Ar},  and the condition $r^2\le u(\xi)/M$ in \eqref{eqn:Mr2-gen}, we obtain
\begin{align*}
  a_{2R} &\ge u(\xi)-c_0r^2 -c_3'  r^2+ 4\eps\big(u(\xi)-c_0r^2\big)
\ge  (1+3\eps)u(\xi),
\end{align*}
provided $M\ge (c_0+c_3'+4\eps c_0)/\eps$. 
Since $a_{2R}=\sum_{\partial B_{2R}}P_R(\xi,x)u(x)$ and $\sum_{\partial B_{2R}}P_R(\xi,x)=1$, there is therefore $\xi' \in \partial B(\xi,2R)\subset {B(\xi,CR)}$ such that 
\begin{align*}
    u({\xi'})\ge (1+3\eps)u(\xi),
\end{align*}
where $C=2c^2>2$ is the radius scaling constant given in \eqref{eqn:ball-relation}. 
Just as in case I,  $u(x)=0$ for $x \not\in A $ implies that $\xi'\in A \cap   B(\xi,CR)$,  
which completes the proof of Lemma~\ref{lem:u-growth-gen}. 
\end{proof}

Once Lemma~\ref{lem:u-growth-gen} is established, if $u(\xi)\ge E^{-1}= Mr_0^2$, then one can attempt to apply the lemma inductively to construct a sequence of balls with growing radii $r_{k+1}=(1+\eps)r_{k}$ satisfying 
 \eqref{eqn:MR2-gen}, until one exhausts the region $A$.
Because the outcome \eqref{eqn:MR2-gen} is inputted into the next application of Lemma~\ref{lem:u-growth-gen}, the only remaining condition we need to check each time to apply the lemma is \eqref{eqn:Jr-lower}. 
The probability that \eqref{eqn:Jr-lower} holds in each step of the induction can be estimated in terms of the cumulative probability distribution function $F_V$ and $F_\mu$, leading to:
\begin{proof}[Proof of Lemma~\ref{lem:PmaxF}]
    For $E>0$, let $r_0=(ME)^{-1/2}$ be the initial scale to be used for Lemma~\ref{lem:u-growth-gen}, and let  $\mathcal P =\{B(z_i,r_0)\}$ be a cover of $\Gamma$ with the finite overlap property \eqref{eqn:finite-cover-lambda}.
    For any ball $B(x_0,r_0)\in \mathcal P$, we will bound the probability 
\begin{align}\label{eqn:probbound}
\P\big \{\min_{\xi\in B} \frac{1}{u(\xi)}\le E\big\}
=\P\big\{\exists \xi \in B(x_0,r_0): u(\xi)\ge E^{-1}=Mr_0^2\big\}
\end{align}
    from above. 
Considering the event $u(\xi)\ge Mr_0^2$, we will apply Lemma~\ref{lem:u-growth-gen} repeatedly, assuming \eqref{eqn:Jr-lower} holds, on each  scale $r_k=(1+\eps)r_{k-1}$, for $k=1,\ldots,K$, until $r_{K}>{\rm diam}A/2$. 
At this point, the conclusion \eqref{eqn:MR2-gen} will no longer hold if we start with $\xi_\infty$ where $u$ obtains its maximum, and so somewhere along the way (for some $k=1,\ldots,K$), the condition \eqref{eqn:Jr-lower} must have failed. The gives an upper bound for \eqref{eqn:probbound} in terms of probabilities of events of the form \eqref{eqn:Jr-lower}, which can in turn be bounded in terms of the CDFs $F_V$ and $F_\mu$.

For the first step, we start with the event $u(\xi_0)\ge Mr_0^2$ for some $\xi_0\in B(x_0,r_0)$. We suppose the following event also holds 
    \[\mathcal E_0:=\Big\{ \big|J_{r_0}\big|\ge \lambda \big |B(x_0,r_0)\cap A \big | \Big\},  \]
    where $J_{r_0}=J_{r_0,x_0}$ is defined as in \eqref{eqn:def-J-gen}. Then Lemma~\ref{lem:u-growth-gen} guarantees a point $\xi_1\in B(\xi,Cr_1) \cap A$ such that $u(\xi_1)\ge Mr_1^2$. Note $\xi_1$ is contained in $  B(x_0,Cr_0+Cr_1)\cap A$ since  
    \[d(x_0,\xi_1)\le d(x_0,\xi_0)+d(\xi_0,\xi_1)\le  r_0+Cr_1 \le C(r_0+r_1). \]
As in the argument in the proof of Proposition~\ref{prop:covering}(b), we only need at most $N_1\le c(\frac{r_0+r_1}{r_1})^{\alpha}$ many balls of radius $r_1$ to cover $  B(x_0,Cr_0+Cr_1)$, where $c$ depends only on the overlap constant $b_\Gamma$ in \eqref{eqn:finite-cover-lambda}, and the volume control constants (and $\alpha$) in \eqref{eqn:vol-control}. Then $\xi_1$ must be located in one of these balls, which we call $B(x^0_1,r_1)$. We denote by $B(x^i_1,r_1)$, for $i=1,2,\ldots,N_1-1$  the rest of the balls of the same radius $r_1$ that we used to cover $  B(x_0,Cr_0+Cr_1)$. To apply Lemma~\ref{lem:u-growth-gen} again, we will assume the following event holds,
     \[
\mathcal{E}_1:=\bigcap_{i=0}^{N_1-1}\big\{ \big|J_{r_1,x_1^i}\big|\ge \lambda \big |B(x^i_1,r_1)\cap A \big | \big\},
\]
     where $J_{r_1,x_1^i},i=0,1,\ldots,N_1-1$ are defined as in \eqref{eqn:def-J-gen} on different balls $B(x^i_1,r_1)$. 

Inductively, suppose we have  $\xi_k\in B(\xi_{k-1},Cr_k)\cap A$ such that $u(\xi_k)\ge Mr_k^2$. We can check $\xi_k$ is contained in $  B(x_0,{2C}r_k/\eps)\cap A$ since  
    \[d(x_0,\xi_k)\le    {C}(r_0+r_1+\cdots+r_k)\le {C}r_k\frac{1}{1-(1+\eps)^{-1}}\le \frac{{2C}r_k}{\eps}. \]
    We only need at most $N_k\le c\eps^{-\alpha}$ many balls of radius $r_k$ to cover $  B(x_0,{2C}r_k/\eps)$, where $c$ again only depends on the overlap constant $b_\Gamma$ in \eqref{eqn:finite-cover-lambda} and volume control constants (and $\alpha$) in \eqref{eqn:vol-control}. Proceeding as before, the new event we need to assume to hold is   
     \[
\mathcal{E}_k:=\bigcap_{i=0}^{N_k-1}\big\{ \big|J_{r_k,x_k^i}\big|\ge \lambda \big |B(x^i_k,r_k)\cap A \big | \big\},
\]
     where $B(x^i_k,r_k),i=0,1,\cdots$ are all the balls of radius $r_k$ that we used to cover $  B(x_0,{2C}r_k/\eps)$.

       When we reach the first radius such that $r_K>{\rm diam}A/2$ with $u(\xi_K)\ge Mr_K^2$ for some $\xi_K\in B(x_0,{2C}r_K/\eps)\cap A$, we then consider $\xi_\infty \in A$ where the maximum of $u$ is attained. 
Since $u(\xi_\infty)$ is the maximum, then $u(\xi_\infty)\ge u(\xi_K)\ge Mr_K^2$. 
To apply Lemma~\ref{lem:u-growth-gen}, since  ${\rm diam}A/2\le r_K\le {\rm diam}A$, we only need $N_K\le 
       C\left(\frac{{\rm diam}A}{r_K}\right)^\alpha \le (2C)^{\alpha}\lesssim\eps^{-\alpha}$ 
     many balls  $B(x_K^i,r_K)$  to cover the entire domain $A$. 
The corresponding event we need to assume holds is
   \[
\mathcal{E}_K:=\bigcap_{i=0}^{N_K-1}\big\{ \big|J_{r_K,x_K^i}\big|\ge \lambda \big |B(x^i_K,r_K)\cap A \big | \big\}.
\]
Lemma~\ref{lem:u-growth-gen} then produces a point $\xi' \in A$ such that $u(\xi')\ge (1+3\eps)u(\xi_\infty)>u(\xi_\infty)$, which contradicts the maximality of $u(\xi_\infty)$. Thus
     \[\big\{\exists \xi \in B(x_0,r_0): u(\xi)\ge  Mr_0^2\big\}\subset \Big(\mathcal E_0\cap \mathcal E_1\cdots \cap \mathcal E_{K} \Big)^C=\mathcal E_0^C\cup \mathcal E_1^C\cup\cdots \cup   \mathcal E_K^C, \]
and
     \begin{align}\label{eqn:243}
         \P\big \{\min_{\xi\in B(x_0,r_0)} \frac{1}{u(\xi)}\le E\big\}
\le \P(\mathcal E_0^C)+\cdots+\P(\mathcal E_K^C). 
     \end{align}
     
Finally, it remains to estimate the probabilities of the events $\mathcal{E}_k^C$ in terms of the CDFs $F_V$ and $F_\mu$.
For any of the $x_k=x_k^i$, recall the definition of $J_{r_k}=J_{r_k,x_k}$ in \eqref{eqn:def-J-gen} and rewrite   $J_{r_k}= J^V_{r_k}\cup J^{\mu}_{r_k}  $, where
\begin{align*} 
J_{r_k}^V&=\big\{x\in {B(x_k,r_k)}\cap A:\ V_x  \ge m   r_k^{-2} \big \},\\
J_{r_k}^{\mu}&=\big\{x\in {B(x_k,r_k)}\cap A:\ \sum_{y\in {B(x_k,r_k)}\cap A:  y\sim x}\, (1-\mu_{xy})  \ge m   r_k^{-2} \big \}. \end{align*}
Because counting $J^\mu_{r_k}$ involves dependent variables, first bound its size in terms of the independent variables $\sigma_{xy}=1-\mu_{xy}$,
\begin{align*}
|J^\mu_{r_k}| \ge \frac{1}{2M_\Gamma}\left|\{\sigma_{xy}\ge mr_k^{-2}:x,y\in {B(x_k,r_k)}\cap A,\,x\sim y\}\right|.
\end{align*}
Applying the Chernoff bound to the binomial random variables $|J_{r_k}^V|$ and $|\{\sigma_{xy}\ge mr_k^{-2}\}|$ with the optimal parameter (cf. \cite[Lemma~4.5]{arnold2022landscape}), there are $r_\ast$ and $\lambda_\ast$ such that if $r_k\ge r_\ast$ and $\lambda\le \lambda_\ast$, then
\begin{align}
  \P\{|J_{r_k}^V|\le \lambda|{B(x_k,r_k)}\cap A|  \}\le F_V(m   r_k^{-2})^{|{B(x_k,r_k)}\cap A|/2}\le F_V(m   r_k^{-2})^{cr_k^\alpha}, \label{eqn:FV-bound}\\
  \P\{|J_{r_k}^{\mu}|\le \lambda|{B(x_k,r_k)}\cap A|  \}\le F_{\mu}(m   r_k^{-2})^{|{B(x_k,r_k)}\cap A|/2}\le F_{\mu}(m   r_k^{-2})^{cr_k^\alpha}, \label{eqn:Fmu-bound}
\end{align}
where in the last inequalities we used the lower bound $|{B(x_k,r_k)}\cap A|\ge cr_k^{\alpha}$ from \eqref{eqn:AintersectB} and that $F_V,F_{\mu}\le 1$. 
Since $J_{r_k}=J_{r_k}^V\cup J_{r_k}^\mu$, then using independence of the $\mu_{xy}$ and $V_x$,
\begin{align*}
    \P \Big\{ \big|J_{r_k}\big|\le \lambda \big |B(x_k,r_k)\cap A \big | \Big\}\le  & \P\{|J_{r_k}^V|\le \lambda|{B(x_k,r_k)}\cap A|  \}\, \P\{|J_{r_k}^{\mu}|\le \lambda|{B(x_k,r_k)}\cap A|  \} \notag \\
    \le & \Big(F_V(m   r_k^{-2}) F_{\mu}(m   r_k^{-2})\Big)^{cr_k^\alpha}. \label{eqn:FVmu}
\end{align*}
As the above argument does not depend on the particular center $x_k$, \eqref{eqn:FV-bound}, \eqref{eqn:Fmu-bound} and \eqref{eqn:FVmu} hold for all $J_{r_k,x_k^i}$ on the different balls $B(x_k^i,r_k)$.
Since there are at most $N_k\le c\eps^{-\alpha}$ many balls for each $k\in\{0,\cdots,K\}$, then 
\begin{align*}
    \P(\mathcal E_k^C) &\le 
\sum_{i=0}^{N_k-1} \P\Big\{ \big|J_{r_k,x^i_k}\big|\le \lambda \big |B(x^i_k,r_k)\cap A \big | \Big\}\\
    &\le  c\eps^{-\alpha}\Big(F_V(m   r_k^{-2}) F_{\mu}(m   r_k^{-2})\Big)^{cr_k^\alpha} 
\le   c\eps^{-\alpha}\Big(F_V(m   r_0^{-2}) F_{\mu}(m   r_0^{-2})\Big)^{cr_k^\alpha},\numberthis
\end{align*}
 where for the last inequality we used that $F_V(E),F_{\mu}(E)>0$ are non-decreasing for $E\le E_0$, and that $mr_k^{-2} \le mr_0^{-2}  $. Furthermore, $r_k=(1+\eps)^kr_0$ implies $r_k^\alpha\ge (1+\eps\alpha k)r_0^\alpha\ge r_0^\alpha+k$ provided $r_0\ge r_\ast(\eps,\alpha)$. Combined with the assumption $F_V(E)F_{\mu}(E)<1$ for all $E<E_0$, we obtain
\begin{align*}
  \P(\mathcal E_k^C) \le    c\eps^{-\alpha}\left(F_V(m   r_0^{-2}) F_{\mu}(m   r_0^{-2})\right)^{cr_0^\alpha+ck},\ \ k=0,1,\cdots, K.
\end{align*} 
Putting all these estimates together with \eqref{eqn:243} and setting $F_\ast=F_V(E_\ast) F_{\mu}(E_\ast)<1$  with $E_\ast=\min\{E_0,1/(Mr_\ast^2)\}$, yields 
\begin{align*}
   \P\big \{\min_{\xi\in B(x_0,r_0)} \frac{1}{u(\xi)}\le E\big\} 
   &\le  \frac{2c\eps^{-\alpha}}{1-F_\ast^c}\big(F_V(m   r_0^{-2}) F_{\mu}(m   r_0^{-2})\big)^{cr_0^\alpha} \\
   &= C \big(F_V(mME) F_{\mu}(m ME)\big)^{cM^{-\alpha/2}E^{-\alpha/2}},
\end{align*}
which proves the landscape Lifshitz tail upper bound \eqref{eqn:Lif-upper-Nu}. 

Note the above induction starts with the initial scale $cE^{-1/2}=r_0<{\rm diam}A/2$. For $\operatorname{diam}A/2\le r_0$, we can do the following instead:
Since $\max u\lesssim ({\rm diam}A)^2$ by Proposition~\ref{prop:piconseq}(ii), then $N_u(E)=0$ for $E\le c_1/({\rm diam}A)^2$ (equivalently, $r_0\ge C\operatorname{diam}A$), and the Lifshitz tail upper bound \eqref{eqn:Lif-upper-Nu} holds trivially. For $c_1/({\rm diam}A)^2\le E \le c_2/({\rm diam}A)^2$, with constants $c_1,c_2$ corresponding to the regime ${\rm diam}A/2\le r_0\le C{\rm diam}A$, we can take $r_K=r_0$ and jump to the last step directly. Note the $C\ge 1$ here is the one that can be taken in \eqref{eqn:AintersectB}.
Hence, there is no restriction on the diameter of $A$ as in the Lifshitz tails lower bound. 
\end{proof}

As a consequence of Theorem~\ref{thm:Lif}, we obtain the landscape law for random models in Corollary~\ref{cor:LLaw-random}. 
We do not prove Corollary~\ref{cor:LLaw-random}  directly, as previous work \cite{DFM,arnold2022landscape} reduced the proof of Corollary~\ref{cor:LLaw-random} to the Landscape Law bounds \eqref{eqn:LLaw-upper}, \eqref{eqn:LLaw-lower1}, \eqref{eqn:LLaw-lower2} and Lifshitz tail bounds \eqref{eqn:Lif-lower-Nu}, \eqref{eqn:Lif-upper-Nu} for $N_u$. In particular, the landscape law for random models in \eqref{eqn:LLaw-random}
follows from Theorems~\ref{thm:LLaw} and \ref{thm:Lif} using the argument of \cite[Thm. 3.56]{DFM}. See also \cite[\S4.2]{arnold2022landscape} for an axiomatic version of this method.

\section{Details for applications to specific models}\label{sec:ex}
In this section, we provide details for applications to random band models (Section~\ref{sec:RBM}), stacked graphs (Section~\ref{subsec:stack}), and the Sierpinski gasket graph (Section~\ref{sec:SG}).

\subsection{Random band model $H_{d,W}$ and proof of Corollary~\ref{cor:HdW}}\label{sec:RBM}
 Recall the standard $\Z^d$ lattice satisfies Assumption~\ref{ass:0}  
and a (weak) Poincar\'e  inequality \eqref{eqn:WPI}, see e.g.   Example \ref{ex:Zd} and \cite{barlow2017random}. 
 Retain the definitions of the band graph $\Gamma_{d,W}$ in \eqref{eqn:Gamma-dW}. As mentioned in Example \ref{ex:band-graph} in Section~\ref{subsec:rough-isom}, $\Gamma_{d,W}$ is roughly isometric to $\Z^d$. As a consequence of Proposition~\ref{prop:isom}, all aforementioned properties of $\Z^d$ will be preserved on $\Gamma_{d,W}$ by the rough isometry. Hence, Theorem~\ref{thm:LLaw} can be applied to $H_{d,W}$ on $\Gamma_{d,W}$ since all requirements are met. The (weak) Poincar\'e  inequality \eqref{eqn:WPI} also guarantees the Lifshitz tails \emph{lower} bound \eqref{eqn:Lif-lower-Nu} for $H_{d,W}$ as long as the domain $A$ satisfies \eqref{eqn:AintersectB}.   
 In order to obtain the Lifshitz tails upper bound \eqref{eqn:Lif-upper-Nu}, it remains to explicitly construct a harmonic weight as required in Assumption~\ref{ass:harmonic-weight}  for $\Gamma_{d,W}$. This will then complete the proof of  Corollary~\ref{cor:HdW}. 
 
Denote by  $ d(x,y)=\sqrt {\sum_{1\le i \le d}(x_i-y_i)^2}$ the usual Euclidean distance for $x,y\in \Z^d$. For $\xi \in \Z^d$, let $B_R^d(\xi)=\{x:d(\xi,x)\le R\}$ be  Euclidean ball of radius $R$. 
 Recall on $\Gamma_{d,W}$, that $x\sim y$ iff $d(x,y)\le W$. The exterior boundary $B_R^d(\xi)$, with respect to this metric, is 
\begin{align*}
 \partial B_R^d(\xi)=\big\{x\not\in B_R^d(\xi): \exists y \in  B_R^d(\xi)\  {\rm such \ that}\ d(x,y)\le W \big\} .   
\end{align*}
 Let $P_{B^d_\rho}(\xi,y)$ be the Poisson kernel on $\partial B^d(\xi,\rho)$ in $\Gamma_{d,W}$, defined as in \eqref{eqn:Pr-Gr}. 
We will use $P_{B^d_\rho}(\xi,y)$ to construct the desired harmonic weight $h_{B^d_R}$ on the Euclidean ball $B_R^d(\xi)$. 
Intuitively, we want to define the harmonic weight layer by layer using the Poisson kernel, but the graph structure for $W>1$ complicates the boundary regions.
We will thus first need the following technical lemma, 
which says that ``spherical shells'' in $\Z^d$ can be covered by thin layers of exterior boundaries in $\Gamma_{d,W}$.

\begin{lemma}\label{claim:layer}
    There is $\rho_0(d,W)>0$ so that for any $\rho\ge \rho_0(d,W)$, 
    \begin{align*}
        B^d(\xi,\rho+1) \backslash B^d(\xi,\rho) \subset   
        \bigcup_{k=1}^{k_\ast} \partial B^d(\xi,\rho+a_k),
    \end{align*}
    where $k_\ast:=\lfloor \frac{2\sqrt d}{W}\rfloor$, $a_k=1-k\frac{W}{2\sqrt d}$ for $k=0,1\cdots,k_\ast-1$, and $a_{k_\ast}=0$. 
\end{lemma}
We note that if $W>2\sqrt d$, then  $ B^d(\xi,\rho+1) \backslash B^d(\xi,\rho) \subset   \partial B(0,\rho)$, and we do not need the extra thin layers in between. Also note that $a_{k-1}-a_k= \frac{W}{2\sqrt d}$ for $k=1\cdots,k_\ast-1$, and  that
\begin{align}\label{eqn:layer-bound}
   a_{k_\ast-1}-a_{k_\ast}=1-\Big(\Big\lfloor \frac{2\sqrt d}{W}\Big\rfloor-1\Big)\frac{W}{2\sqrt d}=1- \Big\lfloor \frac{2\sqrt d}{W}\Big\rfloor \frac{W}{2\sqrt d}+\frac{W}{2\sqrt d}\ge \frac{W}{2\sqrt d}.  
\end{align}

   We will now use Lemma~\ref{claim:layer} to explicitly construct the harmonic weight and prove Corollary~\ref{cor:HdW}. The proof of the lemma is left to the end of this section. Let $\rho_0$ and $a_i$ be as in Lemma~\ref{claim:layer}. Then applying Lemma~\ref{claim:layer} several times,  for e.g. $R\ge 2W$,
\[B^d(\xi,R-W)\backslash B^d(\xi,\rho_0)=\bigcup_{\rho=\rho_0}^{R-W-1}B^d(\xi,\rho+1)\backslash B^d(\xi,\rho)\subset\bigcup_{\rho=\rho_0}^{R-W-1}\bigcup_{k=1}^{k_\ast}\partial B^d(\xi,\rho+a_k).  \]
We rearrange the above radii $\{\rho+a_k\}_{\rho,a_k}$ in increasing order  and denote them by  $\rho_i\in [\rho_0,\rho_{i_\ast}],\ i=0,\cdots,i_\ast$,  so that $\rho_{i+1}-\rho_i\ge \frac{W}{2\sqrt d}$ , and $\rho_{i_\ast}=R-W-\frac{W}{2\sqrt d}$. The covering can then be rewritten as
\[B^d(\xi,R-W)\backslash B^d(\xi,\rho_0)\subset \bigcup_{i=0}^{i_\ast}\partial B^d(\xi,\rho_i):=Y_R.  \]
We define $X_R:=B^d(\xi,R)\backslash Y_R$, which will be the ``bad set'' in Assumption~\ref{ass:harmonic-weight}. Note that  
\begin{align}\label{eqn:XR}
  \big|X_R\big|= \big|B(\xi,R)\backslash Y_R\big| \le   \big|B^d(\xi,\rho_0) \big| +\big|B^d(\xi,R )\backslash B^d(\xi,R-W)\big|\le C  , 
\end{align}
 for some constant $C$ depending only on $d,W$, so that $|X_R|/|B^d(\xi,R)|\to 0$ as $R\to \infty$. 
Overall, 
 Lemma~\ref{claim:layer} thus allows us to cover a large portion of $B^d(\xi,R)$  using thin exterior boundaries of Euclidean balls. With such a filtration, we can define the associated harmonic weight $ h_{B^d_{R}}(\xi,y)$ on $B^d(\xi,R)$ layer by layer. 
 
For $y\in  Y_R=\cup_{i=0}^{i_\ast}\partial B^d(\xi,\rho_i)$, we then set 
\begin{align}\label{eqn:h-0}
 h_{B_R^d(\xi)}(y)=  
  \frac{1}{\big|B^d_{R-W}\backslash B^d_{\rho_0} \big |} \sum_{i=0}^{i_*}\big| B^d_{\rho_{i+1}} \backslash B^d_{\rho_{i}} \big |\, \cdot P_{B^d_{\rho_i}}(\xi,y) \cdot \one_{\{y\in   \partial B^d_{\rho_{i}}\}},  
\end{align}
  and for $y\in X_R$, set $ h_{B^d_{R}}(\xi,y)=0$.  
  
  We first verify that $ h_{B^d_{R}}$ is a harmonic weight satisfying \eqref{eqn:h-submean-ass} from Assumption~\ref{ass:harmonic-weight}. Suppose $-\Delta f(x) \le 0$ (for $x\in \Omega \supset B^d_{R}\cup \partial B^d_{R}$). Then by the surface submean property \eqref{eqn:submean-surface}, for any $\rho_i\le R-W$, 
\[f(\xi)\le \sum_{y\in \partial B^d(\xi,\rho_i)}P_{B^d_{\rho_i}}(\xi,y)f(y).\]
Multiply  the   equation both sides by the volume (cardinality) of $\big| B^d(\xi,\rho_{i+1}) \backslash B^d(\xi,\rho_i) \big |$ and then sum all the equations from $i=0$ to $i=i_\ast$ (setting $\rho_{i_\ast+1}=R-W$) to obtain
\[\sum_{   i= 0 }^{ i_\ast }\big| B^d_{\rho_{i+1}} \backslash B^d_{\rho_{i}} \big | f(\xi)\le  \sum_{   i= 0 }^{ i_\ast }\sum_{y\in \partial B^d(\xi,\rho_i)}\big| B^d_{\rho_{i+1}} \backslash B^d_{\rho_{i}} \big |P_{B^d_{\rho_i}}(\xi,y)f(y),    \]
which implies 
\begin{align*}
     \nonumber f(\xi)&\le \frac{1}{\big|B^d_{R-W}\backslash B^d_{\rho_0} \big |} \sum_{   i= 0 }^{ i_\ast }\sum_{y\in \partial B^d(\xi,\rho_i)}\big| B^d_{\rho_{i+1}} \backslash B^d_{\rho_{i}} \big |P_{B^d_{\rho_i}}(\xi,y)f(y) \\
 &= \sum_{y\in Y_R}h_{B_R^d(\xi)}(y)f(y), 
\end{align*}
 where $h_{B^d_{R}} $ is given as in \eqref{eqn:h-0}. 
 Putting together the definition   $ h_{B^d_{R}}(\xi,y)=0$ for $y\in X_R$, one concludes that for $-\Delta f(x) \le 0$, 
 \[ f(\xi)\le      \sum_{y\in B^d(\xi,y)}h_{B_R^d(\xi)}(y)f(y),   \]
which verifies \eqref{eqn:h-submean-ass}. Clearly, the equality holds if $f$ is harmonic ($\Delta f=0$).

We already verified the `bad' set $X_R$ is small in \eqref{eqn:XR}. It thus remains to show  that $h_{B^d_{R}}$ has the desired lower bound in \eqref{eqn:h-lower-ass} on   $Y_R$.  For this, we will need the following Poisson kernel bounds.
\begin{proposition}[Lemma~6.3.7 in \cite{lawler2010random}]\label{prop:lawler}
  There are constants $c_1,c_2>0$ depending only on $d,W$ such that 
 \begin{align}\label{eqn:Poisson-lawler}
  c_1 \rho^{1-d} \le    P_{B^d_\rho}(\xi,y)\le  c_2 \rho^{1-d}, \  \ \xi\in \Z^d, \ y \in \partial B^d_\rho(\xi).
 \end{align}   
\end{proposition}
 
From classical
results, as $s\to\infty$ the standard Euclidean ball $B^d(\xi,s)$  contains $\omega_ds^d+o(s^{d-1})$ lattice points, where $\omega_d$ is the volume of the unit Euclidean ball in $\R^d$; see e.g. \cite{Walfisz,HeathBrown,Huxley} for further references and more precise error estimates. 
Thus there is $c_3>0$ so that for a fixed $a>0$ and all $s\ge1$, 
\begin{align*}
|B^d(\xi,s)|-|B^d(\xi,s-a)| = d\omega_d as^{d-1}(1+o(1)) \ge c_3 as^{d-1}.
\end{align*}
Using this lower bound between the layers $\rho_{i+1}-\rho_{i}\ge \frac{W}{2\sqrt d}$ in \eqref{eqn:layer-bound}, we have 
\[\big| B^d(\xi,\rho_{i+1}) \backslash B^d(\xi,\rho_i) \big |\ge c_3\frac{W}{2\sqrt d} \rho_i^{d-1}, \ \ \ i =0,\cdots,i_\ast.  \]
Combined with  the  lower bound    in \eqref{eqn:Poisson-lawler}, this implies for $y \in \partial B^d_{\rho_i}(\xi)$, 
\begin{align*}
  \big| B^d(\xi,\rho_{i+1}) \backslash B^d(\xi,\rho_i) \big | P_{B^d_{\rho_i}}(\xi,y) \ge c_1c_3\frac{W}{2\sqrt d}.  
\end{align*}
Since for $y\in Y_R=\bigcup_{i=0}^{i_\ast}\partial B^d_{\rho_i}(\xi)$, the sum \eqref{eqn:h-0} contains at least one (non-zero) term. Hence,  
\[h_{B_R^d(\xi)}(y)\ge   \frac{1}{\big|B^d(\xi,R-W)\backslash B^d(\xi,\rho_0) \big |}c_1c_3\frac{W}{2\sqrt d}\ge \frac{c_5}{\big|B^d(\xi,R) \big |}, \ y\in Y_R,\]
which verifies \eqref{eqn:h-lower-ass} in Assumption~\ref{ass:harmonic-weight}.

\begin{proof}[Proof of Lemma~\ref{claim:layer}]
    Since the annulus is symmetric and $\Z^d$-translation invariant, it is enough to consider $x=(x_1,\cdots,x_d)\in  B^d(0,\rho+1) \backslash B^d(0,\rho)$ where $x_1\ge0$, and  $x_1\ge x_i,i=2,\cdots,d$. In this region, $x$ satisfies
    \[ \rho^2<\sum_{i=1}^dx_i^2\le (\rho+1)^2 \Longrightarrow  x_1\ge \frac{\rho}{\sqrt d}. \]
    Since $x_1$ is the maximal direction, consider the point $x'=x-We_1$, which is a neighbor $x'\sim x$ in $\Gamma_{d,W}$.
    We will see that if $W$ is relatively large (compared to $\sqrt d$), then $x'\in B(0,\rho)$ so that $x\in \partial B(0,\rho)$. Otherwise, we will need intermediate layers of thickness $\frac{W}{2\sqrt d}$ to reach the entire annulus. Direct computation using $\|x\|_2^2\le(\rho+1)^2$ and $x_1\ge\rho/\sqrt{d}$ shows
    \begin{align*}
        \|x'\|_2^2 = (x_1-W)^2+ \sum_{i=2}^dx_i^2 
        &\le \rho^2+2\Big(1-\frac{W}{2\sqrt d}\Big)\rho-
       \frac{W}{\sqrt d}\rho
        +W^2+1.
    \end{align*}
    Hence, if $W\ge 2\sqrt d$, 
    then for $\rho>\rho_0:=\sqrt{d}(W+1/W)$, 
    $\|x'\|_2^2 
    \le \rho^2-\frac{W}{ \sqrt d}\rho+W^2+1
    \le  \rho^2$, 
    so that $x'=x-We_1\in B^d(0,\rho)$ and $x\in\partial B^d(0,\rho)$.

     If  $W< 2\sqrt d$, set \[ a_k=1-k\frac{W}{2\sqrt d}, \ \ k=0,1\cdots,\Big\lfloor \frac{2\sqrt d}{W}\Big\rfloor-1. \]
    For each $k=1,\cdots,\Big\lfloor \frac{2\sqrt d}{W}\Big\rfloor-1 $, if $x\in B(0,\rho+a_{k-1})\backslash B(0,\rho+a_{k})$, then  $ \|x\|_2^2\le (\rho+a_{k-1})^2$ and $x_1\ge\rho/\sqrt{d}$ imply
     \begin{align*}
   \|x'\|_2^2=     (x_1-W)^2+ \sum_{i=2}^dx_i^2
        &\le  \rho^2+2\Big(a_{k-1}-\frac{W}{\sqrt d}\Big)\rho+W^2+a_{k-1}^2  
             \le   (\rho +a_{k})^2,
    \end{align*}
    for $\rho\ge (W^2+1)/(\frac{W}{\sqrt{d}})>(W^2+a_{k-1}^2-a_{k}^2)/(\frac{W}{\sqrt d})$. 
    This shows $x'\in B(0,\rho+a_{k})$, hence $x \in \partial B(0,\rho+a_{k})$  since $x\not \in B(0,\rho+a_{k})$.
   
   The above works for all $k\le  k_\ast-1$. 
   In the last layer when $k=k_\ast=\lfloor \frac{2\sqrt d}{W}\rfloor\le \frac{2\sqrt d}{W}$, we have 
   \[a_{k_\ast-1}=1-(k_\ast-1)\frac{W}{2\sqrt d}\ge 1-\Big(\frac{2\sqrt d}{W}-1\Big)\frac{W}{2\sqrt d} =\frac{W}{2\sqrt d} ,\]
and similarly as before,
   \begin{align*}
   \|x'\|_2^2 
        &\le    \rho^2+2\big(a_{k_\ast-1}-\frac{W}{2\sqrt d}\big)\rho-\frac{W}{\sqrt d} \rho+W^2+a_{k_\ast-1}^2 \\
         &=   \rho^2 -\Big(k_\ast+1-\frac{2\sqrt d}{W}\Big)\frac{W}{\sqrt d} \rho+W^2+a_{k_\ast-1}^2  
            \le    \rho  ^2,
    \end{align*}
    provided $\rho>(W^2+a_{k_\ast-1}^2 )/(c\frac{W}{\sqrt d})$, where $c=\lfloor \frac{2\sqrt d}{W}\rfloor+1-  \frac{2\sqrt d}{W}>0$. Hence, $d(0,x')\le \rho$ and $x \in \partial B(0,\rho)=\partial B(0,\rho+a_{k_\ast})$, where we set $a_{k_\ast}=0$. 
\end{proof}

\subsection{Stacked graphs}
\label{subsec:stack}
In this section, we provide details for the application to stacks of graphs in Section~\ref{subsec:app-stack}. In particular, we verify that if Assumption~\ref{ass:harmonic-weight} holds for a graph $\Gamma$ with the natural metric $d_\Gamma$, then the required properties also hold for the stacked graph $\Gamma\times\Z_M$ with the metric $\tilde{d}((x,j),(y,k)):=d_\Gamma(x,y)+\frac{1}{2}\mathbf{1}_{(x=y)\wedge (j\ne k)}$.
Since the weak Poincar\'e inequality is preserved under rough isometry, this together allows for obtaining Lifshitz tails for $\Gamma\times\Z_M$ via Corollary~\ref{cor:LLaw-random}.

First, the metric $\tilde{d}$ is strongly equivalent to the natural metric on $\Gamma\times\Z_M$.
The balls under the metric $\tilde{d}$ are simpler however. 
For $R\ge1/2$, the ball is $B^d_{\Gamma\times\Z_M}((\xi,j),R)=B_\Gamma(\xi,R)\times\Z_M$. 
For $R<1/2$, the ball centered at $(\xi,j)$ is the singleton set $\{(\xi,j)\}$.

To define the harmonic weight $h_{B^d_R((\xi,j))}((y,k))$, we use the random walk formulation of the Poisson kernel for the ball $B^d((\xi,j),R)$. Letting $H_A((\xi,j),(y,k))=\P_{(\xi,j)}[Y_{\tau_{A}}=(y,k)]$ be the Poisson kernel for a region $A$, we will take
\begin{align}
h_{B^d_R((\xi,j))}((y,k))&=\begin{cases}0,&(y,k)\in\{\xi\}\times\Z_M\\
\frac{1}{R}H_{B^d_{r-1}((\xi,j))}((\xi,j),(y,k)),&r=d((\xi,j),(y,k))\ge 2\\
\frac{1}{R}H_{B^d_{1/2}((\xi,j))}((\xi,j),(y,k)),&r=d((\xi,j),(y,k))=1
\end{cases}.
\end{align}
We can then check the required properties in Assumption~\ref{ass:harmonic-weight}.
\begin{itemize}

\item Submean property: 
For $R>3/2$ and $f$ harmonic at points in $B^d(\xi,R-1)$, then for any $3/2\le r\le R$, the Poisson kernel property \eqref{eqn:poisson-exp} implies
\begin{align}\label{eqn:stack-bdy}
f((\xi,j)) &= \sum_{(y,k)\in\partial B^d_{r-1}((\xi,j))}H_{B^d_{r-1}((\xi,j))}((\xi,j),(y,k))f((y,k)).
\end{align}
Averaging over $r\in\{3/2\}\cup\{2,3,\ldots,R\}$, then
\begin{align}
f((\xi,j))&=\sum_{(y,k)\in B^d_R((\xi,j)}h_{B^d_R((\xi,j))}((y,k))f((y,k)).
\end{align}
The required inequality \eqref{eqn:h-submean-ass} then follows for subharmonic functions from \eqref{eqn:stack-bdy}.
If $R\le 3/2$ then we can simply take the ``trivial'' harmonic weight $h_{B_R^d((\xi,j))}=\delta_{(\xi,j)}$, since the second property in Assumption~\ref{ass:harmonic-weight} only matters as $R\to\infty$.

\item Bad set: The plan is to compare $h_{B^d_R((\xi,j))}$ for $\Gamma\times\Z_M$ to the harmonic weight for $\Gamma$, using the random walk relation and Laplace equation properties of the Poisson kernel.
We claim that if $X_R(\xi)$ is the ``bad set'' for $B^d(\xi,R)$ in $\Gamma$, then $(X_R(\xi)\times\Z_M)\cup(\{\xi\}\times\Z_M)$ can be taken as the ``bad set'' for $\Gamma\times\Z_M$.

Since we added the entire stack of $M$ points $\{\xi\}\times\Z_M$ to the bad set, we can ignore it and only consider the Poisson kernel at points $(\xi,j)$ and $(y,k)$ with $r=d((\xi,j),(y,k))\ge1$, i.e. $y\ne\xi$.
Assuming $r\ge1$, we first show that for $k=1,\ldots,M$,  
\[H_{B^d_r((\xi,j))}((\xi,j),(y,k))=\P_{(\xi,j)}[Y_{\tau_{B^d((\xi,j),r)}}=(y,k)] \] 
are all comparable. 
For a path $S=(v_0,\ldots,v_n)$ of length $n=\tau_{B^d((\xi,j),r)}$ exiting at $v_n=(y,k)\in\partial B^d((\xi,j),r)$, we can construct a path $S'$ of length $n+|k-\ell|$ exiting at $(y,\ell)$ instead, by following $S$ up to $v_{n-1}$ and then moving vertically to the $\ell$th layer before exiting. Since the maximum degree in $\Gamma\times\Z_M$ is bounded by $M_\Gamma+2$, this implies
\begin{align}\label{eqn:layer-comparison}
\P_{(\xi,j)}[Y_{\tau_{B^d((\xi,j),r)}}=(y,k)]&\ge\frac{1}{(M_\Gamma+2)^M}\P_{(\xi,j)}[Y_{\tau_{B^d((\xi,j),r)}}=(y,\ell)],
\end{align}
for any $k,\ell\in\{1,\ldots,M\}$.

Averaging \eqref{eqn:layer-comparison} over $\ell\in\{1,\ldots,M\}$ yields for $r\ge 1/2$,
\begin{align*}\label{eqn:unstack}
\P_{(\xi,j)}[Y_{\tau_{B^d((\xi,j),r)}}=(y,k)]&\ge\frac{1}{M(M_\Gamma+2)^M}\sum_{\ell=1}^M\P_{(\xi,j)}[Y_{\tau_{B^d((\xi,j),r)}}=(y,\ell)]\\
&=\frac{1}{M(M_\Gamma+2)^M}\P_{\xi}[\tilde{Y}_{\tau_{B(\xi,r)}}=y],\numberthis
\end{align*}
where $\tilde{Y}_n$ is simple random walk on $\Gamma$. The last equality follows using the Poisson kernel formula for solving Laplace's equation:
For a finite region $\Lambda\subset\Gamma$, let $H_{\Lambda}(x,y)$ be the Poisson kernel on $\Gamma$ and $H_{\Lambda\times\Z_M}((x,j),(y,\ell))$ be the Poisson kernel on $\Gamma\times\Z_M$. If $f:\Gamma\to\R$ is harmonic on $\Lambda\subset\Gamma$ with boundary values $f(y)$ for $y\in\partial\Lambda$, then $\tilde{f}:\Gamma\times\Z_M\to\R$ defined as $\tilde{f}((x,j))=f(x)$ is also harmonic on $\Lambda\times\Z_M$ with boundary values $\tilde{f}((y,\ell))=f(y)$ for $(y,\ell)\in\partial\Lambda\times\Z_M$. Then
\begin{align*}
f(x)=\E_x[f(Y_{\tau_\Lambda})]=\sum_{y\in\partial\Lambda}H_\Lambda(x,y)f(y),
\end{align*}
and also
\begin{align*}
f(x)=\tilde{f}((x,j))&=\sum_{(y,\ell)\in\partial(\Lambda\times\Z^M)}H_{\Lambda\times\Z_M}((x,j),(y,\ell))\tilde{f}(y,\ell)\\
&=\sum_{y\in\partial\Lambda}\left(\sum_{\ell=1}^MH_{\Lambda\times\Z_M}((x,j),(y,\ell))\right)f(y).
\end{align*}
By considering boundary values $\delta_y$, this implies we must have
\begin{align}
H_\Lambda(x,y)&=\sum_{\ell=1}^MH_{\Lambda\times\Z_M}((x,j),(y,\ell)).
\end{align}
Taking $\Lambda=B(\xi,r)$ yields \eqref{eqn:unstack}.

\end{itemize}

\subsection{Sierpinski gasket graph}\label{sec:SG}

In this part, we discuss the landscape law and Lifshitz tails for Jacobi operators on the Sierpinski gasket graph $\Gamma_{SG}$, drawn in Figure~\ref{fig:sierpinski}. 
The Sierpinski gasket graph   is a fractal-like graph, and is not roughly isometric to any $\Z^d$. 
Let $K_0$ be the unit triangle on $\R^2$ with vertices $\V_0=\{(0,0),(1,0),(1/2,\sqrt 3/2)\}=\{ a_1,a_2,a_3\}$.  
Then $\Gamma_{SG}$ is constructed using the images of $\V_0$ (as vertices) by the iteration of $\Psi$.  Let $\V_\ell=2^\ell\Psi^\ell(\V_0)$ be the set of vertices of the triangles of side length 1 in $2^\ell\Psi^\ell(K_0)$. Let $\V=\cup_{\ell\ge 0}\V_\ell$.  The edge set $\mathcal E=\{x\sim y:x,y\in \V\}$ is defined by the relation $x \sim y$ iff $x,y$ belongs to a triangle in $2^\ell\Psi^\ell(K_0)$ for some $\ell$.  The  Sierpinski gasket graph is then $\Gamma_{SG}=(\V,\mathcal E)$.

A key feature of $\Gamma_{SG}$ is that it satisfies sub-Gaussian heat kernel bound $\mathrm{HKC}(\alpha,\beta)$; for $t\ge \max(1,d_0(x,y))$,
\[ 
\frac{c_3}{t^{\alpha/2}}\exp\left(-c_4\big(d_0(x,y)^\beta/t)^{\frac{1}{\beta-1}}\right)\le  q_t(x,y)\le \frac{c_1}{t^{\alpha/2}}\exp\left(-c_2\big(d_0(x,y)^\beta/t)^{\frac{1}{\beta-1}}\right),
 \]
 where $q_t(x,y)$ is the (continuous time) heat kernel as in \eqref{eqn:gaussian-heatkernel}, 
$\alpha=\log 3/\log 2$ is the volume growth parameter as in \eqref{eqn:vol-control},  and $\beta=\log 5/\log 2$  is the sub-Gaussian parameter. (See e.g. \cite[Cor. 6.11]{barlow2017random}.)
Since the volume control property \eqref{eqn:vol-control} holds, 
we have the covering of balls provided by Proposition~\ref{prop:covering}.  As mentioned in Remark~\ref{remk:beta-PI}, the parameter $\beta$ is equivalent to the parameter used in (weak) Poincar\'e inequalities. 
In the previous sections of this paper, we focused on the case $\beta=2$ and assumed 
the corresponding weak Poincar\'e inequality  \eqref{eqn:WPI}. 
On $\Gamma_{SG}$, as a result of the sub-Gaussian heat kernel bounds, a  $\beta$-version of \eqref{eqn:WPI} holds \cite[\S6]{barlow2017random}, with $r^\beta$ rather than $r^2$ in the Poincar\'e   inequality, i.e., 
\[  \sum_{x\in B }\big(f(x)-\bar f^B\big)^2 \le  C_P\, r^\beta\sum_{x,y\in B^\ast: x\sim y}\big(f(x)-f(y)\big)^2  ,\]
with the same average $\bar f^B$, and the ball scaling as in $\eqref{eqn:WPI}$. 
This suggests to consider a slightly different landscape  counting function $N_{u,\beta}$, which depends on the parameter $\beta>2$. 
More precisely, unlike   coverings of radius $R=E^{-1/2}$ used in \eqref{eqn:Nu-1/2}, instead, we consider  coverings of radius $R=E^{-1/\beta}$.
The associated landscape counting function is defined as 
\begin{align}\label{eqn:Nu-1/beta}
N_{u,\beta}(E)=  N_u^{\mathcal P(E^{-1/\beta}),A}(E). 
\end{align}

Following the proof in Section~\ref{sec:LL-upper-det}, replacing the use of Poincar\'e inequality \eqref{eqn:WPI} by its $\beta$-version, one can obtain the landscape law upper bound \begin{align}\label{eqn:LLaw-upper-SG}
       N(E)\le N_{u,\beta}(CE) ,\ \ {\rm for \ all }\ E >0,
    \end{align} where $C$ depends on $\alpha,\beta$.

Unfortunately, the argument for the landscape lower bounds \eqref{eqn:LLaw-lower1}, \eqref{eqn:LLaw-lower2} does not go through for $\Gamma_{SG}$. 
Additionally, we do not verify Assumption~\ref{ass:harmonic-weight} for the gasket graph, which prevents using the landscape law argument to obtain a $\beta$ version of the Lifshitz tails property.

 Recently, authors in \cite{balsam2023density} established the Lifshitz tails singularity of the  integrated density
of states for certain random operators on (continuous) nested fractals, including the continuous Anderson model on the planar Sierpinski gasket set; see also earlier related work in \cite{pietruska1991lifschitz, shima1991lifschitz,kaleta2015integrated}. In particular, for (infinite volume) IDS $ N^\infty (E)$ of a random Schr\"odinger operator (under mild condition on the hopping and the random distribution) on the planar Sierpinski gasket set, Ref.~\cite{balsam2023density} showed that
\begin{align}\label{eqn:lif-beta}
  \lim_{E\searrow 0}   \frac{\log \big|\log  N^\infty (E)\big|}{\log E}=-\frac{\alpha}{\beta}. 
\end{align}

To the best of our knowledge, there is no such result for the (combinatorial) Sierpinski gasket graph. 
In upcoming work (in preparation), we prove \eqref{eqn:lif-beta} for the Anderson model on $\Gamma_{SG}$, via a modified Neumann--Dirichlet method close to the spirit of the work in \cite{shima1991lifschitz}.

 \section{Numerical cases}\label{sec:numerics}

In this section, we introduce and discuss a series of detailed numerical simulations aimed at investigating the behavior of the landscape counting function $N_u(E)$. 
These simulations identify more precise behavior (such as explicit numerically determined scalings) governed by our general results Theorems~\ref{thm:LLaw} and \ref{thm:Lif}, and also provide evidence for a landscape law or Lifshitz tails in models where we lack an analytical proof.
To comprehensively explore the applicability of the landscape law, we will consider a variety
of cases, including random band models, and the Anderson model on the Sierpinski gasket graph and Penrose tiling.

\subsection{Random band models}
Let's first recall some of the notations to be used in this section. Let $\Gamma_{d,W}=(\Z^d,\mathcal E_W)$ be the graph defined in Section~\ref{subsec:rbm}, where the vertex set is $\Z^d$ and the edge set $\mathcal E_W$ has the ``$W$-step band structure'' in \eqref{eqn:Gamma-dW}. Our results in Section~\ref{subsec:rbm} apply to Jacobi operators $H_{d,W}$ \eqref{eqn:HdW} with both onsite and bond  disorders. In the numerical simulations, we will focus only on the bond disorder, i.e., operators in the form 
\[ H_{d,W}f(x)= \sum_{\substack{y\in \Z^d \\ \|x-y\|\le W}} \big(f(x)-\mu_{xy}f(y)\big) .\]
We will consider the above operator $H_{d,W}$ in the cases $d=1,2$ (see  $\Gamma_{1,W}$ in  Figure~\ref{fig:Lap72}   and $\Gamma_{2,W}$ in  Figure~\ref{W2L1}) for different choices of  bandwidth $W$.
For the $\Gamma_{2,W}$ cases, the graph is induced by a band matrix on $\Z^2$. Specifically, we employ the $\ell^1$-norm to define $\mathcal E_W$. As a point of comparison with Figure~\ref{fig:Lap72}, Figure~\ref{W2L1} provides an illustrative example that displays all nodes connected to the central node within a bandwidth of $W=2$.

\begin{figure}[!ht]
	\centering
	\includegraphics[width=0.4\linewidth]{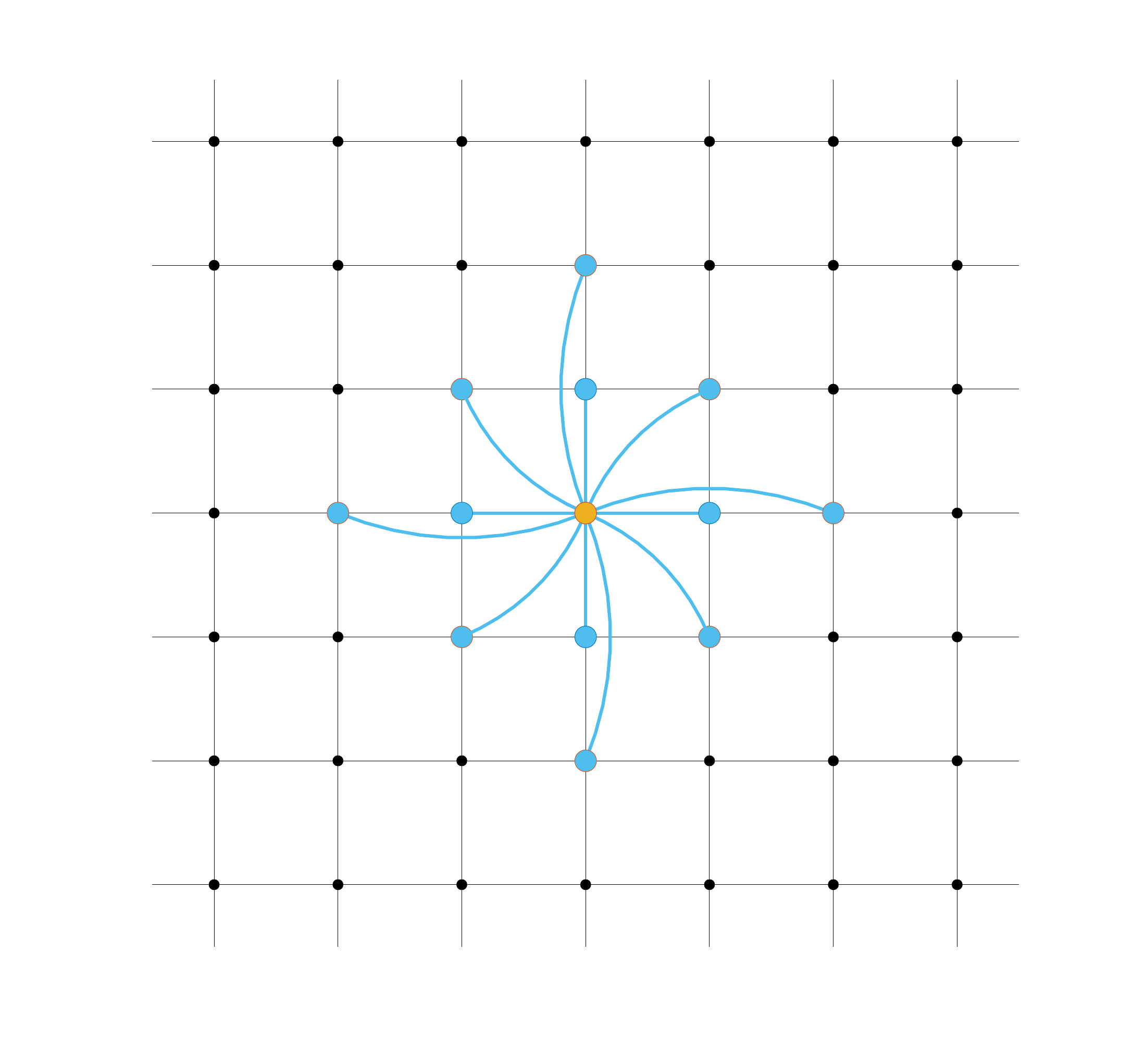}
	\caption{All nodes and edges corresponding to the centering node (in orange) in the graph $\Gamma_{2,2}=(\Z^2,\mathcal E_2)$}
	\label{W2L1}
\end{figure}

In Figure~\ref{rbm}, the top row shows an example in 1D with $|A|=20000$ and $W=10$, where  the bond interactions will be modeled using a uniform distribution over the interval [0,1]. The bottom row  shows a 2D example using a Bernoulli distribution taking values in \{0,1\} with half and half probabilities, where $|A|=100^2$ and $W=2$. Furthermore, the on-site non-negative potential is set to be 0. It is important to note that the simulations presented are based on a single random realization, rather than on the average of multiple realizations.

Building upon the previously discussed configurations, we globally get  the control of the IDS from above and below through the landscape counting functions.
Besides, beyond the global bounds from above and below, the landscape counting function, upon appropriate scaling, can serve as a good approximation for the Lifshitz tail of the IDS.  The examples shown in Figure~\ref{rbm}(b) and (e) are also to show how the suitably scaled landscape counting function closely mirrors the behavior of the Lifshitz tail  \eqref{eqn:Lif-dW}   under varying conditions:

\begin{figure}[!ht]
	\centering
	\subfigure []{
		\includegraphics[width=5 cm]{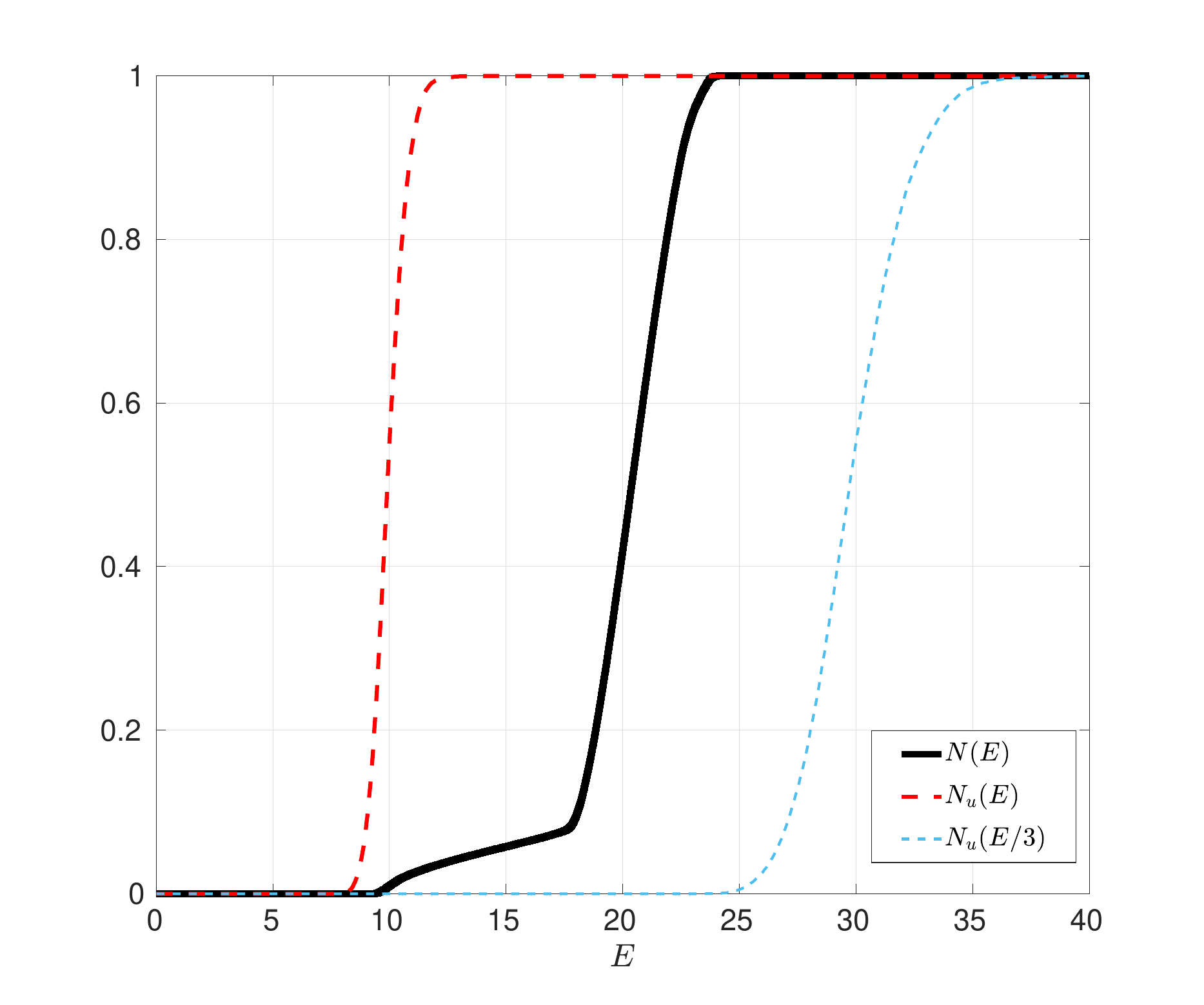}
	}
	\subfigure[]{
		\includegraphics[width=5cm]{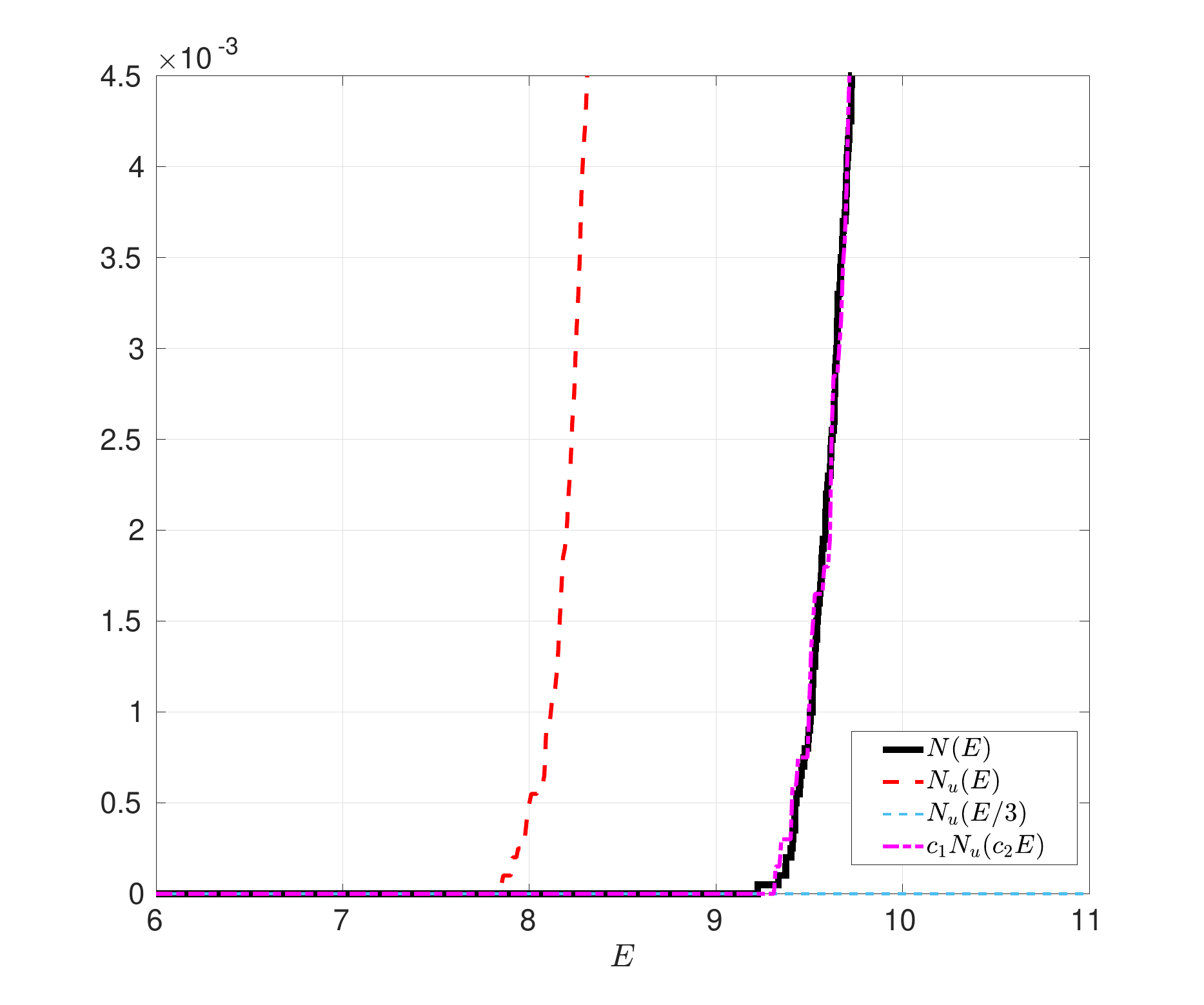}
	}
	\subfigure[]{
		\includegraphics[width=5 cm]{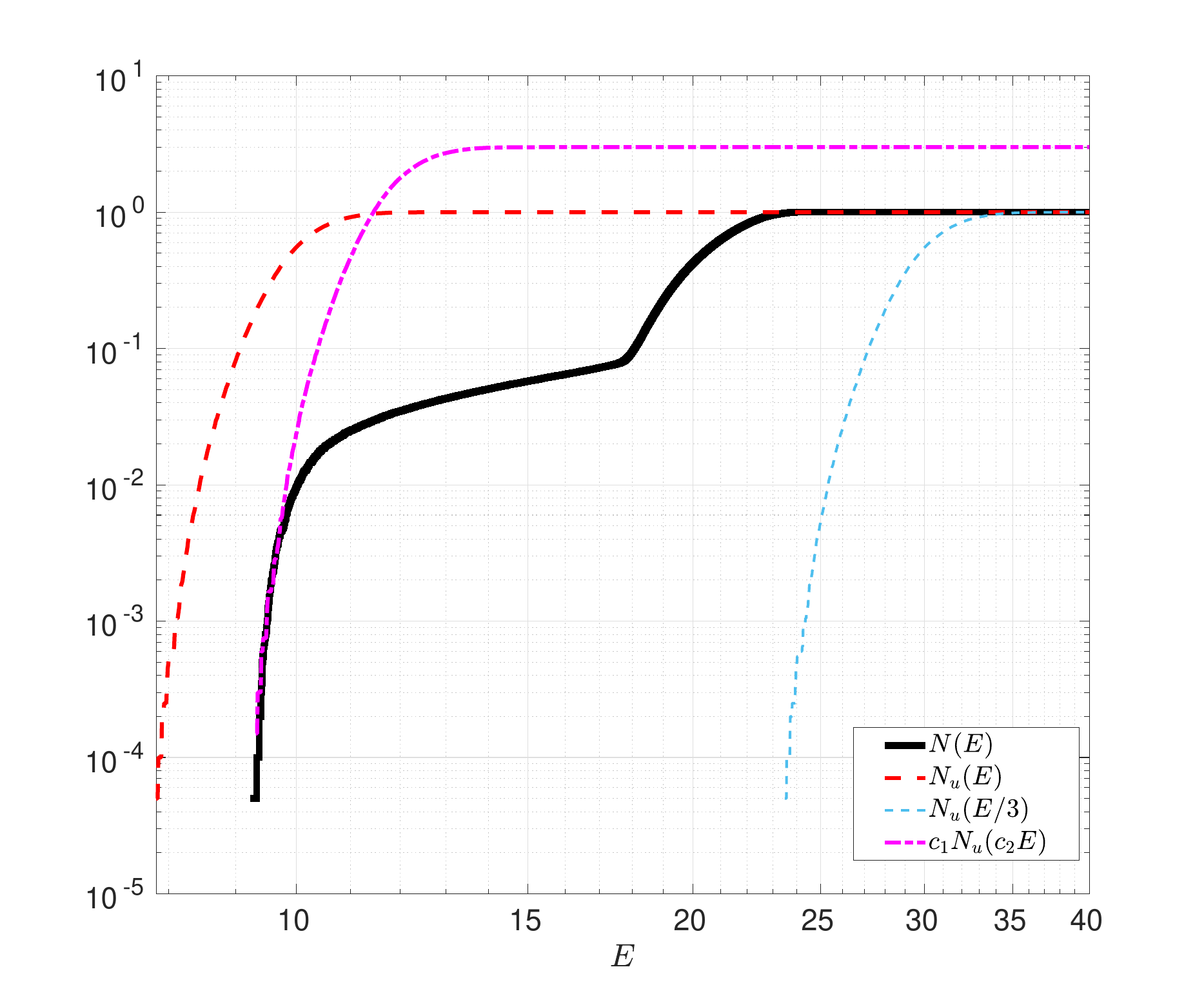}
	}
	\subfigure[]{
		\includegraphics[width=5 cm]{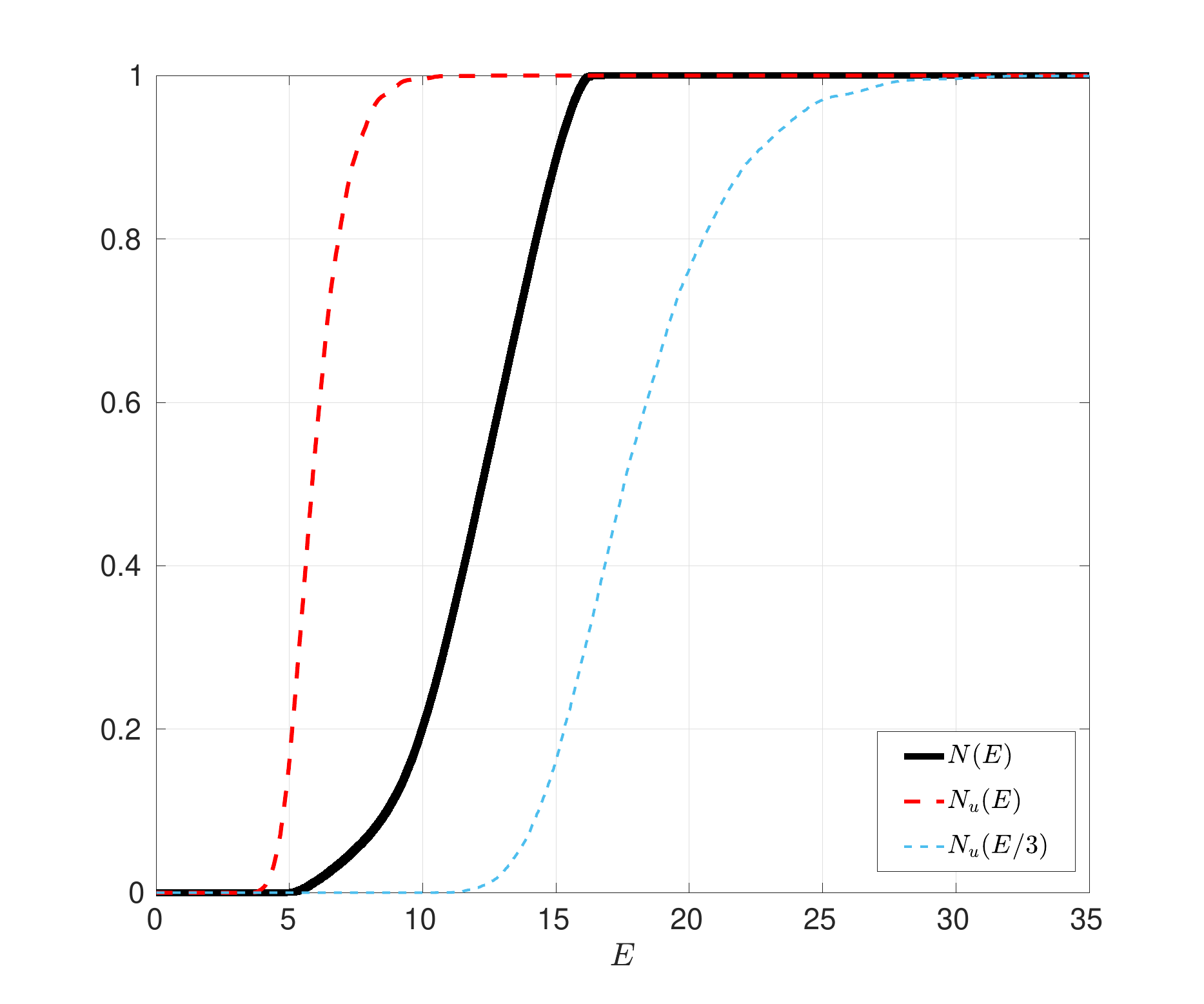}
	}
	\subfigure[]{
		\includegraphics[width=5 cm]{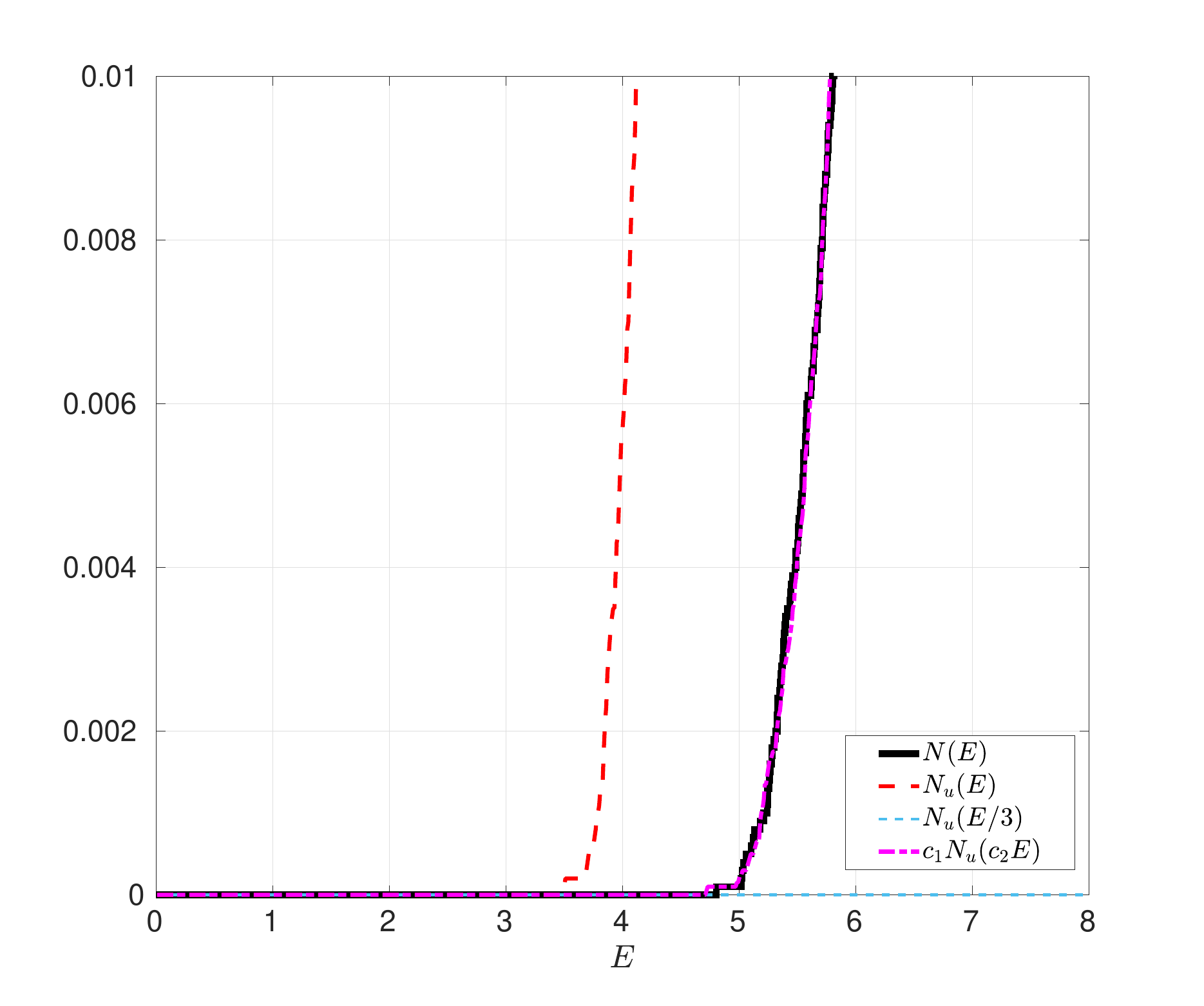}
	}
	\subfigure[]{
		\includegraphics[width=5 cm]{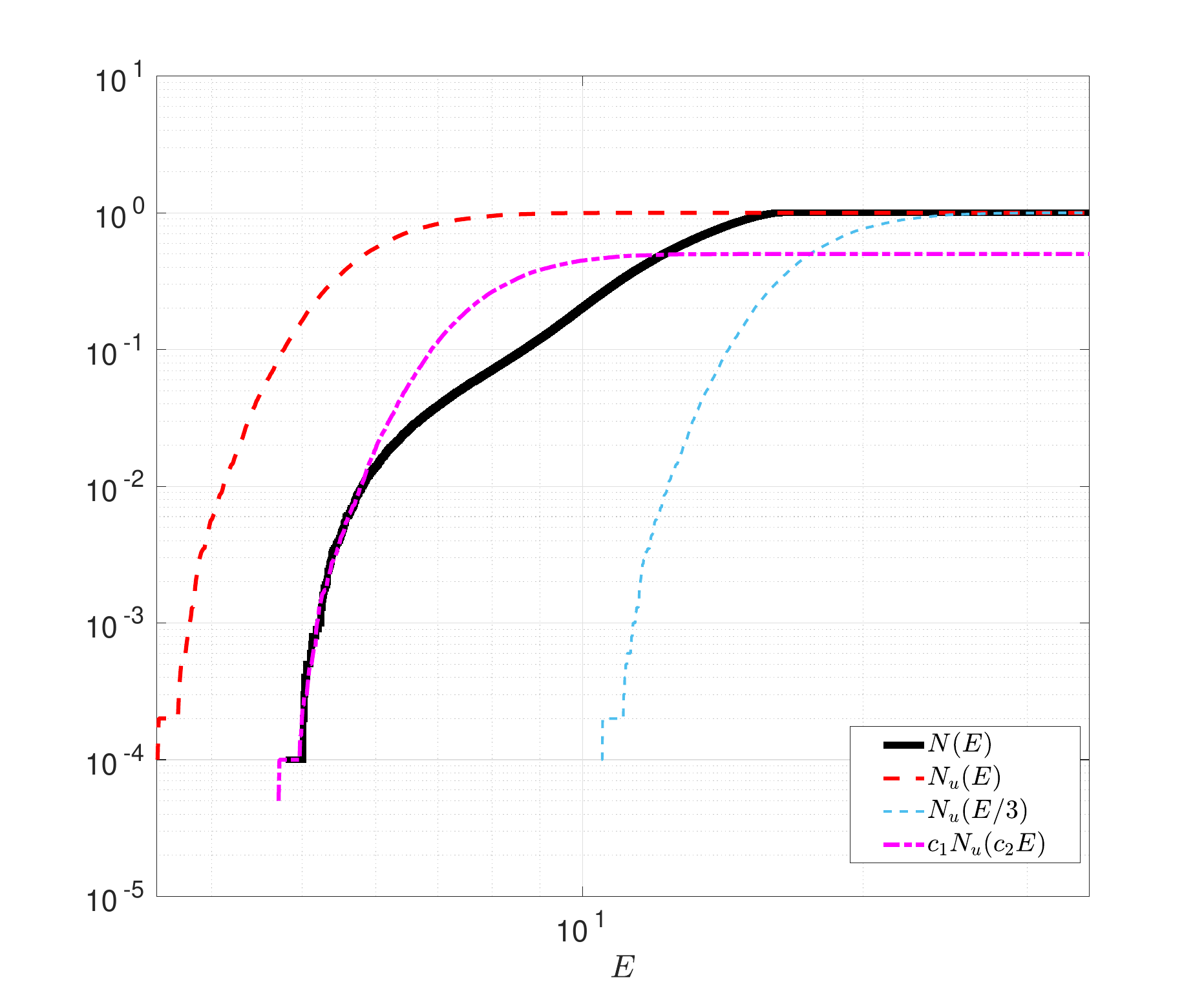}
	}
	\caption{The landscape law for random band models in 1D (top row (a), (b), (c)) and 2D (bottom row (d), (e), (f)) cases. (a), (d): the global control from above and below. (b), (e): the scaled landscape counting $c_1N_u(c_2E)$ effectively approximates the Lifshitz tail of $N(E)$ following appropriate scaling adjustments. Specifically, in (b) $c_1=3,c_2=\frac{1}{1.19}$; in (e): $c_1=\frac{1}{2},c_2=\frac{1}{1.35}$. (c), (f): the log-log view.}	
	\label{rbm}
\end{figure}

\subsection{Sierpinski gasket}
Next, we consider the Anderson model (without the bond disorder)
\[ Hf(x)= \sum_{y\in \V:y\sim x} \big(f(x)-f(y)\big) +V_x f(x)\]
on  Sierpinski gasket graph $\Gamma_{SG}=(\V,\mathcal E)$ discussed in Section~\ref{sec:SG} (see Figure~\ref{fig:sierpinski}, a fractal structure composed of interconnected equilateral triangles). For the landscape counting function, we specifically employ equilateral triangular boxes for the counting process. 
Figure~\ref{GasExample} illustrates one such example, demonstrating the methodology employed to compute the landscape counting function 
$N_{u,\beta}$ in \eqref{eqn:Nu-1/beta}.

\begin{figure}[!ht]
	\centering
	\includegraphics[width=0.7\linewidth]{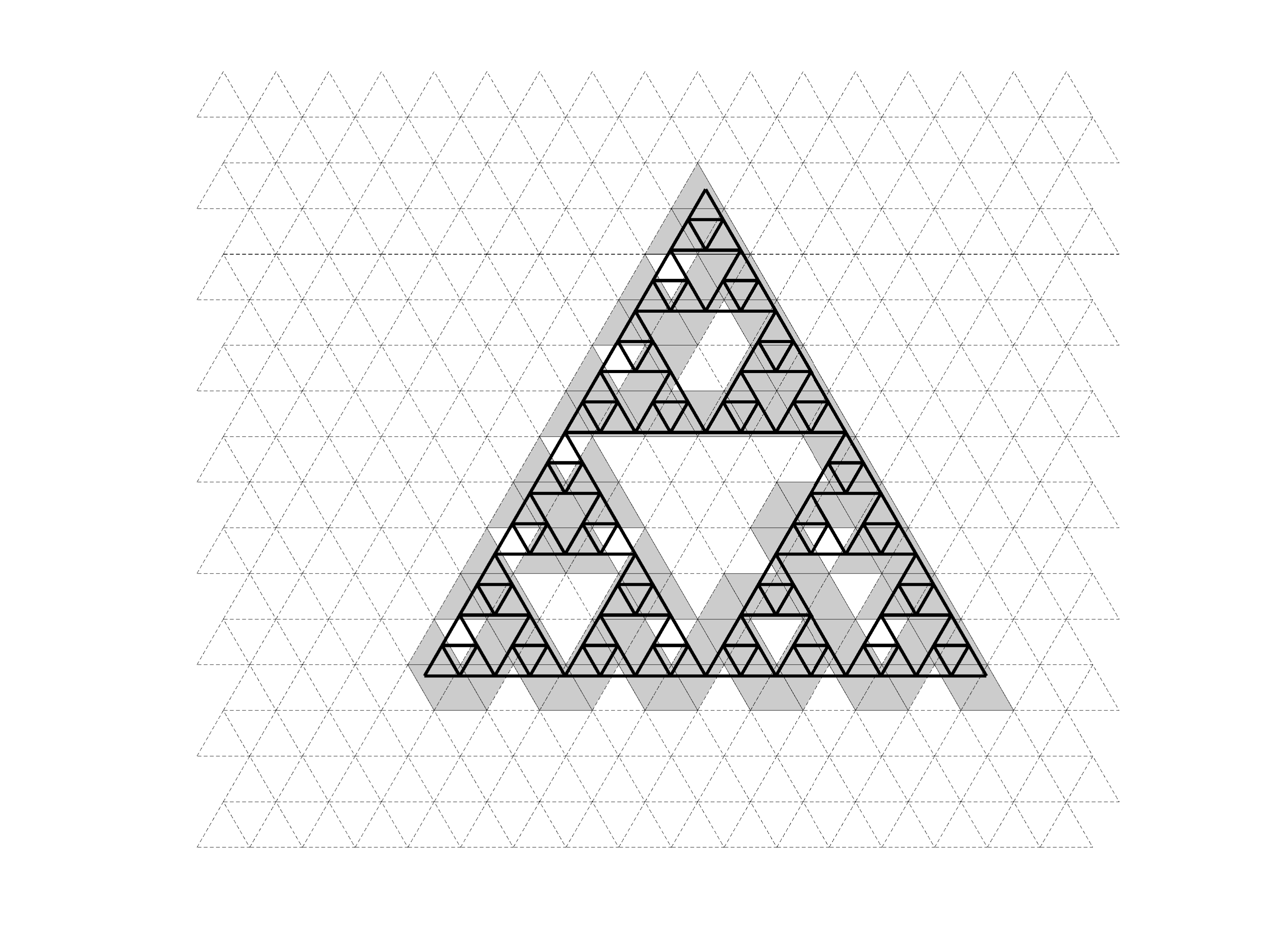}
	\caption{Triangular boxes used for computing the landscape counting function. In the equilateral triangles lattice, only those boxes that include at least one vertex of the gasket (highlighted in gray) are considered in the counting process.}
	\label{GasExample}
\end{figure}

Given that the Sierpinski gasket possesses a  volume growth parameter 
$\alpha=\frac{\log3}{\log2}$, we define the box size using $r = E^{-\frac{1}{\beta}}$, where $\beta = \frac{\log5}{\log2}$. As discussed in Section~\ref{sec:SG}, we are only able to prove the landscape law upper bound \eqref{eqn:LLaw-upper-SG} on the Sierpinski gasket. We expect a landscape law lower bound should also hold.  We will next illustrate the application of the landscape law. 
Note that in this example, we consider only the on-site non-negative potential $V$, without implementing any bond interaction. The subsequent figure presents the landscape law, with $V$ selected uniformly at random between 0 and 10, and $|A| = 9843$.

\begin{figure}[!ht]
	\centering
	\subfigure []{
		\includegraphics[width=7.05 cm]{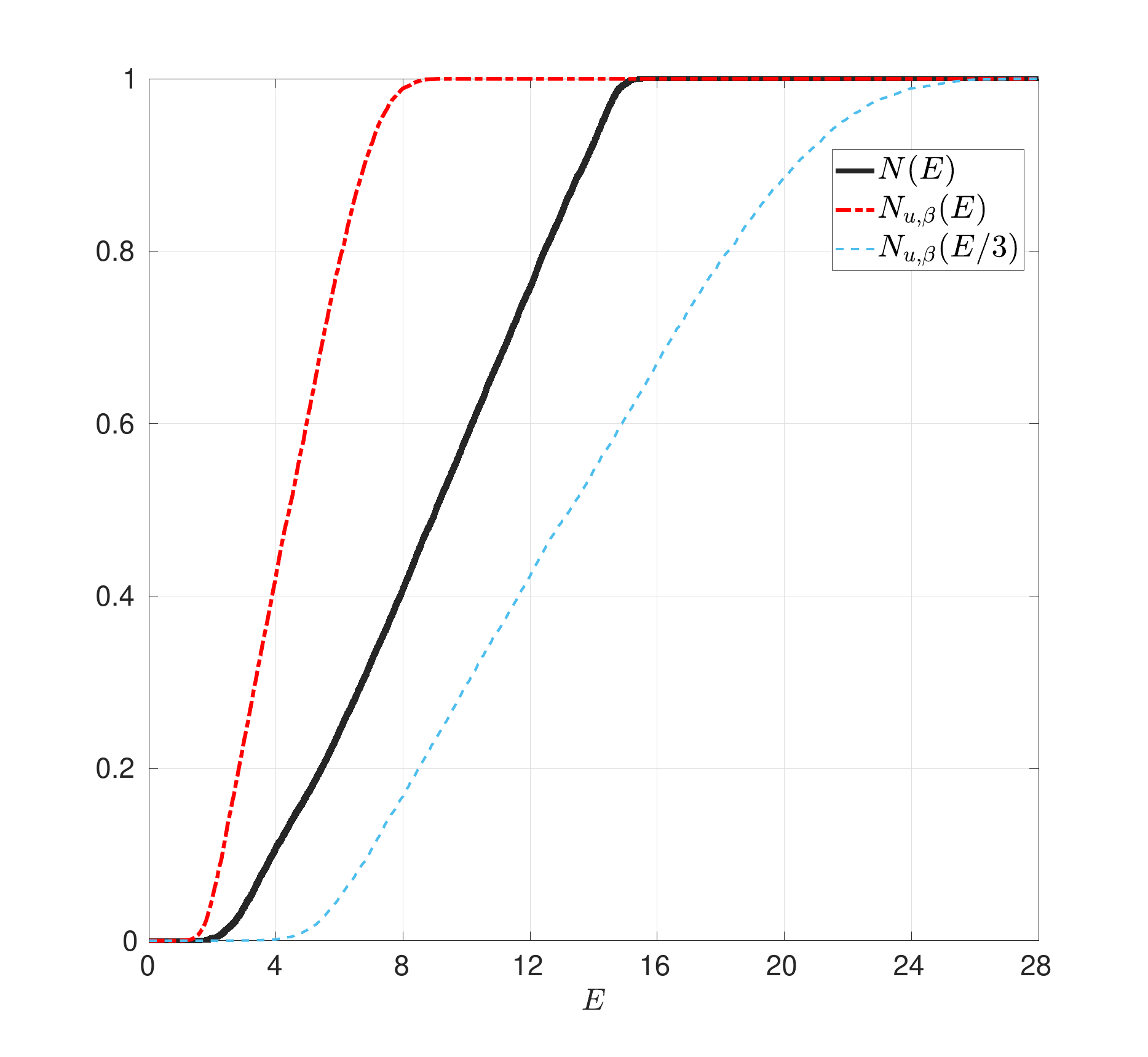}
	}
	\quad
	\subfigure[]{
		\includegraphics[width=7.1cm]{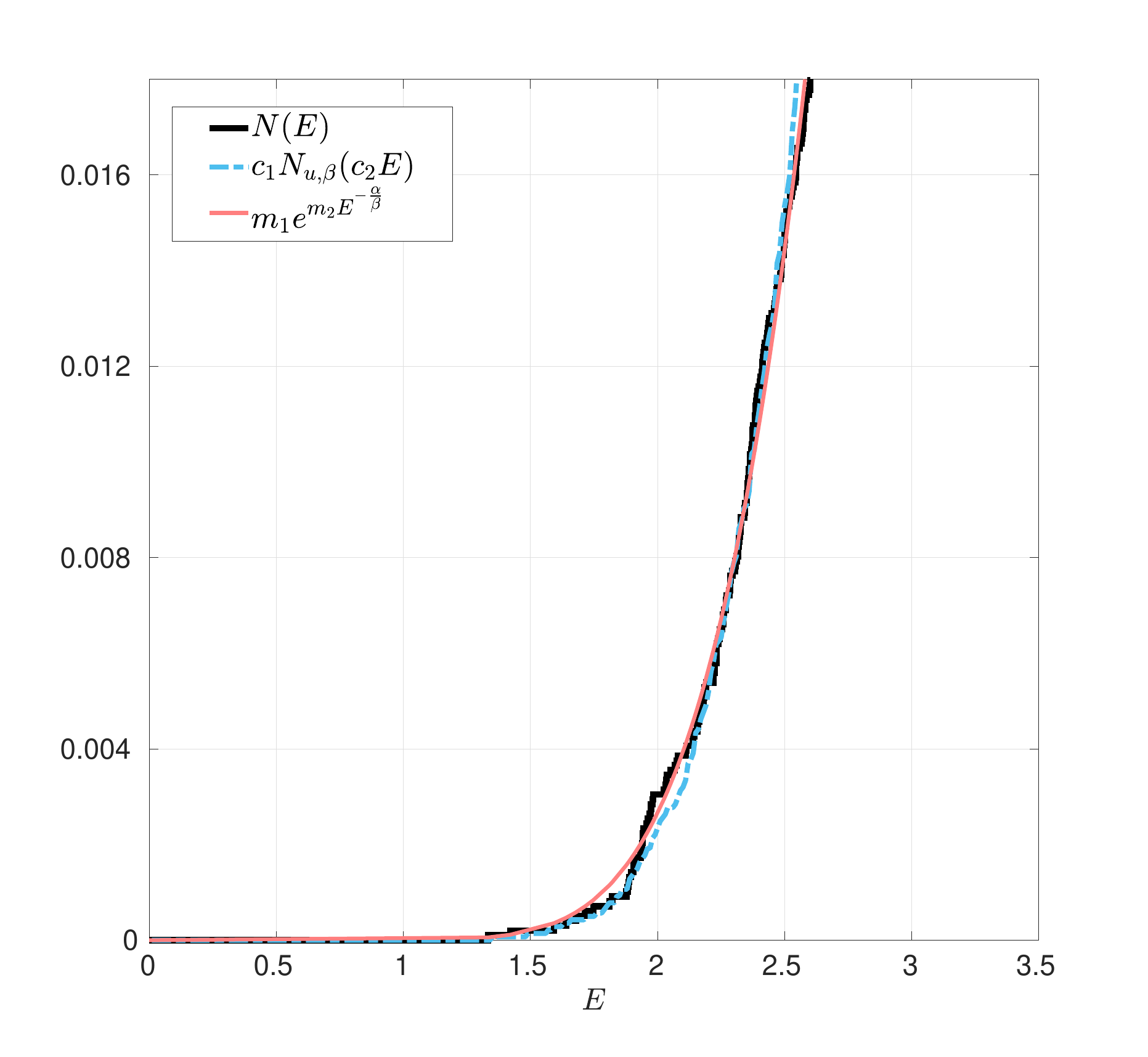}
	}
	\caption{Landscape Law and Lifshitz tail behavior for the Anderson model on the Sierpinski gasket graph. Left: the control of the IDS from above and below. Right: an expanded view of the lower energy spectrum, modeled by an exponential function with constants $m_1=419.4 , m_2=-19.2$. The scaled landscape counting function employs $c_1=0.7 , c_2=\frac{1}{1.4}$.}	\label{fig:gasket-ids}
\end{figure}

\subsection{Penrose tiling}

Lastly, we present an example illustrating the Lifshitz tail for the Anderson model on the Penrose tiling in Section~\ref{subsec:appl-isom} (see Figure~\ref{fig:penrose}), employing Neumann boundary conditions. The on-site non-negative potential $V$ is randomly assigned values from a uniform distribution $[0,4]$. As discussed in Example~\ref{ex:Penrose}, the Penrose tiling is roughly isometric to $\Z^2$. Hence, Theorem~\ref{thm:LLaw} can be applied to the Anderson model on it with $\alpha=2$. 
As we do not verify the harmonic weight assumption in Assumption~\ref{ass:harmonic-weight} for the Penrose tiling, we do not obtain the Lifshitz tails results through Corollary~\ref{cor:LLaw-random}, though here we provide numerical evidence supporting the behavior.
In Figure~\ref{PT}(c), we display the low energy regime
of the IDS, fitted by  an exponential function. Additionally, a scaled landscape counting function is presented, which is calculated over the lattice of equilateral triangles as illustrated in Figure~\ref{GasExample}. For comparative purposes, we also calculate the IDS for the free Laplacian over the same tiling in Figure~\ref{PT}(b). 

\begin{figure}[!ht]
	\centering
	\subfigure []{
		\includegraphics[width=9.05 cm]{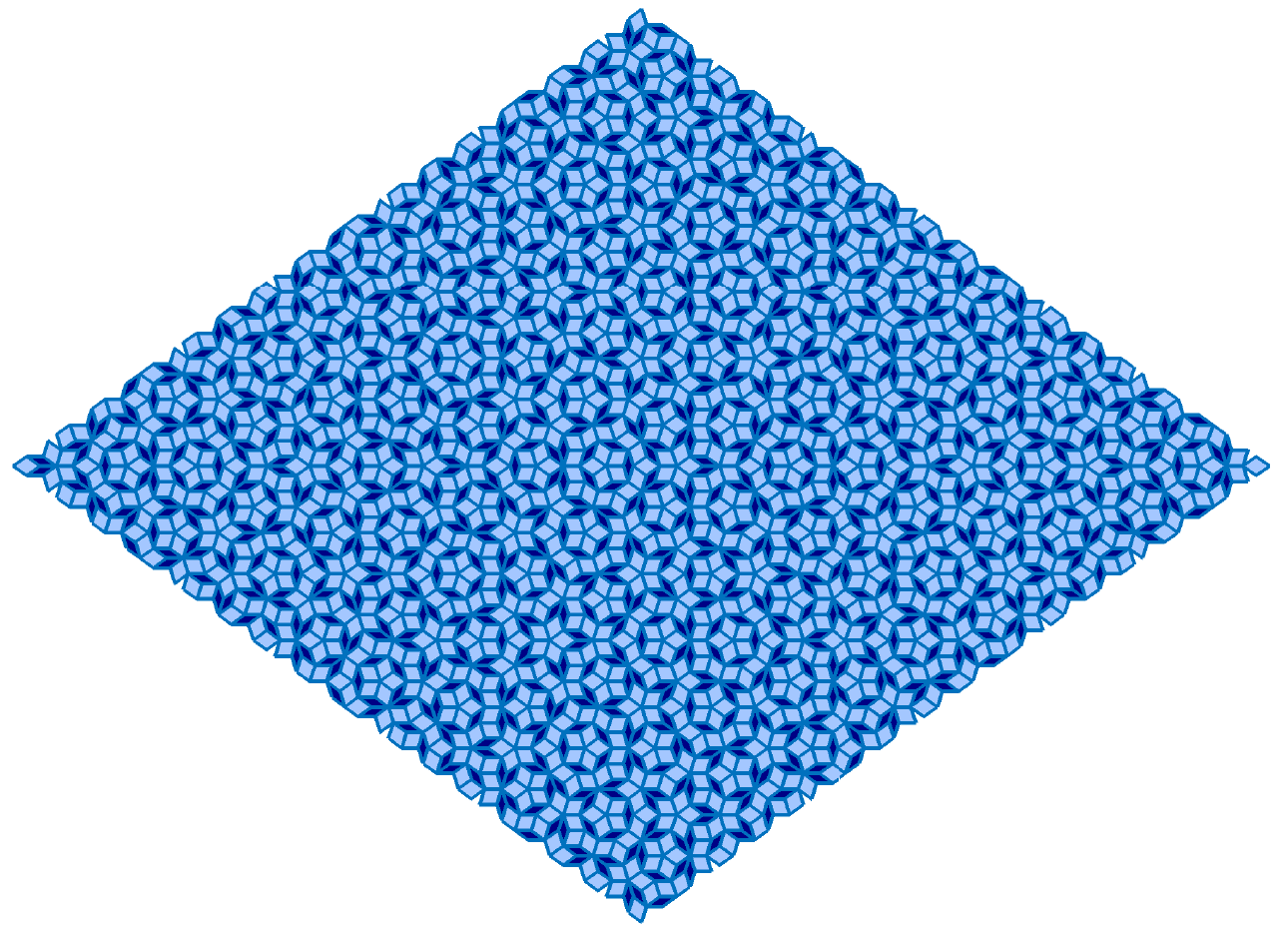}
	}
	\quad
	\subfigure[]{
		\includegraphics[width=7.1cm]{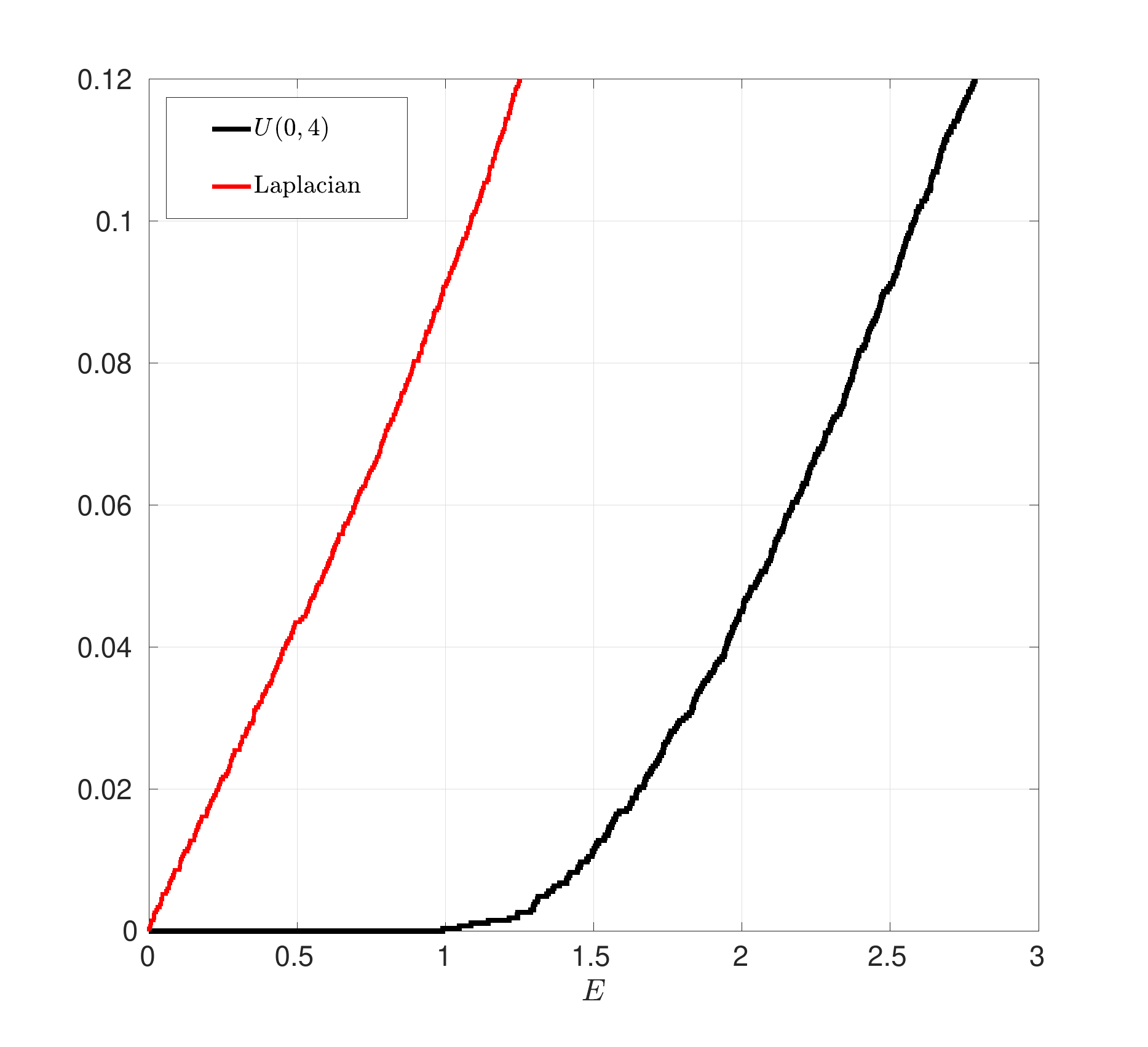}
	}
        \subfigure[]{
		\includegraphics[width=7.1cm]{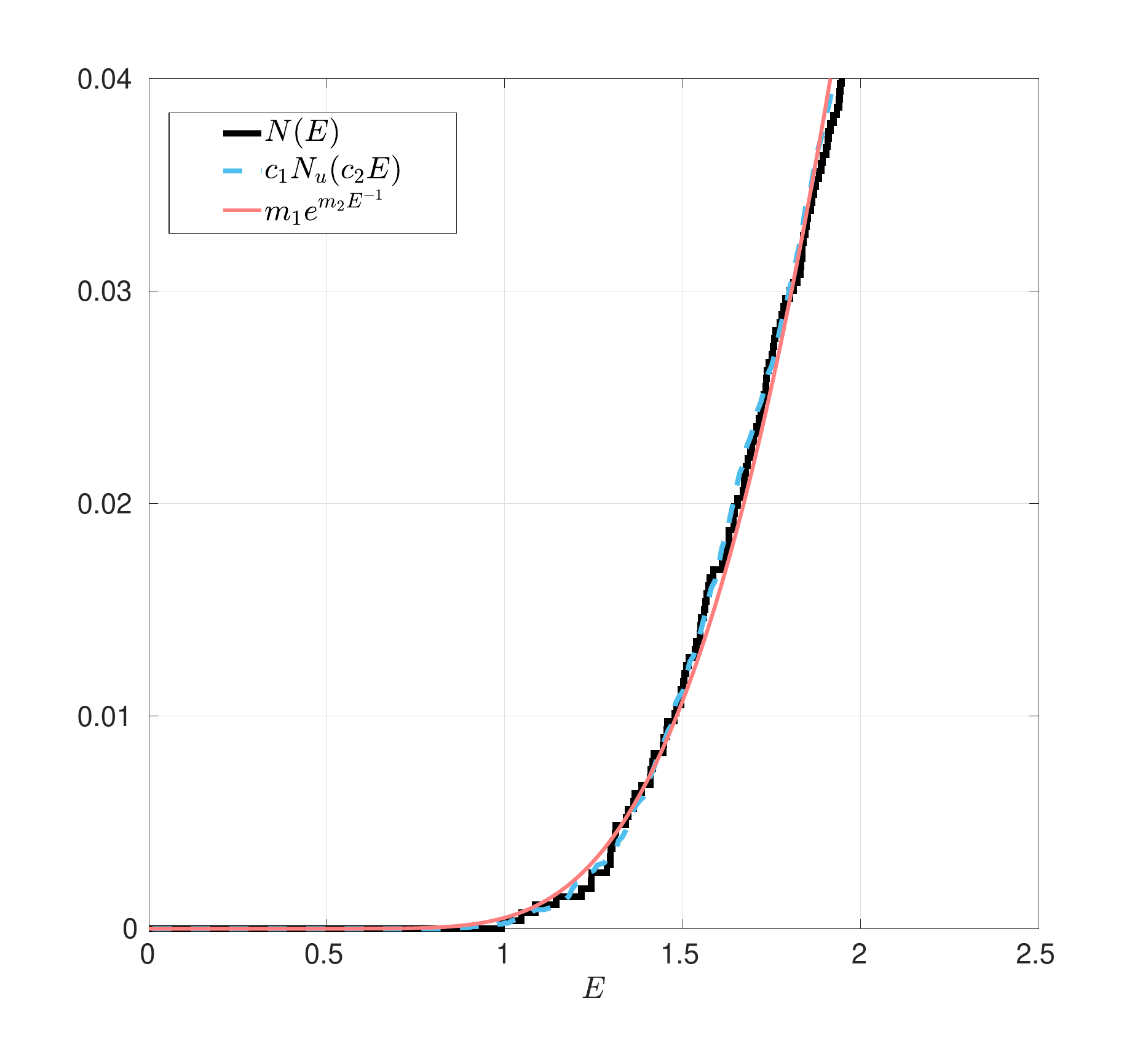}
	}
	\caption{ (a) The Penrose tiling used for calculation. (b) The IDS: the uniform disorder case (black) and the free Laplacian case (red). (c) Enlarged view of  the Lifshitz tail behavior of the IDS with its exponential fitting ($m_1=4.7 , m_2=-9.1$) and scaled landscape counting function ($c_1=0.14, c_2=\frac{1}{1.3}$). }
    \label{PT}
\end{figure}

\newpage 
\appendix

\section{Green's function and Poisson kernel for a Dirichlet Laplacian on a ball}\label{sec:poisson}
We summarize some facts about Green's function and Poisson kernel on graphs. These are standard results and can be found in e.g. \cite{barlow2017random, lawler2010random}. 

Let $\Gamma=(\V,\mathcal E)$ be a graph, equipped with some distance function $d(x,y):\V\times \V \to \R_{\ge 0}$. Note that this distance function and the results in this section are not limited to the natural graph metric (shortest-path distance).    
The graph Laplacian is
\[\Delta f(x)=\sum_{y\in \V:y\sim x}\big(f(y)-f(x)\big). \]
For a ball $B_r(\xi)=\{y\in \V:d(x,y)\le r\}$, the Dirichlet Laplacian on $B_r(\xi)$ is: 
\[(\Delta^{B_r}f)(x)=\sum_{y\in B_r(\xi):y\sim x}\big(f(y)-f(x)\big), \ \ x\in B_r(\xi). \]
The   exterior  boundary of $B_r$ is   $\partial B_r=\{ x\not\in B_r: \exists y \in B_r \ \ {\rm with}\ \  x\sim y \}, $
 and the interior boundary is defined as $\partial^i B_r=\partial(\V\backslash B_r)  $.  Denote by $\bar  B_r=B_r\cup \partial B_r$ the discrete closure of $B_r$. Let $G_{ r}(x,y)=G_{B_r}(x,y)=(-\Delta^{B_r})^{-1}(x,y):\, \bar  B_r\times \bar  B_r\to [0,1]$ be the Green's function associated with $-\Delta^{B_r}$. And denote by $P_{B_r}(x,y):\bar  B_r\times \partial B_r\to [0,1]$ the associated Poisson kernel. Recall that $G_{B_r}$ and $P_{B_r}$ are the unique solutions to
the following systems: for any $y\in B_r$, 
\begin{align*} 
    \begin{cases}
        \Delta  G_{B_r}(x,y)=\delta_y(x), & \   x \in B_r \\
        G_{B_r}(x,y)=0, & \ x\in \partial B_r
    \end{cases},
\end{align*}
and for any $y\in \partial B_r$, 
\begin{align}\label{eqn:Poisson-solution}
    \begin{cases}
        \Delta  P_{B_r}(x,y)=0, & \   x \in B_r \\
        P_{B_r}(x,y)=\delta_y(x), & \ x\in \partial B_r
    \end{cases}. 
\end{align} 

One can verify the following integration by parts formula for $g$ supported on $B_r=B(\xi,r)$ and any~$f$,
\begin{align*} 
\sum_{w\in B(\xi,r)}-\Delta g(w)f(w) &= \sum_{w\in B(\xi,r)}-\Delta f(w)g(w)+\sum_{w\in\partial^{(i)}B(\xi,r)}
g(w)\sum_{\substack{y:y\sim w\\ y\not\in B(\xi,r)}}f(y).
\end{align*} 

By taking $g=G_{B_r}(\xi,\cdot)$, then we see
\begin{align}\label{eqn:ibp-green}
f(\xi)=&\sum_{w\in B_r(\xi)}-\Delta f(w)G_{B_r}(\xi,w)+\sum_{w\in\partial^{i}B_r(\xi)}
G_{B_r}(\xi,w)\sum_{\substack{y:y\sim w\\ y\not\in B_r(\xi)}}f(y) \\
=& \sum_{w\in B_r(\xi)}-\Delta f(w)G_{B_r}(\xi,w)+\sum_{y\in\partial B_r(\xi)} 
P_{B_r}(\xi,y) f(y), \label{eqn:ibp-green2}
\end{align}
where the second line follows from  the relation between the Green's function and the associated Poisson's kernel
\begin{align}\label{eqn:Pr-Gr}
   P_{B_r}(\xi,x) = \sum_{y\in B_r:y\sim x}G_{B_r}(\xi,y), \  x\in \partial B_r .  
\end{align}

Taking $f\equiv1$ in \eqref{eqn:ibp-green2} also shows
\begin{align*}
\sum_{y\in\partial B_r(\xi)}P_{B_r}(\xi,y)=1.  
\end{align*}
If $f$ is subharmonic ($-\Delta f\le 0$), then \eqref{eqn:ibp-green2} implies the surface submean property 
\begin{align}\label{eqn:submean-surface}
   f(\xi)\le \sum_{y\in\partial B_r(\xi)} 
P_{B_r}(\xi,y) f(y), 
\end{align}
and the equality holds if $\Delta f=0$.

\section{An explicit proof of Poisson kernel estimates for the 1D random band model}
To obtain the general Lifshitz tails lower bound, one needs the control on the   Poisson kernel   $P_r=P_{B_r}$ as in Proposition~\ref{prop:lawler}, which provides the harmonic weight as required by Assumption~\ref{ass:harmonic-weight}. One notable example where we have such estimates is the graph $\Gamma_{d,W}$ \eqref{eqn:Gamma-dW} induced by the band model. 
 
The above general case is obtained in \cite{lawler2010random} by the method of random walk. Below, we give a direct proof of the case $d=1,W\ge1$ with explicit constants depending on the band width $W$.  Recall on $\Gamma_{1,W}=(\Z, \mathcal E_W)$, $\mathcal E_W=\{\sim: x\sim y \ {\rm if}\ |x-y|\le W\}$. The associated graph Laplacian is
\[\Delta_W f(x)=\sum_{y\in \Z: |x-y|\le W}\big(f(y)-f(x)\big). \] 
In this part, we use the natural graph metric (shortest path) $d_0(x,y)=\lceil W^{-1}|x-y| \rceil$. Then the ball centered at $\xi$ of radius $r$ is $B_r(\xi)=B(\xi,r)=\{x\in \Z:|x-\xi|\le rW\} $, with exterior boundary $\partial B_r =\{x:rW<|x-\xi|\le (r+1)W\}$, and interior boundary $\partial^i B_r =\{x:(r-1)W<|x-\xi|\le rW\}$. (For 1D, it is enough to consider integer valued radius  $r$ only.)
We denote by $G_r=G_{B_r}$ and $P_r=P_{B_r}$ the Green's function and the Poisson's kernel of $-\Delta^{B_r}_W$ on $B_r$, respectively. 
\begin{lemma}\label{lem:Pr-Gr-bound}
For all $x\in B_r(\xi)$, 
    \begin{align}
 \frac{1}{2W^3} ( rW+1-|x-\xi|)   \le    G_r(\xi,x)\le \frac{1}{W^2}( W+rW-|x-\xi|)  . \label{eqn:Gr-upper-lower}
    \end{align}
    For all $x\in \partial B_r(\xi)$,
    \begin{align}\label{eqn:Pr-upper-lower}
      \frac{1}{2W^3}\le P_r(\xi,x)\le   1.
    \end{align}
\end{lemma}

\begin{proof}[Proof of \eqref{eqn:Gr-upper-lower}]

Without loss of generality, we assume $\xi=0$ and only consider $B_r=B(0,r)$. The following estimates hold for any center $\xi\in A_K$ by translation. It is clear that $B(0,r)=\intbr{-L}{L}:$ with $L=rW$, and $\partial B_r=\partial ^lB_r+\partial ^rB_r$ with $\partial ^rB_r=\intbr{rW+1}{(r+1)W}$ and $\partial ^lB_r=\intbr{-(r+1)W}{-rW-1}$.  

Notice that $G_r(x,y)=G_r(y,x)$. Let $g(x)=G_r(x,0),|x|\le L$ be the centric (0th) column of $G_r$. Since $G_r$ is the inverse of $-\Delta_W^{B_r}$, then 
$-\Delta_W^{B_r} g(x)=\delta_0(x)$. To estimate $g(x)$ from above, we choose a test function 
\begin{align*}
    f(x)=\begin{cases}
   W+L-|x|,\ \ &|x|\le L\\
    0, \ \ &|x|> L
    \end{cases}
\end{align*}
$f(x)=W+L-|x|$ for any $x$.
Direct computation shows that 
\begin{align*}
    -\Delta_W^{B_r} f(x)\begin{cases}
        =W^2+W, \ \ &x=0,\\
        \ge 0, \ \ & 0<|x|\le L-W,\\
        =0,  \ \ & L-W< |x|\le L.
    \end{cases}
\end{align*}
Hence, $-\Delta_W^{B_r} f(x)\ge W^2\delta_{0}(x)=-\Delta_W^{B_r} \big(W^2g(x)\big)$ for all $x\in B_r$.
By the maximum principle of $-\Delta_W^{B_r}$, one obtains $g(x)\le \frac{1}{W^2}f(x)$ for all $|x|\le L$.
In particular, 
for $x \in \partial ^iB_r=\{x:|x|\ge L- W\}$, 
\begin{align}\label{eqn:Gr-upper}
   G_r(0,x) =G_r(x,0)\le \frac{1}{W^2}f(x)=\frac{1}{W^2}( W+L-|x|) \le \frac{2}{W}.
\end{align}

Next, we get a lower bound of $g(x)=G_r(0,x)$ through the    test function 
\begin{align*}
    p(x)=\begin{cases}
    L+1+W(W-1),\ \ & x=0\\
    L-|x|+1, \ \ & 0<|x|\le L \\
    0, \ \ & |x|>L
    \end{cases}.
\end{align*}
Direct computation shows
\begin{align*}
   -\Delta_W^{B_r} p(x)=\begin{cases}
    =W(2W^2-W+1) ,\ \ & x=0, \\
    \le 0, \ \ & 0<|x|<W, \\
    =0, \ \ &W \le x\le L,
   \end{cases} 
\end{align*}
leading to $-\Delta_W^{B_r} p(x)\le 2W^3 \delta_0(x)=-\Delta_W^{B_r}\big( 2W^3g(x)\big)$ for all $x\in B_r$. By the maximum principle again,
$g(x)\ge p(x)/(2W^3) $ for all $|x|\le L$. In particular, for $x\in \partial^i B_r=\{x: L-W\le |x|\le L\}$, one has 
 \begin{align*}
   G_r(0,x) =G_r(x,0)=  g(x)\ge ( L-|x|+1)/(2W^3)\ge 1/(2W^3) .
 \end{align*} 
Together with \eqref{eqn:Gr-upper}, we have that for  $x\in \partial^i B_r=\{x: L-W\le |x|\le L\}$, 
\begin{align}\label{eqn:255}
  \frac{1}{2W^3} \le  G_r(0,x) \le \frac{1}{W}.
\end{align}
Recall the relation between $P_r$ and $G_r$ in \eqref{eqn:Pr-Gr},
\[P_r(0,x)=\sum_{y\in B_r:y\sim x}G_r(0,y),\ \ \ x\in \partial B_r.  \]
Clearly, for  $x\in \partial B_r $,  $\varnothing\neq \{y\in B_r:y\sim x  \}\subset \partial^i B_r$.
As a consequence of \eqref{eqn:255},
\[\frac{1}{2W^3}\le \sum_{y\in B_r:y\sim x}G_r(0,y)= P_r(0,x)\le \sum_{y\in \partial^i B_r}G_r(0,y)\le W/W=1, \]
which completes the proof of \eqref{eqn:Pr-upper-lower}. 
\end{proof}

Notice in this 1D model with the natural metric $d_0$, $\partial  B_{\rho-1}=\partial^i B_\rho=B_{\rho }\backslash B_{\rho-1}$ for all $\rho$. To construct the harmonic weight, we do not need the   filtration lemma Claim \ref{claim:layer}. One can define directly
\[h_{B_r}(\xi,y)=\frac{1}{|B_r|}\begin{cases}
|\partial B_{\rho}|P_{\rho}(\xi,y),\ \  & y\in \partial  B_{\rho}, \rho=0,\cdots, r-1 \\ 
    1, \ \ & y=\xi .
\end{cases}  \]
 
Since $|\partial B_{\rho}|=2W$, then \eqref{eqn:Pr-upper-lower} implies $h_{B_r}(\xi,y)\ge \frac{1}{|B_r|W^2} $  for all $y\in B_r$, 
which gives an explicit bound required by Assumption~\ref{ass:harmonic-weight}.

\section{Moser--Harnack inequality for subharmonic functions }\label{sec:Moser-Harnack}
We say a graph $\Gamma$ satisfies an elliptic Harnack inequality (EHI), given $x_0\in \V$,$\tau>1$, and $R\ge 2/(\tau -1)$, if there exists a constant $C_H$ depending only on $\tau$ and $\Gamma$, such that for $h\ge 0$ on $B(x_0,\tau R)$ and harmonic ($\Delta h=0$) in $B(x_0,\tau R)$, then
\begin{align} \label{eqn:EHI}
    \sup _{B(x_0,R)}h \le C_H \inf _{B(x_0,R)}h .
\end{align}
Under the volume control assumption \eqref{eqn:vol-control}, EHI is equivalent to the (weak) Poincar\'e inequality \eqref{eqn:WPI} or the Gaussian heat kernel estimates $\mathrm{HKC}(\alpha,2)$ \eqref{eqn:gaussian-heatkernel}; for more discussion see the textbook \cite{barlow2017random}, particularly Theorems 6.19 and 7.18, and Lemma~4.21. 

  Clearly, (EHI) implies the Moser--Harnack inequality  
  \begin{align}\label{eqn:Moser-Harnack-harmonic}
        \sup_{y\in B(x,R)} h(y)^2 \le \frac{C_{H}}{|B(x,CR)|}\sum_{y\in B(x,CR)} h(y)^2,
    \end{align}
  for a positive harmonic function $h$ and any scaling constant $C\ge1$.  The version \eqref{eqn:Moser-Harnack} that we need for a subharmonic function $-\Delta f\le 0$ essentially follows from the harmonic version, combined with the  discrete Caccioppoli (``reverse Poincar\'e'') inequality and Poincar\'e inequality with Dirichlet boundary conditions. The authors in \cite{lin2018harnack} proved a Moser--Harnack inequality for subharmonic  functions on general graphs without the volume control assumptions \eqref{eqn:vol-control}, where the Moser--Harnack constant depends exponentially large on the radius of the ball.  
We did not find a radius-independent version of \cite[{Theorem~1.2}]{lin2018harnack} in the literature, and so we sketch the proof of \eqref{eqn:Moser-Harnack} here for completeness.

  Throughout, we write $B_R=B(x_0,R)$ for a ball centered at $x_0\in \V$ of radius $R$. 
Let $S\subset \V$ be a finite set. Define the Dirichlet energy subject to the boundary condition on $\partial S$ to be
\begin{align}\label{eqn:energyI}
    I(f,S):=\sum_{x\in S}\sum_{ \substack{y\in \V\\
    y\sim x}}\big(f(y)-f(x)\big)^2.
\end{align}
Note that in some literature, the concept `Dirichlet form' may refer to the energy with the zero boundary condition on $\partial S$, see e.g. \cite[\S 1.4]{barlow2017random}, that is:
 \[\mathcal E _S(f,f)= \sum_{x\in S}\sum_{\substack{y\in S\\
 y\sim x}}\big(f(y)-f(x)\big)^2 \le I(f,S). \]
The equality holds only if $f=0$ on the exterior boundary $\partial S$.

 The following is the discrete version of the well-known energy minimizing property of a harmonic function. 
\begin{lemma}[Energy minimizer, {\cite[Theorem~3.5]{holopainen1997p}}]
    If $h$ is harmonic on $\bar S=S\cup \partial S$, then $h$ is a minimizer of $I(f,s)$
    among functions in $\bar S$ with the same values on $\partial S$. More precisely, if $\Delta h=0,x\in S$,  then for any $f$ satisfying $f=h$ on $\partial S$,  
   \begin{align}\label{eqn:h-mini-I}
     I(h,S)\le I(f,S).
    \end{align}

\end{lemma}
  
\begin{lemma}[discrete Caccioppoli inequality, see e.g. {\cite[Lemma~2.4]{lin2018harnack}}]
    Let $x_0\in \V$, $R\ge 1$ and $\tau<2 $. Suppose $f$ is a nonnegative and subharmonic function on $B_{2R}=B(x_0,2 R)$. Then
we have 
\begin{align}\label{eqn:Cacciopp}
   I(f, B_{\tau R}) \le \frac{C}{R^2}\sum_{x\in B_{2R}} f (x)^2,
\end{align}
where $C$ depends only on $\tau$ and $M_\Gamma=\sup_{x\in\V}\deg(x)$ from \eqref{eqn:bded-geo}. 
\end{lemma}
This follows the proof of the Caccioppoli inequality for functions on $\R^d$, see e.g. \cite[\S4.1]{regularity}.
\begin{proof}
 Let $0\le \chi (x)\le 1$ the cut-off function constructed (in a similar manner) as in \eqref{eqn:cut-off-2}, satisfying $\chi (x)=1$ on $B(x_0,\tau R)$, $\chi (x)=0$ for $x\notin  B(x_0,2R-1)$, and $|\chi (x)-\chi (y)|\le \frac{c_\tau}{R}$ for all $x\sim y$. 

Since $-\Delta f(x)\le0$ for $x\in B_{2R}$, then for the test function $\varphi(x)=f(x)\chi^2(x)\ge0$, the discrete Gauss–Green Theorem (integration by parts, see e.g. \cite[Theorem~1.24]{barlow2017random}) implies
\begin{align}\label{eqn:ibp-rp}
0\ge \langle \varphi,-\Delta f\rangle 
&=\frac{1}{2}\sum_{\substack{x,y\in B_{2R} \\ x\sim y } }\chi^2(x)\big(f(x)-f(y)\big)^2
+\frac{1}{2}\sum_{\substack{x,y\in B_{2R} \\ x\sim y } }f(y)\big(\chi^2(x)-\chi^2(y)\big)\big(f(x)-f(y)\big),
\end{align}
where we used that 
$
\varphi(x)-\varphi(y)=(f(x)-f(y))\chi^2(x)+f(y)(\chi^2(x)-\chi^2(y)).
$
Therefore \eqref{eqn:ibp-rp}, with Cauchy--Schwarz and properties of $\chi$, implies
\begin{align*}
\sum_{\substack{x,y\in B_{2R} \\ x\sim y } }   \chi^2(x)&\big(f(x)-f(y)\big)^2 \\ 
&\le \sum_{\substack{x,y\in B_{2R} \\ x\sim y } }(\chi(x)+\chi(y))|f(x)-f(y)|f(y)|\chi(x)-\chi(y)|\\
&\le 2\bigg(\sum_{\substack{x,y\in B_{2R} \\ x\sim y } }\chi^2(x)(f(x)-f(y))^2\bigg)^{1/2}\bigg(\sum_{\substack{x,y\in B_{2R} \\ x\sim y } }f(y)^2(\chi(x)-\chi(y))^2\bigg)^{1/2}\\
&\le 2\bigg(\sum_{\substack{x,y\in B_{2R} \\ x\sim y } }\chi^2(x)(f(x)-f(y))^2\bigg)^{1/2}\frac{c_\tau}{R}\bigg(\sum_{y\in B_{2R}}f(y)^2\bigg)^{1/2}M_\Gamma^{1/2}.
\end{align*}
Dividing both sides by the square root of the left hand side, 
and using that $\chi (x)=1$ on $B(x_0,\tau R)$, yields
\begin{align}
\sum_{x\in B_{\tau R}}\sum_{y:y\sim x}\big(f(x)-f(y)\big)^2\le \sum_{\substack{x,y\in B_{2R} \\ x\sim y } }\chi^2(x)\big(f(x)-f(y)\big)^2&\le \frac{4c_\tau^2M_\Gamma}{R^2}\sum_{y\in B_{2R}}f(y)^2.
\end{align}
\end{proof}

Next, we will need that the weak (or ``mean'') Poincar\'e inequality \eqref{eqn:WPI} implies the following zero-boundary Poincar\'e inequality.
We do not actually directly use the weak PI formulation \eqref{eqn:WPI}, but just that the exit time bound in Proposition~\ref{prop:piconseq}(ii) holds.
\begin{lemma}[Poincar\'e inequality with the Dirichlet boundary condition]
Suppose volume control \eqref{eqn:vol-control} and the weak Poincar\'e inequality \eqref{eqn:WPI} hold.
Let $x_0\in \V$ and $R>0$. For $f$ vanishing on the exterior boundary $\partial B(x_0,R)$, then one has 
      \begin{align}\label{eqn:PI-zero}
          \sum_{x\in B(x_0,R)}f(x)^2\le  {C}{R^2}\sum_{x \in B(x_0,R)}  \sum_{y:y\sim x}\big(f(y)-f(x)\big)^2 ,
      \end{align}
      where $C$ depends only on $\Gamma$. 
 \end{lemma}
\begin{proof}
 Let $B_R=B(x_0,R)$ and $\Delta^{B_R}$ be the Dirichlet Laplacian on $B_R$. Let $E_1$ be the smallest eigenvalue of $-\Delta^{B_R}$. 
 By Proposition~\ref{prop:piconseq}(ii), there is the bound $\|(-\Delta_{B_R})^{-1}\one_{B_R}\|_\infty\le cR^2$.
  Since in general $E_1\|(-\Delta_{B_R})^{-1}\one_{B_R}\|_\infty\ge1$ (see e.g. \cite[Lemma~2.1]{lratio}),
 then for some $C$ depending only on $\Gamma$, 
 \[E_1\ge \frac{C}{R^2}. \]
 Notice that when $f=0$ on $\partial B_R$, the Dirichlet energy $I(f,B_R)$ in \eqref{eqn:energyI} equals the Dirichlet energy  $\mathcal E _S(f,f)$  subject to the zero boundary condition:
 \[  I(f,B_R)\equiv\sum_{x\in B_R}\sum_{ \substack{y\in \V\\
    y\sim x}}\big(f(y)-f(x)\big)^2 
= \sum_{x\in B_R}\sum_{\substack{y\in B_R\\
 x\sim y}}\big(f(y)-f(x)\big)^2. \]
 Then \eqref{eqn:PI-zero} follows directly from the definition of the ground state energy $E_1$ by minimizing the zero-boundary Dirichlet energy (among functions not identically zero) on the right hand side. 
\end{proof}

\begin{proof}[Proof of \eqref{eqn:Moser-Harnack}]
    The idea follows from the outline of the proof of \cite[Theorem~1.2]{lin2018harnack}. 
     However, we must consider the (in)dependence of the radius more carefully using \eqref{eqn:Cacciopp} and \eqref{eqn:PI-zero}. Let $f\ge 0$ be a subharmonic function on $\bar B_{2R}=\bar B(x_0,2R)$. Let $h$ be a harmonic function on on $\bar B_{3R/2}=\bar B(x_0,3R/2)$ which agrees $f$ on the (exterior) boundary of $B_{3R/2}$, i.e., $h$ satisfies $\Delta h=0,x\in B_{3R/2}$ and $h(x)=f(x)$ for $x\in \partial B_{3R/2}$. By the maximum principle, $f\le h$ on $B_{3R/2}$. Then the Moser--Harnack property \eqref{eqn:Moser-Harnack-harmonic} (with $\tau=3/2$) for $h$ implies
    \begin{align}
        \sup_{B_R}f^2\le \sup_{B_R}h^2\le 
        \frac{C_H}{|B_{\tau R}|}\sum_{B_{\tau R}}h^2 \le  \frac{2C_H}{|B_{\tau R}|}\sum_{B_{\tau R}}(h-f)^2+\frac{2C_H}{|B_{\tau R}|}\sum_{B_{\tau R}}f^2. \label{eqn:C10}
    \end{align}
   The second term is readily bounded by a constant times the average of $f^2$ on $B_{2R}$ using the volume control \eqref{eqn:vol-control}. 
   For the first sum, since $w:=h-f=0$ on $\partial B_{\tau R}$, by the Poincar\'e inequality \eqref{eqn:PI-zero}, 
   \begin{align*}
           \sum_{x\in B_{\tau R}}w(x)^2\le &  {C}{R^2}\sum_{x \in B_{\tau R}}  \sum_{y\sim x}\big(w(y)-w(x)\big)^2 \\
    \le &     2{C}{R^2}\sum_{x \in B_{\tau R}}  \sum_{y\sim x}\big(h(y)-h(x)\big)^2 + 2{C}{R^2}\sum_{x \in B_{\tau R}}  \sum_{y\sim x}\big(f(y)-f(x)\big)^2 \\
    \le &     4{C}{R^2}\sum_{x \in B_{\tau R}}  \sum_{y\sim x}\big(f(y)-f(x)\big)^2,\numberthis \label{eqn:C13}
   \end{align*}
where in the last line we used \eqref{eqn:h-mini-I}, that the harmonic function $h$ is a minimizer of the energy $I(B_{\tau R},f)$. 
Then the Caccioppoli inequality \eqref{eqn:Cacciopp} implies the sum in \eqref{eqn:C13} is bounded as
\[\sum_{x \in B_{\tau R}}  \sum_{y\sim x}\big(f(y)-f(x)\big)^2\le \frac{C}{R^2}\sum_{B_{2R}}f^2, \]
leading to $\sum_{x\in B_{\tau R}}w(x)^2\le C'\sum_{B_{2R}}f^2  $
 for some constant $C'$ independent of $R$. 
 Using this in \eqref{eqn:C10}, we obtain
\begin{align*}
     \sup_{B_R}f^2\le \frac{2C_H}{|B_{\tau R}|}\sum_{B_{\tau R}}w^2+\frac{2C_H}{|B_{\tau R}|}\sum_{B_{\tau R}}f^2\le \frac{C}{|B_{2R}|}\sum_{B_{2R}}f^2.
\end{align*}
   
\end{proof}

\noindent\textbf{Acknowledgments.}
\phantomsection
\addcontentsline{toc}{section}{Acknowledgments}
The authors would like to thank Douglas N. Arnold and Svitlana Mayboroda for useful discussions.
L.S. was supported in part by the Julian Schwinger Foundation for Physics Research and Simons Foundation grant 563916, SM. 
\bibliographystyle{abbrv}
\bibliography{landscape.bib}

 {
  \bigskip
  \vskip 0.08in \noindent --------------------------------------

\footnotesize
\medskip
L.~Shou, {Condensed Matter Theory Center and Joint Quantum Institute, Department of Physics, University of Maryland, College Park, MD 20742, USA}\par\nopagebreak
    \textit{E-mail address}:  \href{mailto:lshou@umd.edu}{lshou@umd.edu}

\vskip 0.4cm

  W. ~Wang, {LSEC, Institute of Computational Mathematics and Scientific/Engineering Computing, Academy of Mathematics and Systems Science, Chinese Academy of Sciences, Beijing 100190, China}\par\nopagebreak
  \textit{E-mail address}: \href{mailto:ww@lsec.cc.ac.cn}{ww@lsec.cc.ac.cn}
  
\vskip 0.4cm

S.~Zhang, {Department of Mathematics and Statistics, University of Massachusetts Lowell, 
Southwick Hall, 
11 University Ave.
Lowell, MA 01854
 }\par\nopagebreak
  \textit{E-mail address}: \href{mailto:shiwen\_zhang@uml.edu}{shiwen\_zhang@uml.edu}
}

\end{document}